%% file: main.tex
\documentclass[10pt]{article}
\usepackage[utf8]{inputenc}
\usepackage[T1]{fontenc}
\usepackage{ifpdf}
\usepackage{graphicx}
\usepackage[cmex10]{amsmath}
\usepackage{amssymb,amsfonts}
\usepackage{amsthm}
\usepackage[usenames,dvipsnames,svgnames,table]{xcolor}
\usepackage{tikz}
\usepackage[autolanguage]{numprint}
\usepackage{etex}
\usepackage{lmodern}

\usepackage{array}

\usepackage{mdwmath}
\usepackage{booktabs}
\usepackage{tabularx}

\usepackage[caption=false,font=footnotesize]{subfig}

\usepackage{fixltx2e}

\usepackage{url}
\usepackage{xspace}

\usepackage{microtype}
\usepackage{pgfplots}
\usepackage{pgfplotstable}
\usepackage{multirow}
\usepackage{enumitem}

\usepackage[twoside,
  paperwidth=210mm,
  paperheight=297mm,
  textheight=562pt,
  textwidth=384pt,
  centering,
  headheight=50pt,
  headsep=12pt,
  footskip=12pt,
  footnotesep=24pt plus 2pt minus 12pt]{geometry}

\usepackage{algorithmicx,algpseudocode,algorithm}

\newcounter{dummy} 
\newtheorem{theorem}[dummy]{Theorem}
\newtheorem{proposition}[dummy]{Proposition}
\newtheorem{lemma}[dummy]{Lemma}
\newtheorem{corollary}[dummy]{Corollary}
\newtheorem*{conjecture}{Conjecture}

\theoremstyle{definition}
\newtheorem{definition}[dummy]{Definition}

\theoremstyle{remark}

\usetikzlibrary{arrows,shapes,positioning,calc,patterns,spy,decorations.pathmorphing,decorations}

\pgfplotstableset{col sep=comma}
\pgfplotsset{
    /pgfplots/ybar legend/.style={
        /pgfplots/legend image code/.code={%
        \draw[##1,/tikz/.cd,bar width=6pt,yshift=-0.15em,bar shift=0pt]
        plot coordinates {(0cm,6pt) };}
    }}
\newcommand{\rank}{\mathsf{rank}}
\newcommand{\height}{\mathsf{height}}
\newcommand{\sk}{\mathsf{sk}}

\newcommand{\att}{\mathsf{att}}
\newcommand{\lbl}{\mathsf{lab}}
\newcommand{\ext}{\mathsf{ext}}
\newcommand{\N}{\mathbb{N}}
\newcommand{\HGR}{\mathsf{HGR}}
\newcommand{\lhs}{\mathsf{lhs}}
\newcommand{\rhs}{\mathsf{rhs}}
\newcommand{\sav}{\mathsf{con}}
\newcommand{\refs}{\mathsf{ref}}

\newcommand{\FPequiv}{\cong_{\text{FP}}}
\newcommand{\FP}{\ensuremath{{\text{FP}}}\xspace}
\newcommand{\NP}{\ensuremath{{\mathsf{NP}}}\xspace}
\newcommand{\Occ}{{\mathsf{Occ}}}

\newcommand{\neighbor}{N}
\newcommand{\inNeighbor}{\ensuremath{N^+}}
\newcommand{\outNeighbor}{\ensuremath{N^-}}
\newcommand{\incident}{E}

\newcommand{\HG}{\mathsf{handle}}

\newcommand{\degree}{\mathsf{deg}}
\newcommand{\ktree}{$k^2$-tree\xspace}
\newcommand{\ord}{\ensuremath{\omega}\xspace}
\newcommand{\graphrepair}{gRePair\xspace}
\newcommand{\SLOrder}{\leq_\mathsf{{NT}}}
\newcommand{\val}{\mathsf{val}}

\newcommand{\A}{\mathcal{A}}

\newcommand{\PSPACE}{\ensuremath{{\mathsf{PSPACE}}}\xspace}
\newcommand{\EXPTIME}{\ensuremath{{\mathsf{EXPTIME}}}\xspace}
\newcommand{\NL}{\ensuremath{{\mathsf{NL}}}\xspace}

\newcommand{\symb}{\mathsf{symb}}
\newcommand{\dist}{\mathsf{dist}}

\newcommand{\grid}{\mathsf{grid}}
\newcommand{\tf}{\mathsf{tf}}

\newcommand{\ls}{\mathsf{ls}}
\newcommand{\idx}{\mathsf{idx}}

\newcommand{\trans}{\mathsf{trans}}
\newcommand{\sgraph}{\mathsf{s\text{-}graph}}
\newcommand{\tgraph}{\mathsf{t\text{-}graph}}
\newcommand{\nr}{\rho}
\newcommand{\NTSet}{E^{\text{NT}}}
\newcommand{\sib}{\ensuremath{\mathsf{sib}}}
\newcommand{\dt}{\ensuremath{\mathsf{dt}}}
\newcommand{\head}{\ensuremath{\mathsf{head}}}
\renewcommand{\P}{\ensuremath{{\mathsf{P}}}\xspace}
\newcommand{\gRepRun}{\ensuremath{\nu}}

\newcommand{\tblFontSize}{\footnotesize}
\newcommand{\graphFontSize}{\footnotesize}

\newcolumntype{R}{>{\raggedleft\arraybackslash}X}
\definecolor{color1}{HTML}{D7191C}
\definecolor{color2}{HTML}{FDAE61}
\definecolor{color3}{HTML}{ABDDA4}
\definecolor{color4}{HTML}{2B83BA}

\definecolor{l1}{HTML}{E30004}
\definecolor{l2}{HTML}{E9830C}
\definecolor{l3}{HTML}{F0D90D}

\pgfdeclarelayer{background}
\pgfdeclarelayer{foreground}
\pgfsetlayers{background,main,foreground}

\tikzset{
    ncbar angle/.initial=90,
    ncbar/.style={
        to path=(\tikztostart)
        -- ($(\tikztostart)!#1!\pgfkeysvalueof{/tikz/ncbar angle}:(\tikztotarget)$)
        -- ($(\tikztotarget)!($(\tikztostart)!#1!\pgfkeysvalueof{/tikz/ncbar angle}:(\tikztotarget)$)!\pgfkeysvalueof{/tikz/ncbar angle}:(\tikztostart)$)
        -- (\tikztotarget)
    },
    ncbar/.default=0.5cm,
}

\tikzset{square left brace/.style={ncbar=0.2cm}}
\tikzset{square right brace/.style={ncbar=-0.2cm}}

\title{Grammar-Based Graph Compression}
\author{Sebastian Maneth \\ University of Edinburgh \\ smaneth@inf.ed.ac.uk \and Fabian Peternek \\
University of Edinburgh \\ f.peternek@ed.ac.uk}
\date{}

\begin{document}
\maketitle

\begin{abstract}
We present a new graph compressor 
that works by recursively detecting repeated substructures 
and representing them through grammar rules.
We show that 
%Our experiments show that 
for a large number of graphs
the compressor obtains smaller representations than other approaches. 
%For RDF graphs and version graphs it outperforms 
%the best known previous methods. 
Specific queries such as reachability between two nodes or regular path queries 
can be evaluated in linear time (or quadratic times, respectively), over the grammar, thus
allowing speed-ups proportional to the compression ratio.
%Executing general purpose graph algorithms over a compressed grammar is
%possible, but incurs a slow-down. 
%solved in one pass over a grammar; 
%this offers speed-ups proportional to the compression ratio.
\end{abstract}

\input{new_introduction}
\input{related_work}
\input{preliminaries}

\input{graphrepair}
\input{implementation}
\input{experimental_results}
\input{query_evaluation}
\input{conclusion}
\bibliographystyle{plain}
\bibliography{lit}

\end{document}

%% file: new_introduction.tex
\section{Introduction}

Graph databases were investigated already some 30~years ago as described by 
Wood~\cite{Wood12_querySurvey}.
Today, with linked data on the web and social network data, there has been a
resurgence of graph databases and graph processing systems. 
Compression is an important technique of dealing with large graph data:
it saves storage space and data transfer time~\cite{DBLP:conf/icdt/Fan12}.
%; Fan~\cite{DBLP:conf/icdt/Fan12} 
%lists compression as one of the essential techniques for dealing with large graphs. 
Grammar-based compression is a technique by which even query evaluation time can be saved.
%There is a particular breed of compression schemes which allows
%to even save query evaluation time: grammar-based compression.
The idea is to compute a small context-free grammar
generating a given object, e.g., a string, tree, or graph. 
Specific queries, e.g. queries performed by a finite-state automaton can be evaluated
in linear time (and one pass) over such a grammar, thus providing speed-ups 
proportional to the compression ratio.
Grammar-based compression has been known for strings and trees (see,
e.g.,~\cite{Lohrey12_stringSLPsurvey,DBLP:conf/dlt/Lohrey15}).
This paper investigates grammar-based graph compression. 
%More precisely, we deal with edge-labeled hypergraphs. 
Our main contributions are:
\begin{enumerate}
\item[(1)] we generalize to graphs the RePair compression scheme,
\item[(2)] we experimentally evaluate an implementation of the compression scheme, and
\item[(3)] we present new algorithms for query evaluation over compressed graphs.
\end{enumerate}

Let us discuss these contributions in more detail.
It is well-known that finding a smallest context-free grammar for a given string
is NP-complete~\cite{Charikar05_smallest}. Various approximation algorithms 
have been considered, see~\cite{Charikar05_smallest}. 
A particularly effective and simple such algorithm is 
RePair by Larsson and Moffat~\cite{Larsson00_rePair}.
The idea is to repeatedly replace a most-frequent digram in a string
by a new nonterminal, and to introduce a corresponding rule for the nonterminal.
A digram in a string consists of two consecutive symbols.
For instance, the string
\[
abcabcab
\]
has three occurrences of the digram $ab$,
two occurrences of the digram $bc$, 
and two occurrences of the digram $ca$.
Thus, RePair can replace the digram $ab$ by a new nonterminal, say $A$.
The resulting string $AcAcA$ has two occurrences of the digram $Ac$
which may be replaced by the new nonterminal $B$ to obtain this grammar:
\[
\begin{array}{lcl}
S&\to& BBA\\
B&\to& Ac\\
A&\to& ab
\end{array}
\]
Note that the size of the resulting grammar, i.e., the sum of lengths of the right-hand sides
is $7$, i.e., is by one smaller than the length of the original string.
RePair has been generalized to ranked, ordered trees by 
Lohrey, Mennicke, and Maneth~\cite{Lohrey13_treeRePair}.
For such trees, a digram consists of two nodes and an edge, i.e., a node and its $k$-th child
for some $k$. 
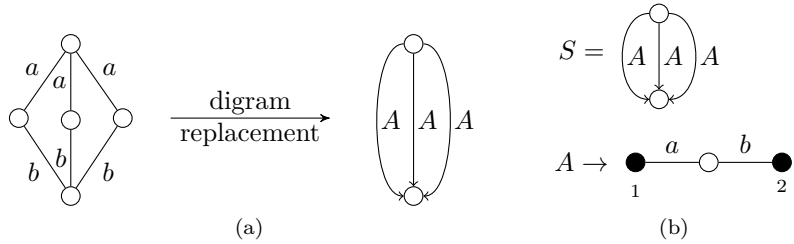
\begin{figure}[!t]
    \centering
    \subfloat[]{\input{figures/intro_replacement_example}\label{fig:intro_replacement_example}}\qquad
    \subfloat[]{\input{figures/intro_grammar_example}\label{fig:intro_grammar_example}}
%    \subfloat[]{\input{figures/intro_rank3_example}\label{fig:intro_rank3_example}}\qquad
%    \subfloat[]{\input{figures/hypergraph_example}\label{fig:hypergraph_example}}\qquad
%    \subfloat[]{\input{figures/intro_hGrammar_example}\label{fig:intro_hGrammar_example}}\qquad
    \caption{A digram replacement step~\protect\subref{fig:intro_replacement_example} and the
    grammar resulting from this replacement~\protect\subref{fig:intro_grammar_example}.}
    \label{fig:intro_example}
\end{figure}
Consider the edge-labeled graph on the left of Figure~\ref{fig:intro_replacement_example}.
Our idea is to define a digram as two edges which have least one node in common. 
The graph contains three occurrences of the digram consisting of an $a$- and
a $b$-edge. The graph also contains three occurrences of the digram
consisting of two $a$-edges (and the same for $b$), but these occurrences are overlapping so that
there is at most one non-overlapping occurrence of that digram.
If we replace each occurrence of the $a/b$-digram by a nonterminal edge labeled $A$, then we obtain
the graph shown on the right of Figure~\ref{fig:intro_replacement_example}.
The complete graph grammar for this graph is shown in Figure~\ref{fig:intro_grammar_example}.
%Figure~\ref{fig:intro_derivation_example} 
%shows how the original graph is derived from the grammar via three
%applications of the $A$-rule.
The size of the original graph (i.e., the sum of numbers of edges and nodes) is
$11$, while the size of the grammar is $10$. 

\begin{figure}[!t]
    \centering
    \subfloat[]{\input{figures/counting_example_center_intro}\label{fig:intro_counting_example_center}}
%    \subfloat[]{\input{figures/counting_example_dfs}\label{fig:counting_example_dfs}}
\qquad
    \subfloat[]{\input{figures/counting_example_opt}\label{fig:intro_counting_example_opt}}
%    \subfloat[]{\input{figures/counting_example_digram}\label{fig:counting_example_digram}}
    \caption{Two different traversals for counting digram occurrences.}
    \label{fig:intro_counting_example}
\end{figure}
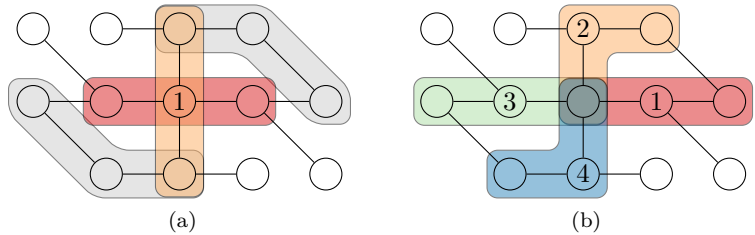
One of the major challenges in implementing RePair for graphs is that
we are unable to determine a most-frequent (non-overlapping) digram in linear time.
For strings and trees this is achieved by a trivial greedy left-to-right and
bottom-up counting procedure, respectively. Let us consider an example to see why 
in a graph 
a
greedy counting procedure does not find a maximal set of non-overlapping occurrences of a digram.
We follow a specific node order and greedily count occurrences of digrams with the current node as center node.
If we start with the node numbered ``1'' in Figure~\ref{fig:intro_counting_example_center}, 
then we count exactly two
non-overlapping occurrences of a digram (note that the gray shaded areas are occurrences of a \emph{different} digram).
If however, we follow the order of nodes shown in Figure~\ref{fig:intro_counting_example_opt}, 
then we find four occurrences of the same digram. 
The most efficient way we are aware of for finding a most-frequent (non-overlapping) digram
in a graph has quadric time complexity. This is prohibitively expensive for large graphs.
We therefore resort to a greedy counting principle which can be performed in linear time. 
We experiment with different node orders. Interestingly, a node order based on
the ``similarity'' of nodes inspired by the numbering from the Weisfeiler-Lehman isomorphism
test~\cite{WeisfeilerLehman68} achieves the best results in our experiments.

%describe experimental results 
We run experiments with a prototype implementation of our RePair for graphs compression
algorithm. Note that typically the algorithm ends up with a large graph that does
not contain repeated digrams. Thus, we need an efficient way
to represent this rest graph. We use the $k^2$-trees of
Brisaboa, Ladra, and Navarro~\cite{DBLP:journals/is/BrisaboaLN14}.
We compare our compressor against $k^2$-trees, 
the list-merge compressor (LM) of Grabowski and Bieniecki~\cite{DBLP:journals/dam/GrabowskiB14},
and the compressor of Hern{\'a}ndez and Navarro~\cite{DBLP:journals/kais/HernandezN14}.
The latter first applies a generalization of the dense substructure removal (DSR)
of Buehrer and Chellapilla~\cite{Buehrer08_scalableWebGraphCompression} and then uses $k^2$-trees.
We find that over network graphs (which have no edge labels), 
the combination of our new RePair algorithm
with DSR gives the best results for all graphs but two (where LM obtains slightly better compression).
On RDF graphs (which contain edge labels) we compare our compressor with
$k^2$-trees. Here our method consistently obtains better compression, sometimes by a factor of
several hundreds.

We finally investigate query evaluation over grammar-compressed graphs.
For strings and trees one basic technique is to run a finite-state automaton
in one pass over the grammar. 
Unfortunately, for graphs there is no well accepted notion of a 
finite-state automaton. Instead, counting monadic second-order logic (CMSO) 
is considered as the graph counterpart to regular languages.
It follows from well-known results
that evaluating a fixed CMSO formula over a grammar-compressed graph can be carried
out in (data complexity) time $O(\eta|G|)$ where $\eta$ is an upper bound on the time needed to
evaluate the formula over a right-hand side of the grammar~\cite{DBLP:journals/tcs/CourcelleM93}. 
Since this may be too expensive for large graphs,
we investigate new algorithms for particular queries.
Given a graph grammar $G$, 
(a) reachability queries (for nodes $s,t$ determine if there is a path from $s$ to $t$) 
can be evaluated in time $O(|G|)$, and 
(b) regular path queries (i.e., determine if there exists two nodes with a 
labeled path between them matching a regular expression $\alpha$) 
can be evaluated in time $O(|G||\alpha|)$.

This paper is based on a preliminary paper that was presented at ICDE 2016 (see\cite{DBLP:conf/icde/ManethP16}).
We have added several new results and new material with respect to
the ICDE 2016 paper. In particular have we greatly extended the Related Work section.
Section~\ref{ss2:StringTreeRelation} on tree- and string-graphs
is entirely new; the main result of that section (Theorem~\ref{thm:treeGeneratingToTree} and
Corollary~\ref{cor:stringGeneratingToString})
essentially shows that SL-HR grammars do
not allow stronger tree or string compression as do the existing
well-known straight-line tree and string grammars.
Section~\ref{ss2:maxRank} is entirely new; it shows that the maximal rank parameter
can heavily influence the compression behavior or our compressor.
Section~\ref{sss:formalism_choice} is entirely new; it discusses the choice of our grammar
formalism and compares it to other formalisms such as node replacement graph
grammars. The experimental section has been extended by new
experiments over synthetic graphs, which allow to show the influence of the individual
parameters of the compressor. The section on Query Evaluation has been
largely extended; for instance, an algorithm to traverse a grammar-represented
graph and a new section on regular path queries has been added.
%
%A preliminary version of a subset of these results was presented at ICDE~2016 
%(see~\cite{DBLP:conf/icde/ManethP16}). This paper has been extended with
%\begin{enumerate}
%    \item a more comprehensive discussion of related work,
%    \item a comparison with previously presented RePair compression of strings and trees,
%    \item a better explanation of the assignment of node-IDs to the grammar-compressed graph,
%    \item a discussion about the chosen grammar-formalism, in particular regarding the option of
%        node-replacing grammars instead of hyperedge replacement,
%    \item experimental results on synthetic graphs explicitly showing previously discussed effects
%        of certain parameters on the compression, and
%    \item a speed-up algorithm to answer RDF-queries.
%\end{enumerate}

%% file: figures/intro_replacement_example.tex
\begin{tikzpicture}
%    \node (S) {$S = $};
    \begin{scope}[every node/.style={circle, draw, inner sep=0pt, minimum size=7pt, node
        distance=50pt}]
        \node[] (So4) {};
        \node[below=of So4] (Su4) {};
        \node[below left=.8 and .5 of So4] (in41) {};
        \node[below=.75 of So4] (in42) {};
        \node[below right=.8 and .5 of So4] (in43) {};
    \end{scope}

    \begin{scope}%[every path/.style={->}]
        \draw (So4) -- node[above left=-2pt] {\normalsize $a$} (in41);
        \draw (in41) -- node[below left=-2pt] {\normalsize $b$} (Su4);
        \draw (So4) -- node[left=-2pt] {\normalsize $a$} (in42);
        \draw (in42) -- node[left=-2pt] {\normalsize $b$} (Su4);
        \draw (So4) -- node[above right=-2pt] {\normalsize $a$} (in43);
        \draw (in43) -- node[below right=-2pt] {\normalsize $b$} (Su4);
    \end{scope}
    \node[inner sep=0pt, minimum size=7pt, node distance=50pt] (X) {};

    \node[right=.5 of in43, align=center, rectangle] (dr) {digram\\replacement};
    \draw[-stealth'] (dr.west) -- (dr.east);
    
    \begin{scope}[every node/.style={circle, draw, inner sep=0pt, minimum size=7pt, node
        distance=50pt}]
%        \node[above right=.45 and .5 of S] (So1) {};
        \node[right=4.25 of So4] (So1) {};
        \node[below=of So1] (Su1) {};
    \end{scope}

    \begin{scope}[every path/.style={->}]
        \draw (So1) -- node[right=-2pt] {\normalsize $A$} (Su1);
        \draw (So1) .. controls +(-.6,0) and +(-.6,0) .. node[right=-2pt] {\normalsize $A$} (Su1);
        \draw (So1) .. controls +(.6,0) and +(.6,0) .. node[right=-2pt](X) {\normalsize $A$} (Su1);
    \end{scope}
\end{tikzpicture}

%% file: figures/intro_grammar_example.tex
\begin{tikzpicture}
    \node (S) {$S = $};
    \begin{scope}[every node/.style={circle, draw, inner sep=0pt, minimum size=7pt, node
        distance=25pt}]
        \node[above right=.15 and .5 of S] (So) {};
        \node[below=of So] (Su) {};
    \end{scope}

    \begin{scope}[every path/.style={->}]
        \draw (So) -- node[right=-2pt] {\normalsize $A$} (Su);
        \draw (So) .. controls +(-.6,0) and +(-.6,0) .. node[right=-2pt] {\normalsize $A$} (Su);
        \draw (So) .. controls +(.6,0) and +(.6,0) .. node[right=-2pt] {\normalsize $A$} (Su);
    \end{scope}

    \node[below=of S] (A) {$A \rightarrow$};
    \begin{scope}[every node/.style={circle, draw, inner sep=0pt, minimum size=7pt, node
        distance=20pt}]
        \node[right=.1 of A, fill=black] (A1) {};
        \node[right=of A1] (A2) {};
        \node[right=of A2, fill=black] (A3) {};
    \end{scope}
    \begin{scope}%[every path/.style={--}]
        \draw (A1) -- node[above] {\normalsize $a$} (A2);
        \draw (A2) -- node[above] {\normalsize $b$} (A3);
    \end{scope}
    \node[below=.05 of A1] (source) {\scriptsize 1};
    \node[below=-.01 of A3] (target) {\scriptsize 2};
\end{tikzpicture}

%% file: figures/counting_example_center_intro.tex
\begin{tikzpicture}
    \begin{scope}[every node/.style={circle, draw, inner sep=0pt, minimum size=12pt, node
        distance=15pt}]
        \node (c) {\normalsize 1};
        \node[left=of c] (l) {};
        \node[left=of l] (l2) {};
        \node[above=of l2] (l1) {};
        \node[above=of c] (a) {};
        \node[left=of a] (a1) {};
        \node[right=of a] (a2) {};
        \node[right=of c] (r) {};
        \node[right=of r] (r1) {};
        \node[below=of r1] (r2) {};
        \node[below=of c] (b) {};
        \node[left=of b] (b1) {};
        \node[right=of b] (b2) {};
    \end{scope}

    \begin{pgfonlayer}{background}
        \begin{scope}[every path/.style={rounded corners, opacity=.5}]
            \draw[fill=black!20] ($ (a.center) + (0,9pt) $) -- +(32pt,0) --
            +(64pt,-32pt) -- +(64pt, -45pt) -- +(50pt, -45pt) -- +(23pt,-18pt) --
            +(-9pt,-18pt) -- +(-9pt,0) --  ($ (a.center) + (0,9pt) $);
            \draw[fill=black!20, rotate=180] ($ (b.center) + (0,9pt) $) -- +(32pt,0) --
            +(64pt,-32pt) -- +(64pt, -45pt) -- +(50pt, -45pt) -- +(23pt,-18pt) --
            +(-9pt,-18pt) -- +(-9pt,0) --  ($ (b.center) + (0,9pt) $);
        \end{scope}
        \begin{scope}[every node/.style={draw, rectangle, minimum height=18pt, minimum width=72pt, rounded corners, opacity=.5}]
            \node[fill=color1] at (c.center) {};
            \node[rotate=90, fill=color2] at (c.center) {};
        \end{scope}
    \end{pgfonlayer}

    \draw (c) -- (l);
    \draw (l) -- (l1);
    \draw (l) -- (l2);
    \draw (c) -- (a);
    \draw (a) -- (a1);
    \draw (a) -- (a2);
    \draw (c) -- (r);
    \draw (r) -- (r1);
    \draw (r) -- (r2);
    \draw (c) -- (b);
    \draw (b) -- (b1);
    \draw (b) -- (b2);
    \draw (a2) -- (r1);
    \draw (b1) -- (l2);
\end{tikzpicture}

%% file: figures/counting_example_opt.tex
\begin{tikzpicture}
    \begin{scope}[every node/.style={circle, draw, inner sep=0pt, minimum size=12pt, node distance=15pt}]
        \node (c) {};
        \node[left=of c] (l) {\normalsize 3};
        \node[left=of l] (l2) {};
        \node[above=of l2] (l1) {};
        \node[above=of c] (a) {\normalsize 2};
        \node[left=of a] (a1) {};
        \node[right=of a] (a2) {};
        \node[right=of c] (r) {\normalsize 1};
        \node[right=of r] (r1) {};
        \node[below=of r1] (r2) {};
        \node[below=of c] (b) {\normalsize 4};
        \node[left=of b] (b1) {};
        \node[right=of b] (b2) {};
    \end{scope}

    \begin{pgfonlayer}{background}
        \begin{scope}[every node/.style={draw, rectangle, minimum height=18pt, minimum width=72pt, rounded corners, opacity=.5}]
            \node[fill=color1] at (r.center) {};
            \node[fill=color3] at (l.center) {};
        \end{scope}
        \begin{scope}[every path/.style={rounded corners, opacity=.5}]
            \draw[fill=color2, rotate=90] ($ (a.center) + (0,9pt) $) -- +(9pt,0) --
            +(9pt,-45pt) -- +(-9pt,-45pt) --
            +(-9pt,-18pt) -- +(-36pt,-18pt) -- +(-36pt,0) --  ($ (a.center) + (0,9pt) $);
            \draw[fill=color4, rotate=270] ($ (b.center) + (0,9pt) $) -- +(9pt,0) --
            +(9pt,-45pt) -- +(-9pt,-45pt) --
            +(-9pt,-18pt) -- +(-36pt,-18pt) -- +(-36pt,0) -- ($ (b.center) + (0,9pt) $);
        \end{scope}
    \end{pgfonlayer}

    \draw (c) -- (l);
    \draw (l) -- (l1);
    \draw (l) -- (l2);
    \draw (c) -- (a);
    \draw (a) -- (a1);
    \draw (a) -- (a2);
    \draw (c) -- (r);
    \draw (r) -- (r1);
    \draw (r) -- (r2);
    \draw (c) -- (b);
    \draw (b) -- (b1);
    \draw (b) -- (b2);
    \draw (a2) -- (r1);
    \draw (b1) -- (l2);
\end{tikzpicture}

%% file: related_work.tex
%{\bf Related Work.}\quad
\section{Related Work}\label{sse:related}
Our grammar formalism is known as context-free 
hyperedge replacement (HR) grammars,
see~\cite{Engelfriet:1997:CGG:267871.267874,DBLP:conf/gg/DrewesKH97}.
An approximation algorithm for finding a small HR grammar that generates a given graph is considered already
by Peshkin~\cite{Peshkin07_graphSequitour}. It is based on the SEQUITUR compression
scheme~\cite{DBLP:journals/jair/Nevill-ManningW97}. However, experiments are only presented
for rather small protein graphs, and we have not been able to obtain their implementation. As far
as we know, no other compressor for straight-line graph grammars has been considered. 
Claude and Navarro~\cite{Claude10_compactWebGraphRep} apply string RePair on 
the adjacency list of a graph. This works well, but is outperformed by newer compression schemes
such as $k^2$-trees~\cite{DBLP:journals/is/BrisaboaLN14}.
More database oriented work is found for semi-structured data. For example the
XMill-compressor~\cite{DBLP:conf/sigmod/LiefkeS00} groups XML-data such that a subsequent use of
general-purpose compression (e.g. gzip) is more effective. Schema information can improve its
effectiveness, but is not required. Deriving schema information from existing data can be seen as a
form of lossy compression. DataGuides~\cite{DBLP:conf/vldb/GoldmanW97} are a way of doing just that
for XML data. As XML documents can be represented as trees, methods to compress trees are applicable.
Grammar-based tree compression can be seen as a precursor to the work presented here. One of the
first such algorithms was BPLEX~\cite{DBLP:journals/is/BusattoLM08}. The results of BPLEX were later
improved by applying the RePair compression scheme to trees~\cite{Lohrey13_treeRePair}, which is
what we are proposing to do for graphs.
\subsection{Succinct Graph Representations}
Several compression approaches have been developed particularly for web graphs. In a web graph,
nodes represent pages (i.e., URLs) and edges represent links from one page to another. Web graphs
have two properties which are useful for compression:
\begin{description}
    \item[Locality:] most links lead to pages within the same host (i.e., the URLs have the same
        prefix) and
    \item[Similarity:] pages on the same host often share the same links.
\end{description}
Due to these properties, ordering the nodes lexicographically by their URL provides an order in
which similar nodes are close to each other. The WebGraph framework~\cite{DBLP:conf/www/BoldiV04} by
Boldi and Vigna is originally based on this order, but was later improved with a different
order~\cite{Boldi09_permutingWebGraphs}. It represents the \emph{adjacency list} of a graph using
several layers of encodings, while retaining the ability to answer out-neighborhood queries. An
out-neighborhood (in-neighborhood) query applied to a node $u$ retrieves all nodes $v$ such that
there is an edge from $u$ to $v$ (from $v$ to $u$). As not every graph is a web graph, a
lexicographical order of the node-names is not always possible or useful. Apostolico and Drovandi
therefore propose to use a BFS-order~\cite{Apostolico09_graphCompressionBFS} combined with another
encoding. A different approach is proposed by Grabowski and
Bieniecki~\cite{DBLP:journals/dam/GrabowskiB14}, where contiguous blocks of the adjacency list are
merged into a single ordered list, and a list of flags which are used to recover the original lists.
They then encode the ordered list and use the deflate-compressor to compress both lists. To our
knowledge, their method is the current state-of-the-art in compression/query trade-off, when only
out-neighborhood queries are considered.

The methods above have in common that they encode the adjacency list of a graph and natively only
support out-neighborhood queries. The $k^2$-trees of Brisaboa et
al.~\cite{DBLP:journals/is/BrisaboaLN14} on the other hand compress the \emph{adjacency matrix} of
the graph. They do this by
recursively partitioning it into $k^2$ many squares. If one of these includes only 0-values, then
it is represented by a 0-leaf in the tree, and otherwise the square is partitioned
further. This Quadtree-like representation is well known (see,
e.g.,~\cite{DBLP:conf/synasc/Simecek09}), but their succinct binary encoding is a clever new
approach. The method provides access to both, in- and out-neighborhood queries, and can be applied
to any binary relation. We use $k^2$-trees to represent the start graph of our grammars. The
\ktree-method was combined by Hern\'andez and Navarro~\cite{DBLP:journals/kais/HernandezN14} with
dense substructure removal, originally proposed by Buehrer and
Chellapilla~\cite{Buehrer08_scalableWebGraphCompression}. A \emph{dense substructure} is defined by
two sets of nodes $U,S$ such that they induce a complete bipartite graph. Note that $U$ and $S$ need
not be disjoint. The edges in these bicliques are replaced by a single ``virtual node''. To our
knowledge, the method of~\cite{DBLP:journals/kais/HernandezN14} is the current state-of-the-art
in compression/query trade-off, when in- and out-neighborhood queries are considered.
\subsection{RDF Graph Compression}\label{sss:rdfGraphCompressionRelated}
The Resource Description Framework (RDF) is a fairly recent specification, first standardized by the
w3c in February 2004. The current version is
RDF1.1\footnote{\url{https://www.w3.org/TR/2014/REC-rdf11-concepts-20140225/}} from February 2014.
It is used to represent linked data and semantic information. Its relationship to graphs is
similar to XML's relationship to trees, in that graphs are a natural representation of the structure
defined by RDF. Roughly speaking, RDF data is represented as a set of triples $(s,p,o)$, connecting
a subject $s$ with an object $o$ by a predicate $p$. Notably, the domains for these can overlap to
some degree (for example, values used as predicates in some triples may be subjects in others). Such
a set of triples can be represented by having nodes for the subjects and objects, and edges for the
predicates. Thus $(s,p,o)$ becomes an edge from $s$ to $o$ labeled $p$. Note that it is not
necessary to model overlapping domains as edges pointing to other edges. Instead, there may be a
node label that also appears as an edge label. As RDF graphs encode
semantic information, the concrete values for subjects, predicates, and objects can be long strings
(for example URIs). It is a common practice (see e.g.~\cite{DBLP:conf/www/FernandezGM10,
DBLP:conf/sac/Martinez-PrietoFC12, DBLP:journals/kais/Alvarez-GarciaB15}) to map the possible values
to integers using a dictionary and to represent the graph using triples of integers. This also leads
to different approaches regarding RDF graph compression: one is to compress the dictionary further
(see, e.g.,~\cite{DBLP:conf/sac/Martinez-PrietoFC12,DBLP:journals/concurrency/UrbaniMDSB13}),
another to compress the underlying graph structure.

Of the methods mentioned in the previous section, only the \ktree has been applied to RDF graph
compression~\cite{DBLP:journals/kais/Alvarez-GarciaB15}. This is done by encoding a separate
adjacency matrix for every predicate (i.e., edge label) in the graph. The aforementioned splitting
into dictionary and triples was first proposed by Fern\'andez et
al.~\cite{DBLP:conf/semweb/FernandezMG10} as a datastructure called HDT (Header-Dictionary-Triples).
Together with an encoding of the triples grouping them by subject as an adjacency list, this yields
an in-memory representation of RDF data. They also investigate compressing this data structure using
conventional general-purpose compression methods (gzip, bzip2, and ppmdi), beating these universal
compressors used on the plain original file. Another approach of Swacha and
Grabowski~\cite{DBLP:conf/slate/SwachaG15} combines techniques to compress graph and dictionary
achieving a succinct representation of the RDF graph, which is subsequently compressed by general
purpose compressors.
%This reduces the number of edges, but increases the
%number of nodes.
%Short overview of graph compression methods we compare against.
%\begin{itemize}
%    \item \ktree \cite{Brisaboa09_k2Trees}
%    \item RePair on Adjacency Lists \cite{Claude10_compactWebGraphRep}
%    \item Webgraph framework (?) \cite{DBLP:conf/www/BoldiV04}
%    \item Dense subgraph discovery (?) \cite{,}
%\end{itemize}
\subsection{Queries with Compressed Input}
Often compressed input induces an exponential blow-up in complexity when deciding queries on the
represented data. Such ``upgrading-theorems'' have been shown for example for the representation of
graphs as boolean circuits, which can be exponentially smaller than an explicit representation, but
in turn make queries exponentially harder to answer (see
e.g.~\cite{DBLP:journals/iandc/GalperinW83,DBLP:journals/iandc/PapadimitriouY86}).
Fortunately, upgrading-theorems do not hold for grammar-compressed graphs as already shown
in~\cite{DBLP:journals/jcss/LengauerW92}. Indeed, it is possible to obtain a \emph{speed-up}
proportional to the compression ratio for certain queries. Lengauer and
Wanke~\cite{DBLP:journals/siamcomp/LengauerW88} show that connectivity of graphs specified using a
hierarchical graph definition (which is essentially the same as the hyperedge replacement mechanism
we use in this paper) can be decided in linear time with respect to the size of the compressed
representation. They also show this for biconnectivity and (undirected) reachability. For strong
connectivity however, they only get a quadratic upper bound. It should be noted that there is still
an increase in complexity here, but a smaller one. Lengauer and
Wagner~\cite{DBLP:journals/jcss/LengauerW92} study several problems and their complexity on explicit
versus hierarchical representations. They show, for example, that reachability is complete for $\P$
on hierarchical graphs, but the problem is \NL-complete for explicit representations. The latter
class is widely believed to be a proper subset of the former, but both classes only include problems
which can be solved in polynomial time, so there is no exponential blow-up in computation time. As
the hierarchical representation can be exponentially smaller than the represented graph, this makes
a speed-up possible.
This smaller increase in complexity is also stated by Marathe et
al.~\cite{DBLP:journals/tcs/MaratheRHR97,DBLP:journals/njc/MaratheHR94,DBLP:journals/siamcomp/MaratheHSR98},
who show that some \NP-complete problems on graphs (e.g., 3-colorization, Max Cut, and Vertex Cover)
become \PSPACE-hard on hierarchically defined graphs, which again does not imply an exponential
blow-up in computation time (assuming widely-held assumptions in complexity theory). They also
present polynomial approximation algorithms for some of these problems with hierarchical input. Note
that we cannot conclude that hierarchical versions of classical problems always lead to a speed-up:
in the same paper they show that the $\P$-complete threshold network flow problem becomes
$\PSPACE$-complete for hierarchical input graphs.

%% file: preliminaries.tex
\section{Preliminaries}\label{sse:preliminaries}
For $n \geq 0$ we denote by $[n]$ the set $\{1,\ldots,n\}$. A \emph{ranked alphabet} consists of an
alphabet $\Sigma$ together with a mapping $\rank: \Sigma \rightarrow \N$ that assigns
a rank to every symbol in $\Sigma$. For the rest of the paper, we assume that $\Sigma$ is fixed, and
of the form $[\sigma]$ for some integer $\sigma \in \N$.

As we do some comparisons to grammar-based string and tree compression, let us recall some
definitions. We denote the empty string by $\varepsilon$, and for two strings $w,v$ the
concatenation of $w$ and $v$ by $w\cdot v$, or $wv$ if that is unambiguous. Let $\Delta$ be an
\emph{alphabet} (that is, a finite set of symbols). The size of a string $w = a_1a_2\cdots a_k$ ($a_i
\in \Delta$ for $i \in [k]$) is defined as $|w| = k$ with $|\varepsilon| = 0$, and we write $a_i \in
w$ to express that $a_i$ is part of the string $w$. We recursively define trees over
a ranked alphabet $\Sigma$ as the smallest set $T$, such that for $k \geq 0$, if $t_1,\ldots,t_k \in
T$, then $\sigma(t_1,\ldots,t_k) \in T$, provided that $\sigma \in \Sigma$ and $\rank(\sigma) = k$.
The nodes of a tree $t$ can be addressed by their \emph{Dewey addresses} $D(t)$. For a tree $t =
\sigma(t_1,\ldots,t_k)$, $D(t)$ is recursively defined as $D(t) = \{\varepsilon\} \cup \bigcup_{i \in
[k]} i \cdot D(t_i)$. Thus $\varepsilon$ addresses the root node, and $u \cdot i$ the $i$-th child
of $u$. For $u \in D(t)$ we denote by $t[u] \in \Sigma$ the symbol at $u$. The size of $t$ is $|t| =
|D(t)|-1$, i.e., the number of edges in $t$.

A \emph{hypergraph over $\Sigma$} (or simply \emph{graph}) is a tuple $g = (V, E, \att, \lbl, \ext)$
where $V = \{1,\ldots,n\}$, $n \geq 1$ is the set of nodes, $E$ is a finite set of edges, $\att: E
\rightarrow V^+$ is the attachment mapping, $\lbl: E \rightarrow \Sigma$ is the label mapping, and
$\ext \in V^*$ is a string of \emph{external nodes}. We define the rank of an edge as
$\rank(e) = |\att(e)|$ and require that $\rank(e) = \rank(\lbl(e)) \geq 1$ for every edge $e$ in
$E$. The latter requirement (and the attachment mapping mapping to $V^+$) mean that we do not allow
edges that are not attached to any nodes. We
add the following two restrictions on hypergraphs:
%\begin{itemize}
%    \item for all edges $e \in E: \att(e)$ contains no node twice,
%    \item $\ext$ contains no node twice, and
%    \item $V = \{1,2,\ldots,m\}$ for some $m$; these numbers are called \emph{node IDs}.
%\end{itemize}
\begin{enumerate}[label=(C\arabic{*})]
    \item for all edges $e \in E: \att(e)$ contains no node twice, and\label{condition1}
    \item $\ext$ contains no node twice.\label{condition2}
\end{enumerate}
%(1)~ for all edges $e \in E: \att(e)$ contains no node twice, (2)~ $\ext$ contains no node twice,
%and (3)~$V = \{1,2,\ldots,m\}$ for some $m$; these numbers are called \emph{node IDs}.
An edge is \emph{simple}, if its rank equals two. A hypergraph $g$ is simple, if all its edges are
simple, and for any distinct edges $e_1,e_2$ of $g$ it holds that $\att(e_1) \neq \att(e_2)$ or
$\lbl(e_1) \neq \lbl(e_2)$, i.e., $g$ has no ``multi-edges''.  For a hypergraph $g = (V,E,\att,
\lbl, \ext)$ we use $V_g, E_g, \att_g, \lbl_g$, and $\ext_g$ to refer to its components. We may omit
the subscript if the hypergraph is clear from context. The rank of a hypergraph $g$ is defined as
$\rank(g) = |\ext_g|$. Nodes that are not external are called \emph{internal}. We define the
\emph{node size} of $g$ as $|g|_V = |V|$, the \emph{edge size} as 
\[|g|_E = |\{e \in E_g \mid \rank(e) \leq 2\}| + \sum_{e \in E_g, \rank(e) > 2} \rank(e),\]
and the \emph{size} as $|g| = |g|_V + |g|_E$. We denote the set of all hypergraphs over $\Sigma$
by $\HGR(\Sigma)$. Note that our size definition differs slightly from the one given
in~\cite{DBLP:journals/siamcomp/LengauerW88}, because they are encoding hyperedges using bipartite
graphs. Thus, an edge of rank $n$ in our model has size $n$ (or 1 if $n \leq 2$), whereas they would
calculate a size of $n$ edges and one node. This makes our sizes slightly smaller. For two nodes
$s,t \in V$ we say that there is a \emph{path $p$} from $s$ to $t$, if there exists a sequence of edges
$p = (e_1,\ldots,e_n) \in E^n$ for some $n \in \N$, such that $\att(e_1) = s\alpha$ ($\alpha \in V^*$),
$\att(e_n) = v\beta$ with $v \in V, \beta \in V^+$ and $t \in \beta$, and for every $i \in [n-1]$ if
$\att(e_i) = uv_1\cdots v_k$ then there exists a $j \in [k]$ such that $\att(e_{i+1}) = v_j\alpha$
for some $\alpha \in V^*$. For an edge $e$ let $\head(e)$ be the first node $e$ is attached to. For
a path $p = (e_1,\ldots,e_n)$ we refer to the nodes $(\head(e_2),\ldots,\head(e_n))$ as the
\emph{nodes on the path $p$}. We call $p$ internal, if all the nodes on $p$ are internal.

% TODO: The following part about data values goes somewhere else
%We also assume that the node IDs may represent arbitrary data values.
%Let $D$ be a set of data values, then for every hypergraph, there is a mapping $\varphi: V
%\rightarrow D$ assigning values to nodes.
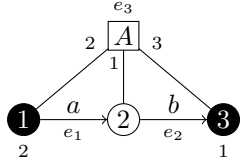
\begin{figure}[!t]
    \centering
    \input{figures/hypergraph_example}
    \caption{Example of a drawing of a hypergraph with three nodes (two of which external) and three
    edges, two of which are simple.}
    \label{fig:hypergraph_example}
\end{figure}
An example of a hypergraph can be seen in Figure~\ref{fig:hypergraph_example}. Formally, the pictured
graph has $V = \{1,2,3\}$, $E = \{e_1,e_2,e_3\}$, $\att = \{e_1 \mapsto 1\cdot2, e_2 \mapsto 2\cdot3,
e_3 \mapsto 2\cdot1\cdot3\}$, $\lbl = \{e_1 \mapsto a, e_2 \mapsto b, e_3 \mapsto A\}$, and $\ext =
3\cdot1$. Note that external nodes are filled black and have indices below them indicating their
position within $\ext$. Similarly, the hyperedge of rank $> 2$ has indices indicating the order of
the attached nodes (simple edges are drawn directed, from their first to second attachment node). In
the following we sometimes omit either of these indices. In these cases, we either use colors to
indicate this order, or the specific order is irrelevant for the example.
%\begin{figure}[!t]
%    \centering
%    \input{figures/hypergraph_example}
%    \caption{Example of a hypergraph with two simple- and one hyperedge.}
%    \label{fig:hypergraph_example}
%\end{figure}
%
%For a node $v \in V_g$ of a hypergraph $g$ we denote by $\neighbor(v) = \{u \in V_g \mid \exists e \in
%E_g: v \in \att_g(e) $ and $ u \in \att_g(e)\}$ the \emph{neighborhood of $v$}. For simple graphs we also
%define $\inNeighbor(v) = \{u \in V_g \mid \exists e \in E_g: \att(e) = uv\}$ and $\outNeighbor(v) =
%\{u \in V_g \mid \exists e \in E_g: \att(e) = vu\}$, the \emph{incoming} and \emph{outgoing
%neighborhoods of $v$}, respectively. Furthermore we let $\incident(v) = \{e \in E_g \mid v \in \att(e)\}$
%be the set of edges incident with $v$. % TODO: possibly define incoming/outgoing edges as well
%
%In the literature (see, e.g., \cite{Engelfriet:1997:CGG:267871.267874}) hypergraphs and hyperedge
%replacement grammars are commonly defined without the restrictions (C1) and (C2) above. We have
%these restrictions, and define the following variant of hyperedge replacement grammars for two
%reasons. First our algorithm would never create such patterns by design. Secondly, having these
%restrictions simplifies the definition of the replacement operation. See
%Section~\ref{sss:formalism_choice} for some more insight into the consequences of this choice of
%formalism.
\begin{definition}
    A \emph{hyperedge replacement grammar (HR grammar) over $\Sigma$} is a tuple $G = (N, P, S)$,
    where $N \subseteq \N$ is a ranked alphabet of \emph{nonterminals} with $N \cap
    \Sigma=\emptyset$, $P \subset N \times \HGR(\Sigma \cup N)$ is the set of \emph{rules} such that
    $\rank(A) = \rank(g)$ for every $(A,g) \in P$, and $S \in \HGR(\Sigma\cup N)$ is the \emph{start
    graph}.
\end{definition}
The \emph{size} of $G$ is defined as $|G| = \sum_{(A,g) \in P}|g|$, and similarly the edge and node
sizes $|G|_E = \sum_{(A,g) \in P}|g|_E$ and $|G|_V = \sum_{(A,g) \in P}|g|_V$. The rank of an
HR grammar $G$ is defined by $\rank(G) = \max\{\rank(A) \mid A \in N\}$.  We often write $p:A
\rightarrow g$ for a rule $p = (A,g)$ and call $A$ the left-hand side $\lhs(p)$ and $g$ the
right-hand side $\rhs(p)$ of $p$. We call symbols in $\Sigma$ \emph{terminals}. Consequently an edge
is called \emph{terminal} if it is labeled by a terminal and \emph{nonterminal} otherwise.

To define a derivation relation for the grammar $G = (N,P,S)$, we first introduce some notation. Let
$g$ be a hypergraph, $\nr: V_g \rightarrow \N$ a bijective function, and $\nr^*: V_g^* \rightarrow
\N^*$ the natural extension of $\nr$ to strings. We call $\nr$ a \emph{node renaming on $g$}. We define
$g[\nr] = (\{\nr(v) \mid v \in V_g\}, E_g, \att^\nr, \lbl_g, \nr^*(\ext_g))$, where $\att^\nr(e) =
\nr^*(\att_g(e))$ for all $e \in E_g$. For a hypergraph $g$ and an edge $e \in E_g$ we denote by
$g[-e]$ the hypergraph obtained from $g$ by removing the edge $e$. For two hypergraphs $g,h$ we
denote by $g \cup h$ the \emph{union} of the two hypergraphs defined by $(V_g \cup V_h, E_g \uplus
E_h, \att_g \uplus \att_h, \lbl_g \uplus \lbl_h, \ext_g)$. Note that this union
\begin{enumerate}
    \item merges nodes that exist in both $g$ and $h$,
    \item creates disjoint copies of $E_h$, $\att_h$, and $\lbl_h$, and
    \item uses the external nodes of $g$ only.
\end{enumerate}
Now, let $g,h$ be hypergraphs and $e \in E_g$ such that $\rank(e) = \rank(h)$. Let $\nr$ be a node
renaming on $h$ such that $\ext_{h[\nr]} = \att_g(e)$, and $\nr(v) \notin V_g$ for every
internal node $v$ of $h$. The \emph{replacement of $e$ by $h$ in $g$} is defined as $g[e/h] = g[-e]
\cup h[\nr]$. 

For a grammar $G = (N, P, S)$ we define its \emph{derivation relation} $\Rightarrow_G \subseteq
\HGR(\Sigma\cup N) \times \HGR(\Sigma\cup N)$ as follows. For $g,h \in \HGR(\Sigma\cup N)$, $g
\Rightarrow_G h$ if and only if there is a nonterminal edge $e$ in $g$ such that $h =
g[e/\rhs(\lbl(e))]$. For $n \geq 1$ we write $g \Rightarrow^n_G h$, if there is a sequence $g
\Rightarrow_G g_1 \Rightarrow_G \cdots \Rightarrow_G g_{n-1} \Rightarrow_G h$ of $n$ derivation
steps to derive $h$ from $g$, and extend this to $g \Rightarrow^* h$, if $g \Rightarrow^n h$ for
some $n \in \N$. The language of a grammar $G$ is defined as $L(G) = \{g \in \HGR(\Sigma) \mid S
\Rightarrow^* g\}$. We omit the subscript $G$ where the grammar is clear from context. For an
example of a derivation, see Figure~\ref{fig:intro_derivation_example}, which shows the full
derivation of the grammar given in Figure~\ref{fig:intro_grammar_example}. Note that the nodes in
the start graph initially have IDs $1$ and $2$, whereas the nodes that originate from the internal
node of the right-hand side of $A$ get variables $x,y$, and $z$. Any pairwise different values for
these variables that are different from $1$ and $2$ yield a correct derivation of this grammar.

\begin{figure}[!t]
    \centering
    \input{figures/intro_derivation_example}
    \caption{Derivation steps of the grammar in Figure~\ref{fig:intro_grammar_example}}
    \label{fig:intro_derivation_example}
\end{figure}
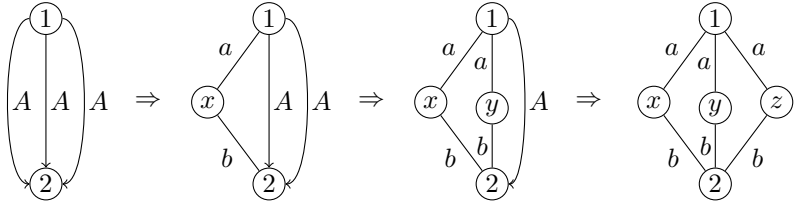
% This definition leaves open how the node IDs are determined during the
%derivation (conventionally HR grammars are used to represent graphs up to isomorphism). For more
%details we refer to~\cite{Engelfriet:1997:CGG:267871.267874}. 

\begin{definition}
    An HR grammar $G$ is called \emph{straight-line} (SL-HR grammar) if 
    \begin{enumerate}
        \item the relation $\SLOrder = \{(A_1, A_2)\mid \exists g: (A_1,g)\in P_G, \exists e \in
            E_g:\lbl_g(e) = A_2\}$ is acyclic,
        \item for every $A \in N$ there exists exactly one rule $p \in P_G$ with
            $\lhs(p) = A$, and
        \item $G$ has no useless (unreachable) rules, i.e., if $S_G \Rightarrow g_1 \Rightarrow
        \cdots \Rightarrow g_n \in \HGR(\Sigma)$ and $A \in N_G$, then $\exists i \exists e \in
        V_{g_i}: \lbl_{g_i}(e) = A$.
    \end{enumerate}
    \label{def:SLHR}
\end{definition}
Note that $L(G) \neq \emptyset$ and $L(G)$ contains (due to the element of choice in the node
renaming, possibly infinitely many) only isomorphic graphs. As the
right-hand side for a nonterminal is unique in SL-HR grammars, we denote the right-hand side $g$ of
$p = (A,g)$ by $\rhs(A)$. By convention, whenever we state something over all right-hand
sides of a grammar, this includes the startgraph (i.e., $S$ is assumed to be implicitly in $N$ and
$\rhs(S) = S$). The height of an SL-HR grammar $\height(G)$ is the height of $\SLOrder$. 
Note that in terms of the expressive power of HR grammars for defining graph languages, the
conditions~\ref{condition1} and~\ref{condition2} do not have an effect, see~\cite[Chapter 1, Theorem
4.6]{DBLP:books/sp/Habel92}. In terms of compression power however, condition~\ref{condition1} is
harmless, but~\ref{condition2} is not. See Section~\ref{sss:formalism_choice} for more details.
In the following we often say \emph{grammar} instead of \emph{SL-HR grammar}.
\subsection{Creating Graphs with Unique Node IDs}\label{sss:nodeIDs}
% TODO: Include following paragraph into introduction
%Recall that the set of nodes of a hypergraph is a prefix $\{1,\ldots,n\}$ of $\N$. We call these the
%\emph{node IDs}, which we assume to refer to more complex data values associated with the node.
%This means, for a graph $g$ there exists a mapping $\varphi:V_g \mapsto D$ for some (potentially
%infinite) set of data values $D$. We will not include the cost of this mapping in our
%considerations, as many graph analysis algorithms (e.g. PageRank) only consider the structure of the
%graph, and not the data on each node. We therefore only deal with node IDs, and assume that the
%list of concrete data values can be reordered.
%
%That being said, even with reordering of the data values, we still want to be able to recover
%the same graph from the grammar, not just an isomorphic one. The grammar formalism as defined only
%recovers some isomorphic copy of the original graph, as it is not defined what node IDs are given to
%internal nodes during a derivation. We now present a method to assign precise node IDs.
Let $G$ be an SL-HR grammar. Instead of considering all the (possibly infinitely many) isomorphic graphs in
the set $L(G)$, we would like to fix one particular graph $g$ of $L(G)$.
First of all, we fix a particular node renaming during a derivation step. Let $g$ be a graph with a
nonterminal edge $e$ labeled $A$, and let $h = \rhs(A)$. Further, let $|V_g| = n$ and let 
$v_1,\ldots,v_m$ be the internal nodes of $h$ such that $v_i < v_j$ for $i < j$. In the definition
of $g[e/h] = g[-e] \cup h^\nr$, we now require $\nr(v_i) = n+i$ for $i \in [m]$.

This alone is not enough. We also need to define in which order the nonterminal edges are derived,
to make sure the full derivation ends up with a unique graph. To do so, let first $g$ be a
graph, and $\NTSet_g = \{e \in E_g \mid e$ is nonterminal$\}$ be the set of nonterminal edges in
$g$. We define the \emph{sibling-tuple} $\sib(\NTSet_g)$ of $\NTSet_g$ as the tuple
$(e_1,\ldots,e_k)$ such that $\bigcup_{i \in [k]} e_i = \NTSet_g$, and if $i < j$ (for $i,j\in [k]$,
i.e., $e_i$ comes before $e_j$ in the tuple), then
\begin{itemize}
    \item $\att_g(e_i) <_{\text{lex}} \att_g(e_j)$ (here $<_{\text{lex}}$ is the lexicographical
        order), or
    \item $\att_g(e_i) = \att_g(e_j)$ and $\lbl_g(e_i) \leq \lbl_g(e_j)$.
\end{itemize}
Note that the order for edges with $\att_g(e_i) = \att_g(e_j)$ and $\lbl_g(e_i) = \lbl_g(e_j)$ is
arbitrary, as it does not matter. Intuitively, the sibling-tuple is an ordered sequence of siblings
within the derivation tree of $G$. Using this order we define the \emph{derivation-tree} $\dt(G)$
for an SL-HR grammar $G = (N,P,S)$ recursively in the following way: let $e$ be a nonterminal edge
labeled $A$, let $g = \rhs(A)$, and let $\sib(\NTSet_g) = (e_1,\ldots,e_m)$ for $m\geq 0$. Then
$\dt(e) = e(\dt(e_1),\ldots,\dt(e_m))$. To define the root of the tree, let $\sib(\NTSet_S) =
(e_1,\ldots,e_m)$, and $\dt(G) = S(\dt(e_1),\ldots,\dt(e_m))$. Figure~\ref{fig:dtExample} shows an
example of a derivation tree. The order in which the nonterminal edges of $G$ are derived is now
given by a pre-order traversal (i.e., a depth-first traversal, which always visits the leftmost
unvisited node next) of $\dt(G)$. This yields a unique hypergraph out of the many isomorphic options
in $L(G)$. We denote this hypergraph by $\val(G)$. We denote the hypergraph derived by a single edge
$e$ in this way by $\val(e)$. Finally we denote by $\val(A)$ the hypergraph derived from a graph
with $\rank(A)$ many nodes attached to a single $A$-edge.

For an example consider first again the derivation shown in
Figure~\ref{fig:intro_derivation_example}. The only allowed choice of node-IDs for $\val(G)$ here
would be $x=3$, $y=4$, and $z=5$. Furthermore, Figure~\ref{fig:valGExample} shows the graph
resulting from the derivation of the grammar given in Figure~\ref{fig:dtExample}. Note how in both
cases the order in which the $A$-edges between nodes $1$ and $2$ are derived does not matter: if we
were to switch the node-IDs created by them ($4$ and $5$ in Figure~\ref{fig:valGExample}) we would
still have the same graph, not just an isomorphic one.
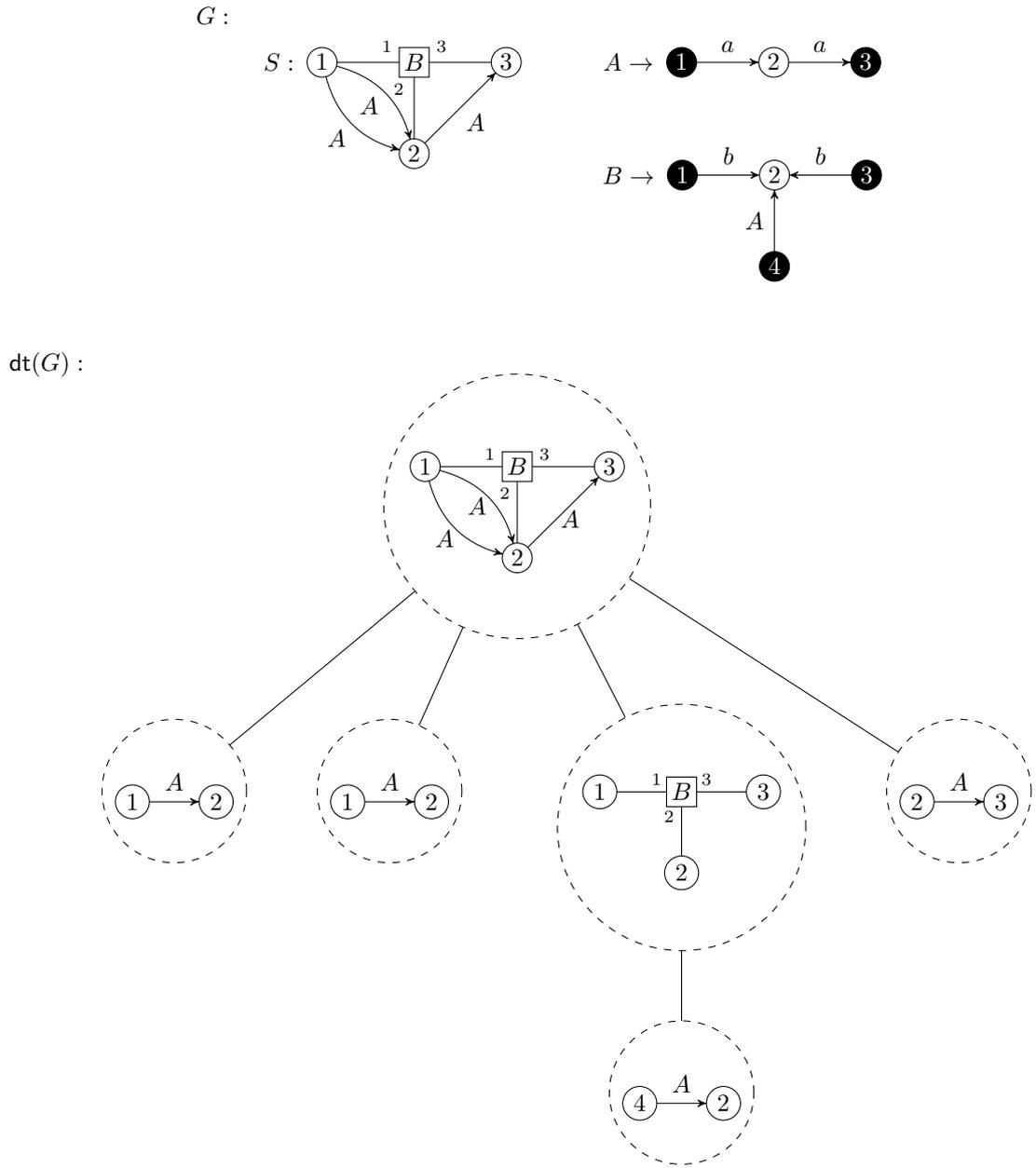
\begin{figure}[!t]
    \centering
    \input{figures/derivationTreeExample}
    \caption{A grammar and its derivation tree.}
    \label{fig:dtExample}
\end{figure}
\begin{figure}[!t]
    \centering
    \input{figures/valGExample}
    \caption{The graph $\val(G)$ for the grammar given in Figure~\ref{fig:dtExample}.}
    \label{fig:valGExample}
\end{figure}
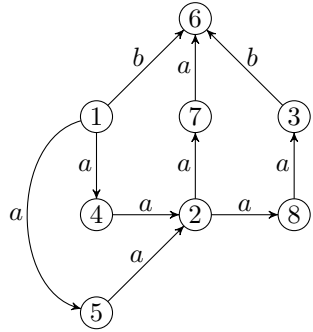
\subsection{String- and Tree-Graphs}\label{sss:stringTreeDef}
Strings and trees can be conveniently encoded as hypergraphs. As the latter structure is
more complex, this affects the size. We make use of these encodings, when discussing the
relation of graph compression to previously proposed string- and tree-compression in
Section~\ref{ss2:StringTreeRelation}. Let $w = a_1\cdots a_n$ be a string. We define the hypergraph
$\sgraph(w)$ representing the string $w$ as $(\{0,1,\ldots,n\}, E_w, \att_w, \lbl_w, \ext_w)$, where
$E_w = \{e_1,\ldots,e_n\}$. For an edge $e_i \in E_w$ we let $\att_w(e_i) = i-1\cdot i$ and
$\lbl(e_i) = w_i$. Finally the external nodes are $\ext_w = 0\cdot n$. The size of $\sgraph(w)$ is
$2n+1 = 2|w|+1$. Note that the external nodes indicate beginning and end of the string.
Figure~\ref{fig:stringgraph_example} shows $\sgraph(w)$ for $w = abca$.

\begin{figure}[!t]
    \centering
    \subfloat[]{\input{figures/example_stringgraph}\label{fig:stringgraph_example}}\qquad
    \subfloat[]{\input{figures/example_treegraph}\label{fig:treegraph_example}}
    \caption{Example of a string-graph~\protect\subref{fig:stringgraph_example}, and a
    tree-graph~\protect\subref{fig:treegraph_example}.}
    \label{fig:ex_stringTreeGraph}
\end{figure}
The encoding for trees is a little less straight-forward, because the node labels are moved
into the edges, and the children of a node in a tree are ordered. We convert a node with $k$
children into a hyperedge of rank
$k+1$ (as it is connected to the parent as well). The order of the hyperedge retains the order of
the children and the parent is always located at the first node attached to a hyperedge. This is a
well-known encoding, see e.g.,~\cite{DBLP:journals/acta/EngelfrietH92}. Let $t$ be a tree, and let
$<_t$ be the lexicographical order on $D(t)$. For an address $i \in D(t)$ we denote by
$\mathsf{pos}(i)$ the position of $i$ within the order $<_t$, such that $\mathsf{pos}(\varepsilon) =
0$. We define the hypergraph $\tgraph(t)$ representing the tree $t$ as $(\{0,1,\ldots,|D(t)|\},
\{e_i \mid i \in D(t)\}, \att, \lbl, 0)$, where for $i \in D(t)$ and $k = \rank(t[i])$,
\begin{itemize}
    \item $\att(e_i) = (\mathsf{pos}(i))\cdot(\mathsf{pos}(i)+1) \cdots (\mathsf{pos}(i)+k)$, and
    \item $\lbl(e_i) = t[i]$.
\end{itemize}
%For a tree $t$, we denote its encoding as a \emph{tree-graph} by $\tgraph(t)$.
The size of $\tgraph(t)$ is between $2|t|+2$ (for a monadic tree $t$) and $3|t|+2$ (for a tree where
every inner node has rank greater than $1$). Figure~\ref{fig:treegraph_example} shows $\tgraph(t)$
for $t = f(a, g(a,b))$.
%Here, the external node represents the root of the tree. The encoding is about two to three times
%the size of $t$ with $2|t| + 2 \leq |g_t| \leq 3|t|+2$. To clarify: if $t$ has no nodes of rank $1$,
%then $|g_t| = 3|t|+2$. As rank $1$ nodes will be turned into simple edges, these contribute less to
%the size of $g_t$, with $|g_t| = 2|t|+2$ in the extreme case of an entirely monadic tree.
%If an SL-HR grammar $G$ was constructed to represent a hypergraph $g$, it is
%possible to provide a mapping $\varphi': V_{\val(G)} \to D$ such that applying $\varphi'$ to the
%nodes of $\val(G)$ and $\varphi$ to the nodes of $g$ yields the same graph. 
%\begin{figure}[!t]
%    \centering
%    \input{figures/graph_derivation_order_example}
%    \caption{Result of the derivation of the Grammar in Figure~\ref{fig:grammar_example} with fixed
%order of the nonterminal edges.}
%    \label{fig:derivation_order_example}
%\end{figure}

%% file: figures/hypergraph_example.tex
\begin{tikzpicture}
    \begin{scope}[every node/.style={circle, draw, inner sep=0pt, minimum size=12pt, node
        distance=25pt}]
        \node[fill=black] (1) {\normalsize \color{white}$1$};
        \node[right=of 1] (2) {\normalsize $2$};
        \node[right=of 2,fill=black] (3) {\normalsize \color{white} $3$};
    \end{scope}
    \node[below=0pt of 3] (x1) {\scriptsize 1};
    \node[below=0pt of 1] (x2) {\scriptsize 2};
    \node[above=20pt of 2, rectangle, draw, inner sep=2pt, outer sep=0pt] (A) {\normalsize $A$};
    \node[above=-.2pt of A] {\scriptsize $e_3$};

    \draw[->] (1) -- node[above=-1pt] {\normalsize $a$} node[below] {\scriptsize $e_1$} (2);
    \draw[->] (2) -- node[above=-1pt] {\normalsize $b$} node[below] {\scriptsize $e_2$} (3);

    \draw (A) -- node[above left=-2pt, pos=.1] {\scriptsize $2$} (1);
    \draw (A) -- node[left=-2pt, pos=.25] {\scriptsize $1$} (2);
    \draw (A) -- node[above right=-2pt, pos=.1] {\scriptsize $3$} (3);
\end{tikzpicture}

%% file: figures/intro_derivation_example.tex
\begin{tikzpicture}
%    \node (S) {$S = $};
    \begin{scope}[every node/.style={circle, draw, inner sep=0pt, minimum size=12pt, node
        distance=50pt}]
%        \node[above right=.45 and .5 of S] (So1) {};
        \node (So1) {1};
        \node[below=of So1] (Su1) {2};
    \end{scope}

    \begin{scope}[every path/.style={->}]
        \draw (So1) -- node[right=-2pt] {\normalsize $A$} (Su1);
        \draw (So1) .. controls +(-.6,0) and +(-.6,0) .. node[right=-2pt] {\normalsize $A$} (Su1);
        \draw (So1) .. controls +(.6,0) and +(.6,0) .. node[right=-2pt](X) {\normalsize $A$} (Su1);
    \end{scope}

%    \node[right=1.5 of S] (RA1) {$\Rightarrow$};
    \node[right=.1 of X] (RA1) {$\Rightarrow$};

    \begin{scope}[every node/.style={circle, draw, inner sep=0pt, minimum size=12pt, node
        distance=50pt}]
        \node[right=2.5 of So1] (So2) {1};
        \node[below=of So2] (Su2) {2};
        \node[below left=.8 and .5 of So2] (in21) {$x$};
    \end{scope}

    \begin{scope}[every path/.style={->}]
        \draw (So2) -- node[right=-2pt] {\normalsize $A$} (Su2);
        \draw (So2) .. controls +(.6,0) and +(.6,0) .. node[right=-2pt] (X2) {\normalsize $A$} (Su2);
    \end{scope}
        \draw (So2) -- node[above left=-2pt] {\normalsize $a$} (in21);
        \draw (in21) -- node[below left=-2pt] {\normalsize $b$} (Su2);

    \node[right=.1 of X2] (RA2) {$\Rightarrow$};

    \begin{scope}[every node/.style={circle, draw, inner sep=0pt, minimum size=12pt, node
        distance=50pt}]
        \node[right=2.5 of So2] (So3) {1};
        \node[below=of So3] (Su3) {2};
        \node[below left=.8 and .5 of So3] (in31) {$x$};
        \node[below=.75 of So3] (in32) {$y$};
    \end{scope}

    \begin{scope}[every path/.style={->}]
        \draw (So3) .. controls +(.5,0) and +(.5,0) .. node[right=-2pt] (X3) {\normalsize $A$} (Su3);
    \end{scope}
        \draw (So3) -- node[above left=-2pt] {\normalsize $a$} (in31);
        \draw (in31) -- node[below left=-2pt] {\normalsize $b$} (Su3);
        \draw (So3) -- node[left=-2pt] {\normalsize $a$} (in32);
        \draw (in32) -- node[left=-2pt] {\normalsize $b$} (Su3);

    \node[right=.1 of X3] (RA3) {$\Rightarrow$};

    \begin{scope}[every node/.style={circle, draw, inner sep=0pt, minimum size=12pt, node
        distance=50pt}]
        \node[right=2.5 of So3] (So4) {1};
        \node[below=of So4] (Su4) {2};
        \node[below left=.8 and .5 of So4] (in41) {$x$};
        \node[below=.75 of So4] (in42) {$y$};
        \node[below right=.8 and .5 of So4] (in43) {$z$};
    \end{scope}

    \begin{scope}%[every path/.style={->}]
        \draw (So4) -- node[above left=-2pt] {\normalsize $a$} (in41);
        \draw (in41) -- node[below left=-2pt] {\normalsize $b$} (Su4);
        \draw (So4) -- node[left=-2pt] {\normalsize $a$} (in42);
        \draw (in42) -- node[left=-2pt] {\normalsize $b$} (Su4);
        \draw (So4) -- node[above right=-2pt] {\normalsize $a$} (in43);
        \draw (in43) -- node[below right=-2pt] {\normalsize $b$} (Su4);
    \end{scope}
\end{tikzpicture}

%% file: figures/derivationTreeExample.tex
\begin{tikzpicture}
    \node[rectangle, inner sep=5pt] (grammar) {
        \begin{tikzpicture}
            \node (S) {$S:$};
            \node[right=4 of S] (A) {$A \rightarrow$};
            \node[below=of A] (B) {$B \rightarrow$};

            \begin{scope}[every node/.style={circle, draw, inner sep=0pt, outer sep=0pt, minimum size=12pt,
                node distance=25pt}]
                \node[right=0pt of S] (s1) {$1$};
                \node[right=of s1, rectangle] (sB) {$B$};
                \node[below=of sB] (s2) {$2$};
                \node[right=of sB] (s3) {$3$};
            \end{scope}
            \draw (sB) --node[above=-1pt,pos=.2] {\scriptsize 1} (s1);
            \draw (sB) --node[left=-1pt,pos=.2] {\scriptsize 2} (s2);
            \draw (sB) --node[above=-1pt,pos=.2] {\scriptsize 3} (s3);
            \begin{scope}[every path/.style={-stealth'}]
                \draw (s1) edge[bend left] node[below left=-4pt] {$A$} (s2);
                \path (s1) edge[bend right] node[below left=-4pt] {$A$} (s2);
                \draw (s2) --node[below right=-4pt] {$A$} (s3);
            \end{scope}

            \begin{scope}[every node/.style={circle, draw, inner sep=0pt, outer sep=0pt, minimum size=12pt,
                node distance=25pt}]
                \node[right=0pt of A, fill=black] (a1) {\color{white} $1$};
                \node[right=of a1] (a2) {$2$};
                \node[right=of a2, fill=black] (a3) {\color{white} $3$};
            \end{scope}
            \begin{scope}[every path/.style={-stealth'}]
                \draw (a1) --node[above=-1pt] {$a$} (a2);
                \draw (a2) --node[above=-1pt] {$a$} (a3);
            \end{scope}

            \begin{scope}[every node/.style={circle, draw, inner sep=0pt, outer sep=0pt, minimum size=12pt,
                node distance=25pt}]
                \node[right=0pt of B, fill=black] (b1) {\color{white} $1$};
                \node[right=of b1] (b2) {$2$};
                \node[right=of b2, fill=black] (b3) {\color{white} $3$};
                \node[below=of b2, fill=black] (b4) {\color{white} $4$};
            \end{scope}
            \begin{scope}[every path/.style={-stealth'}]
                \draw (b1) --node[above=-1pt] {$b$} (b2);
                \draw (b3) --node[above=-1pt] {$b$} (b2);
                \draw (b4) --node[left=-1pt] {$A$} (b2);
            \end{scope}
        \end{tikzpicture}
    };
    \node[left=0pt of grammar.north west] (G) {$G:$};

    \node[below=of grammar] (dt) {
    \begin{tikzpicture}
        \node[dashed, circle, draw, minimum size=12pt, inner sep=2pt, outer sep=0pt] (tS) {%$S$};
        \begin{tikzpicture}[solid]
            \begin{scope}[every node/.style={circle, draw, inner sep=0pt, outer sep=0pt, minimum size=12pt,
                node distance=25pt}]
                \node[right=0pt of S] (ts1) {$1$};
                \node[right=of s1, rectangle] (tsB) {$B$};
                \node[below=of sB] (ts2) {$2$};
                \node[right=of sB] (ts3) {$3$};
            \end{scope}
            \draw (tsB) --node[above=-1pt,pos=.2] {\scriptsize 1} (ts1);
            \draw (tsB) --node[left=-1pt,pos=.2] {\scriptsize 2} (ts2);
            \draw (tsB) --node[above=-1pt,pos=.2] {\scriptsize 3} (ts3);
            \begin{scope}[every path/.style={-stealth'}]
                \draw (ts1) edge[bend left] node[below left=-2pt] {$A$} (ts2);
                \path (ts1) edge[bend right] node[below left=-2pt] {$A$} (ts2);
                \draw (ts2) --node[below right=-4pt] {$A$} (ts3);
            \end{scope}
        \end{tikzpicture}
        };
        \node[dashed, below left=2 and -.25 of tS, circle, draw, minimum size=12pt, inner sep=2pt, outer sep=0pt] (t2) {
            \begin{tikzpicture}[solid]
                \begin{scope}[every node/.style={circle, draw, inner sep=2pt, outer sep=0pt,
                    node distance=20pt, minimum size=12pt}]
                    \node (t21) {1};
                    \node[right=of t21] (t22) {2};
                \end{scope}
                \draw[-stealth'] (t21) --node[above] {$A$} (t22);
            \end{tikzpicture}
        };
        \node[dashed, below right=2 and -.25 of tS, circle, draw, minimum size=12pt, inner sep=2pt, outer sep=0pt] (t3) {
            \begin{tikzpicture}[solid]
                \begin{scope}[every node/.style={circle, draw, inner sep=2pt, outer sep=0pt,
                    node distance=20pt, minimum size=12pt}]
                    \node (t31) {1};
                    \node[rectangle, right=of t31] (t3B) {$B$};
                    \node[below=of t3B] (t32) {2};
                    \node[right=of t3B] (t33) {3};
                \end{scope}
                \draw (t3B) --node[above=-1pt,pos=.2] {\scriptsize $1$} (t31);
                \draw (t3B) --node[left=-1pt,pos=.2] {\scriptsize $2$} (t32);
                \draw (t3B) --node[above=-1pt,pos=.2] {\scriptsize $3$} (t33);
            \end{tikzpicture}
        };
        \node[dashed, left=of t2, circle, draw, minimum size=12pt, inner sep=2pt, outer sep=0pt] (t1) {
            \begin{tikzpicture}[solid]
                \begin{scope}[every node/.style={circle, draw, inner sep=2pt, outer sep=0pt,
                    node distance=20pt, minimum size=12pt}]
                    \node (t11) {1};
                    \node[right=of t11] (t12) {2};
                \end{scope}
                \draw[-stealth'] (t11) --node[above] {$A$} (t12);
            \end{tikzpicture}
        };
        \node[dashed, right=6 of t2, circle, draw, minimum size=12pt, inner sep=2pt, outer sep=0pt] (t4) {
            \begin{tikzpicture}[solid]
                \begin{scope}[every node/.style={circle, draw, inner sep=2pt, outer sep=0pt,
                    node distance=20pt, minimum size=12pt}]
                    \node (t42) {2};
                    \node[right=of t42] (t43) {3};
                \end{scope}
                \draw[-stealth'] (t42) --node[above] {$A$} (t43);
            \end{tikzpicture}
        };
        \node[dashed, below=of t3, circle, draw, minimum size=12pt, inner sep=2pt, outer sep=0pt] (t5) {
            \begin{tikzpicture}[solid]
                \begin{scope}[every node/.style={circle, draw, inner sep=2pt, outer sep=0pt,
                    node distance=20pt, minimum size=12pt}]
                    \node (t54) {4};
                    \node[right=of t54] (t52) {2};
                \end{scope}
                \draw[-stealth'] (t54) --node[above] {$A$} (t52);
            \end{tikzpicture}
        };
        \draw (tS) -- (t1);
        \draw (tS) -- (t2);
        \draw (tS) -- (t3);
        \draw (tS) -- (t4);
        \draw (t3) -- (t5);
    \end{tikzpicture}
    };
    \node[left=0pt of dt.north west] {$\dt(G):$};
\end{tikzpicture}

%% file: figures/valGExample.tex
\begin{tikzpicture}
    \begin{scope}[every node/.style={circle, draw, inner sep=0pt, outer sep=0pt, minimum size=12pt,
        node distance=25pt}]
        \node (s1) {$1$};
        \node[below=of s1] (A4) {$4$};
        \node[below=of A4] (A5) {$5$};
        \node[right=of A4] (s2) {$2$};
        \node[above=of s2] (A7) {$7$};
        \node[above=of A7] (B6) {$6$};
        \node[right=of A7] (s3) {$3$};
        \node[right=of s2] (A8) {$8$};
    \end{scope}
    \begin{scope}[every path/.style={-stealth'}]
        \draw (s1) edge[bend right=75] node[left=-2pt] {$a$} (A5);
        \draw (s1) edge node[left=-2pt] {$a$} (A4);
        \draw (A4) edge node[above=-2pt] {$a$} (s2);
        \draw (A5) edge node[above left=-4pt] {$a$} (s2);
        \draw (s2) --node[left=-2pt] {$a$} (A7);
        \draw (A7) --node[left=-2pt] {$a$} (B6);
        \draw (s1) --node[above left=-4pt] {$b$} (B6);
        \draw (s3) --node[above right=-4pt] {$b$} (B6);
        \draw (s2) --node[above=-2pt] {$a$} (A8);
        \draw (A8) --node[left=-2pt] {$a$} (s3);
    \end{scope}
\end{tikzpicture}

%% file: figures/example_stringgraph.tex
\begin{tikzpicture}
    \node (string) {String: \textit{abca}};
    \begin{scope}[every node/.style={circle, draw, inner sep=0pt, node distance=20pt,
        minimum size=7pt}]
        \node[draw=none, right=.25 of string] {$\Rightarrow$};
        \node[right=1 of string, fill=black] (1) {};
        \node[right=of 1] (2) {};
        \node[right=of 2] (3) {};
        \node[right=of 3] (4) {};
        \node[right=of 4,fill=black] (5) {};
    \end{scope}
    \node[below=0pt of 1] {\scriptsize 1};
    \node[below=0pt of 5] {\scriptsize 2};

    \begin{scope}[every path/.style={-stealth'}]
        \draw (1) -- node[above] {$a$} (2);
        \draw (2) -- node[above] {$b$} (3);
        \draw (3) -- node[above] {$c$} (4);
        \draw (4) -- node[above] {$a$} (5);
    \end{scope}
\end{tikzpicture}

%% file: figures/example_treegraph.tex
\begin{tikzpicture}
    \node (tree) {Tree:};
    \begin{scope}[every node/.style={inner sep=0pt, node distance=20pt, minimum size=7pt}]
        \node[right=1 of tree] (f) {$f$};
        \node[below left=of f] (a1) {$a$};
        \node[below right=of f] (g) {$g$};
        \node[below left=of g] (a2) {$a$};
        \node[below right=of g] (b) {$b$};
    \end{scope}
    \draw (f) -- (a1);
    \draw (f) -- (g);
    \draw (g) -- (a2);
    \draw (g) -- (b);
    
    \begin{scope}[every node/.style={draw, inner sep=0pt, node distance=20pt,
        minimum size=7pt}]
        \node[circle, above right=1 and 4 of f,fill=black] (1) {};
        \node[rectangle, below=of 1,minimum size=12pt] (tf) {$f$};
        \node[circle, below left=of tf] (2) {};
        \node[circle, below right=of tf] (3) {};
        \node[rectangle, below=of 2,minimum size=12pt] (ta1) {$a$};
        \node[rectangle, below=of 3,minimum size=12pt] (tg) {$g$};
        \node[circle, below left=of tg] (4) {};
        \node[circle, below right=of tg] (5) {};
        \node[rectangle, below=of 4,minimum size=12pt] (ta2) {$a$};
        \node[rectangle, below=of 5,minimum size=12pt] (tb) {$b$};
    \end{scope}
    \node[right=of g] {$\Rightarrow$};

    \draw (tf) --node[pos=.2,left=-1pt] {\scriptsize $1$} (1);
    \draw (tf) --node[pos=.2,above=-1pt] {\scriptsize $2$} (2);
    \draw (tf) --node[pos=.2,above=-1pt] {\scriptsize $3$} (3);
    \draw (tg) --node[pos=.2,left=-1pt] {\scriptsize $1$}  (3);
    \draw (tg) --node[pos=.2,above=-1pt] {\scriptsize $2$} (4);
    \draw (tg) --node[pos=.2,above=-1pt] {\scriptsize $3$} (5);
    \draw (2) -- (ta1);
    \draw (4) -- (ta2);
    \draw (5) -- (tb);
\end{tikzpicture}

%% file: graphrepair.tex
\section{GraphRePair}
We present our generalization to graphs of the RePair compression scheme.
\subsection{RePair Compression Scheme for Strings and Trees}\label{sss:StringTreeRePair}
%RePair compression was first introduced for strings by Larsson and Moffat~\cite{Larsson00_rePair}
%and was later generalized to trees by Lohrey et al.~\cite{Lohrey13_treeRePair}. 
%In contrast to other approximation for a smallest grammar, 
%RePair works globally instead of common window-based approaches: it replaces the
%most frequent pair of adjacent symbols (``digram'') by a new nonterminal, and recursively repeats
%this until no digram appears more than once. Linear time is achieved through a highly sophisticated
%data structure. As in TreeRePair~\cite{Lohrey13_treeRePair}, 
%in a final ``pruning'' phase, we remove productions that do not contribute to the compression. 
Let us first explain the classical RePair compressor for strings and trees. The RePair compression
scheme is an approximation algorithm for a smallest context-free grammar for a given string. The
smallest grammar problem is the problem of deciding, given a string $w$ and a number $k \in \N$,
whether there exists a straight-line grammar $G$ of size at most $k$ generating $w$. It
was shown by Charikar et al.~\cite{Charikar05_smallest} that the smallest grammar problem is
\NP-complete. Hence, already for string-graphs it follows that finding a smallest grammar is
\NP-hard. In a string, a digram consists of two consecutive symbols. The idea of Larsson and
Moffat's RePair Compression scheme~\cite{Larsson00_rePair} is to repeatedly replace all
(non-overlapping) occurrences of a most frequent digram by a new nonterminal. The process ends when
no repeating digrams occur. Consider as an example the string 
\[w = \mathit{abcabcabc}.\] 
It contains occurrences of the digrams \textit{ab} (3 times), \textit{bc} (3 times), and \textit{ca}
(2 times). If \textit{ab} is replaced by $A$ then we obtain \textit{AcAcAc}. If now
\textit{Ac} is replaced by $B$, then we obtain this grammar
\[
    \begin{array}{lcl}
        S &\rightarrow& \mathit{BBB} \\
        B &\rightarrow& \mathit{Ac} \\
        A &\rightarrow& \mathit{ab}
    \end{array}
%$\{ S\to BBB, B\to Ac, A\to ab \}$.
\]
Note that the original string $w$ has size $|w| = 9$, whereas the grammar has size $7$ (the size of
a grammar is the sum of the sizes of its right-hand sides). Note that overlapping occurrences only
exist for digrams of the form $aa$.

To compute in linear time such a grammar from a given string requires a set of carefully designed
data structures. The input string is represented as a doubly linked list. Additionally a list of
active digrams (digrams that occur at least twice) is maintained. Every entry in the list of active
digrams points to an entry in a priority queue $Q$ of length $\sqrt{n}$ (with $n$ being the length
of the input string) containing doubly linked lists. The list with priority $i$ in $Q$ contains
every digram that occurs $i$ times, the last list contains every digram occurring $\sqrt{n}$ or more
times. The list items also contain pointers to the first occurrence of the respective digram. This
queue is used to find the most frequent digram in constant time. Larsson and
Moffat~\cite{Larsson00_rePair} prove that $\sqrt{n}$-length guarantees a linear runtime for the
complete algorithm. Basically, in a string of length $n$ the most occurrences a digram can have is
$n/2$ and in that case no other digram can occur as often in the string. There may, in general, be
more than one digram occurring more than $\sqrt{n}$ times, but the last list can not have more than
$\sqrt{n}$ entries, and the frequency of the most frequent digram is monotonely decreasing. This
overall leads to a linear running time. All these data structures are updated whenever an occurrence
is removed. Consider removing one occurrence of \textit{ab} in the example above: when doing so, one
occurrence of \textit{bc} and possibly \textit{ac} need to be removed from the list. On the other
hand, new occurrences of \textit{Ac} and possibly \textit{cA} are created.

An even smaller grammar than the one above can be obtained through \emph{pruning}, which removes
nonterminals that are referenced only once, i.e., pruning would remove the nonterminal $A$ in the
grammar above, so that the $B$-rule becomes $B \rightarrow \mathit{abc}$. This reduces the size of
the grammar from $7$ to $6$. Note that pruning can never increase the size of the grammar (but may
not decrease it).
%Pruning is applied as final stage in TreeRePair~\cite{Lohrey13_treeRePair}.

RePair was generalized to trees by Lohrey et al.~\cite{Lohrey13_treeRePair}. Here a digram consists
of two nodes, and the $i$-edge between them, with $i$ meaning that the second node is the $i$-th
child of the first node, denoted by $(a,i,b)$. Note that overlapping
occurrences can only happen for digrams of the form $(a,i,a)$. A
digram $(a,i,b)$ has, in a binary tree, at most three ``dangling edges''. Dangling edges in
context-free tree grammars are represented by \emph{parameters} of the form $y_1,\ldots,y_k$. The
number $k$ is the \emph{rank} of the rule (digram). E.g. the $A$-rule in
Figure~\ref{fig:treeRePair_example} represents a digram of rank 1, while the $B$-rule represents a
digram of rank 3. The rank of a grammar is defined as the maximal rank of a rule occurring in the
grammar. Parameters have a similar role to external nodes in HR grammars, as they indicate
how the right-hand side is connected with the tree during a derivation. Similarly to string grammars, the size
of a tree grammar is the sum of sizes of its right-hand sides. Note that neither of the two
replacements shown in Figure~\ref{fig:treeRePair_example} make the resulting grammar smaller and
would thus not be considered \emph{contributing}: the digram represented by the $B$-rule only occurs
once, therefore replacing it makes the grammar larger. Replacing the digram represented by the
$A$-rule reduces the size of the tree by $2$, but the right-hand side of the rule has size two as
well, thus leading to a grammar that has the same size as the original graph. Were the digram to
occur once more however, the replacement would become contributing.
%
%This addition makes rules in tree grammars more expensive than in string grammars,
%which manifests in the pruning step. First, every nonterminal only referenced once can be removed
%just like before. Following that, however, it is now possible that nonterminals which are referenced
%more than once still do not contribute to the compression, but make the grammar larger. To see this,
%consider~\cite[Example 8]{Lohrey13_treeRePair}: in the following straight-line context-free tree
%grammar both nonterminals $A$ and $B$ are not contributing to the compression, but appear twice
%each:
%\begin{align*}
%    S &\rightarrow f(A(a,a), B(A(a,a))) \\
%    A(y_1,y_2) &\rightarrow f(B(y_1), y_2) \\
%    B(y_1) &\rightarrow f(y_1,a)
%\end{align*}
%The pruning step would remove them, but the order is important: removing $A$ first, turns $B$ into a
%contributing nonterminal and thus yields a compressed grammar. If on the other hand $B$ is removed
%first, $A$ remains non-contributing and is also removed, essentially yielding a grammar without
%compression. For all details regarding this example see~\cite{Lohrey13_treeRePair}. Thus the order
%in which the nonterminals are removed may also affect the quality of the pruning. Finding an optimal
%order for this removal is a hard problem.
The rank of a grammar also has an impact on further
algorithms run on it, see e.g.~\cite{Lohrey06_SLPTreeComplexity,DBLP:journals/jcss/LohreyMS12}.
Keeping the rank small is thus desirable. Therefore, TreeRePair has a user-defined ``maxRank''
parameter.
\subsection{RePair on Graphs}
The first step in generalizing RePair to graphs, is defining the notion of a digram in a graph. For
simple graphs two options come up naturally: two neighboring nodes and the edge between them, or two
edges with a common node. The first option is not viable, as it does not allow to compress some very
basic graphs: consider a cycle $C_n$ consisting of $n$ consecutive edges (and $n$ nodes).
Replacing a digram of the first kind does not remove nodes nor edges, and only relabels edges by a
nonterminal. Thus no compression is achieved. We therefore use the latter option.
\begin{definition}\label{def:digram}
    A \emph{digram} over $\Sigma$ is a hypergraph $d \in \HGR(\Sigma)$, with $E_d = \{e_1, e_2\}$ such
    that 
    \begin{enumerate}
        \item for all $v \in V_d$, $v \in \att_d(e_1)$ or $v \in \att_d(e_2)$,
        \item there exists a $v \in V_d$ such that $v \in \att_d(e_1)$ and $v \in \att_d(e_2)$, and
        \item $\ext_d \neq \varepsilon$.
    \end{enumerate}
\end{definition}
\begin{figure}[!t]
    \centering
    \input{figures/all_digrams}
\vspace*{-4mm}
    \caption{Every possible digram using two unlabeled, undirected edges.}
    \label{fig:all_digrams}
\end{figure}
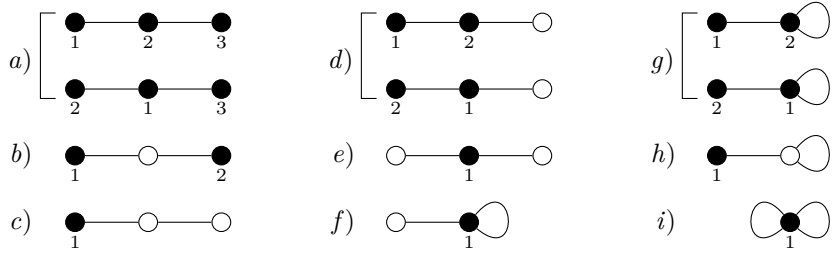
%As we only replace digrams where the two edges have at least one node in common, we assume
%$\rank(\sigma) \geq 1$ for all $\sigma\in\Sigma$ from now on. This is also the reason for the last
%requirement ($\ext_d \neq \varepsilon$): while we could replace such digrams by a rank-$0$ edge,
%these could then not be recursively replaced. 
Every possible digram over undirected, unlabeled edges is shown in Figure~\ref{fig:all_digrams}.
Note that we grouped some of the digrams. This is because, one can represent one by the other by using a
different order in the $\att$-relation of the nonterminal edge. Another example for digrams are the
right-hand sides of the two $A$-rules in Figure~\ref{fig:size_example}. Note that both grammars in
the figure generate the graph on the left. However, they differ in size: the grammar in the middle
has size 12, while the grammar on the right has size 9 (recall that simple edges have size~1).

As mentioned before, RePair replaces a digram that has the largest number of non-overlapping occurrences.
\begin{definition}\label{def:occurrence}
    Let $g,d \in \HGR(\Sigma)$ such that $d$ is a digram. Let $e_1^d,e_2^d$ be the two
    edges of $d$. Let $o = \{e_1,e_2\} \subseteq E_g$ and let $V_o$ be the set of
    nodes incident with edges in $o$. Then $o$ is an \emph{occurrence} of $d$ in 
    $g$, if there is a bijection $\psi: V_o \rightarrow V_d$ such that for $i \in \{1,2\}$ and $v
    \in V_o$
    \begin{enumerate}
        \item $\psi(v) \in \att_d(e_i^d)$ if and only if $v \in \att_g(e_i)$,
        \item $\lbl_d(e_i^d) = \lbl_g(e_i)$, and
        \item $\psi(v) \in \ext_d$ if and only if $v \in \att_g(e)$ for some $e \in E_g\setminus
            \{o\}$.
    \end{enumerate}
\end{definition}
The first two conditions of this definition ensure that the two edges of an occurrence induce a
graph isomorphic to $d$. The third condition requires that every external node of $d$ is mapped to a
node in $g$ that is incident with at least one other edge. Thus, the edges marked in the graph on
the left of Figure~\ref{fig:size_example} constitute an occurrence of the digram $(b)$ of
Figure~\ref{fig:all_digrams}, which is the same as the one of the $A$-rule of the right grammar, but
not an occurrence of the digram of the $A$-rule in the middle grammar (digram $(a)$ of
Figure~\ref{fig:all_digrams}). We call the nodes in $V_o$ that are mapped to external nodes of $d$,
\emph{attachment nodes} of $o$, and the ones mapped to internal nodes, \emph{removal nodes} of $o$.
Two occurrences $o_1,o_2$ of the same digram $d$ are called \emph{overlapping} if $o_1 \cap o_2 \neq
\emptyset$. Otherwise they are \emph{non-overlapping}. If there are at least two non-overlapping
occurrences of $d$ in a graph $g$, we call $d$ an \emph{active digram}.

\begin{algorithm}
    \begin{algorithmic}[1]
        \Require{Graph $g = (V,E,\att,\lbl,\ext)$}
        \Ensure{Grammar $G$ with $\val(G) \cong g$ and $|G'| \leq |g|$}
        \State $N = P = \emptyset$
        \State $S \gets g$
        \State $L(d) \gets$ list of non-overlapping occurrences for every digram $d$ appearing in $g$\label{item:traverse}
        \While{$|L(d)| > 1$ for at least one digram $d$}
        \State Select a most frequent digram $\mathit{mfd}$\label{item:repeat}
        \State $A \gets$ fresh nonterminal of rank $\rank(\mathit{mfd})$
        \State In the graph $S$, replace every occurrence $o$ in $L(\mathit{mfd})$ by a new $A$-edge
        \State $N \gets N \cup \{A\}$
        \State $P \gets P \cup \{A \rightarrow \mathit{mfd}\}$
        \State Update the occurrence lists\label{item:update}
        \EndWhile
%        \State Prune the grammar\label{item:prune}
        \State \textbf{return} $G = (N, P, S)$
    \end{algorithmic}
    \caption{Main replacement-loop of the \graphrepair algorithm}
    \label{alg:repair}
\end{algorithm}
Let $A$ be a symbol of rank $k$ and $d$ a digram of rank $k$. The \emph{replacement of an occurrence
$o$ of $d$ in $g$ by $A$} is the graph obtained from $g$ by removing the edges in $o$ from $g$,
removing the removal nodes of $o$, and adding an edge labeled $A$ that is attached to the attachment nodes
of $o$, in such a way that applying the rule $A \rightarrow d$ yields the original graph.
Consider again Figure~\ref{fig:size_example}: the start graph of the right grammar is the 
replacement of the shaded occurrence of $d$ in the left graph by $A$
(where $d$ is $\rhs(A)$ of the grammar on the right).
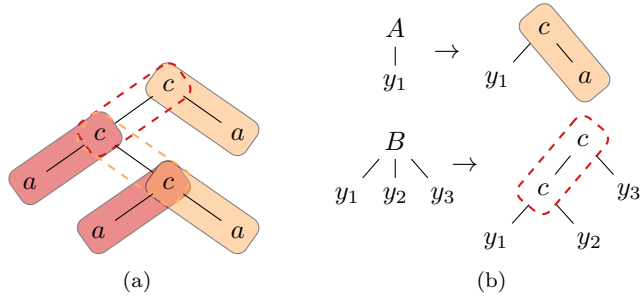
\begin{figure}[!t]
    \centering
    \subfloat[]{\input{figures/treeRePair_example}\label{fig:treeRePair_example_tree}}\qquad
    \subfloat[]{\input{figures/treeRePair_example_productions}\label{fig:treeRePair_example_productions}}
    \caption{Different digrams TreeRePair considers in a tree and their replacement rules.}
    \label{fig:treeRePair_example}
\end{figure}
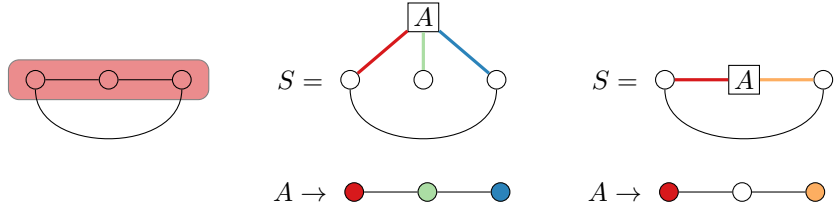
\begin{figure*}[!t]
    \centering
    \input{figures/size_example}
    \caption{Two ways of replacing a pair of edges by a nonterminal. Only the one on the right is
    considered by \graphrepair in this case.}
    \label{fig:size_example}
\end{figure*}
\subsection{The Algorithm}
Given a graph $g$, \graphrepair first performs the steps given in Algorithm~\ref{alg:repair}.
%\begin{enumerate}
%    \item Let $G = (N,P,S)$ be a grammar with $N = P = \emptyset$ and $S = g$.
%    \item Determine a list $L(d)$ of non-overlapping occurrences for every digram $d$ appearing in
%        $g$.\label{item:traverse}
%    \item Select a most frequent digram $\mathit{mfd}$.\label{item:repeat}
%    \item Let $A$ be a fresh nonterminal of rank $\rank(\mathit{mfd})$. In the graph $S$, replace
%        every occurrence $o$ in $L(\mathit{mfd})$ by a new $A$-edge.
%    \item Let $N = N \cup \{A\}$ and $P = P \cup \{A \rightarrow \mathit{mfd}\}$.
%    \item Update the occurrence lists.\label{item:update}
%%    \item If there is an active digram in $S$: repeat from Step~\ref{item:repeat}.
%    \item Prune the grammar.\label{item:prune}
%\end{enumerate}
As an additional step after the loop finishes, 
we connect the disconnected components of the graph by
virtual edges and run the algorithm again before pruning (and in a final step remove the
virtual edges from the grammar). This improves the compression on graphs
with disconnected components. We now provide more details on the steps in lines~\ref{item:traverse}
and~\ref{item:update}, and on the pruning step.
\subsubsection{Counting Occurrences (Step~\ref{item:traverse})}
We aim to find a set of non-overlapping occurrences of a digram that occurs in $g$, that is of
maximal size. As stated in the Introduction, this can be solved in $O(|E|^4)$ time. This is done by
reducing it to maximum matching: Let $g$ be a graph with $n$ nodes and $m$ edges, and $d$ a digram. The first
step is to compute a set $O$ containing all occurrences of $d$ in $g$, including overlapping ones.
This already takes $O(m^2)$ time, as $|O| \in O(m^2)$. To see this, consider a ``star''; a graph with
one central node connected to $m$ other nodes, which in turn have no neighbors except the central
node. There are $m(m-1)/2$ pairs of edges and thus as many overlapping occurrences of
the same digram. We encode this set of occurrences into a graph $g_O$ such that every
occurrence $\{e_1,e_2\}$ in $O$ is represented by an edge from $e_1$ to $e_2$ in $g_O$. Thus, $g_O$
potentially has a node for every edge in $g$, and an edge for every occurrence in $O$. It can be
shown that a maximum matching (i.e., a maximal cardinality set of edges $X$ such that no two edges
$e_1,e_2 \in X$ have a common node) on $g_O$ corresponds to a maximal non-overlapping
subset of $O$. Computing a maximum matching in graphs can be done in $O(|V|^2|E|)$ by using,
e.g.~the Blossom-algorithm~\cite{blossom}. As $g_O$ has $O(m)$ nodes, and $|O|$ edges, the total
running time is in $O(m^2|O|) = O(m^4)$.

Doing the above for one digram is already prohibitively expensive. Thus we
approximate. Let $\ord$ be an order on the nodes of $g$. We traverse the nodes of $g$ in this
order, and at every node iterate through occurrences centered around this node, as detailed in the
next section below.
The node order \ord heavily influences the compression behavior. Consider the graph in
Figure~\ref{fig:counting_example}. We want to find the non-overlapping occurrences of the first
digram in Figure~\ref{fig:all_digrams}. Note that all three nodes are external, that is, we are
looking for three nodes $u,v,w$ such that $u$ has an edge to $v$, $v$ has an edge to $w$, and all
three nodes also have edges to other nodes (different from $u,v,w$).
%We generally mark external nodes by filling them in a
%color. For this example the order is not important, therefore we used black for all of them.
%Note the nodes filled in black: they are called \emph{external nodes}
%and indicate where nodes of a production are merged with nodes incident with a nonterminal edge.
%Essentially this means here that every node in an occurrence needs to be incident with edges
%\emph{other} than the two edges of the occurrence. 
Figure~\ref{fig:counting_example_center} shows
the non-overlapping occurrences found if we start in the central node of the graph. Using the DFS-type order
starting at a different node
given by the numbers in Figure~\ref{fig:counting_example_dfs}, three occurrences are determined.
Using the ``jumping'' order in Figure~\ref{fig:counting_example_opt}, a maximum set of four
non-overlapping occurrences is found. 
Note that for strings and trees, maximum sets of non-overlapping occurrences can be
obtained in linear time using left-to-right and post order, respectively, and assigning occurrences
in a greedy way.
%
%It should be noted that a different digram
%also occurs four times and leads to a more succinct representation:
%%would be chosen for this graph by our algorithm,
%the one shown on the right of Figure~\ref{fig:grammar_example}.
%%Here the two colored nodes are the external nodes.
%The original graph in
%Figure~\ref{fig:counting_example} can be obtained by replacing every $A$-labeled hyperedge
%in the left graph of Figure~\ref{fig:grammar_example} by the digram,
%while merging the colored nodes according to the colors of the hyperedge.
%Ordering the edges from ``left'' to ``right'' in this drawing yields the full derivation given in
%Figure~\ref{fig:derivation_order_example}.
\begin{figure*}[!t]
    \centering
    \subfloat[]{\input{figures/counting_example_center}\label{fig:counting_example_center}}
    \subfloat[]{\input{figures/counting_example_dfs}\label{fig:counting_example_dfs}}
    \subfloat[]{\input{figures/counting_example_opt}\label{fig:counting_example_opt}}
%    \subfloat[]{\input{figures/counting_example_digram}\label{fig:counting_example_digram}}
    \caption{Three different traversals visiting the nodes in the numbered order to find occurrences
        of digram $(a)$ of Figure~\ref{fig:all_digrams}.
        %\protect\subref{fig:counting_example_digram}.
    The occurrences found
    for each traversal are marked.
    %by differently colored boxes.
}
    \label{fig:counting_example}
\end{figure*}
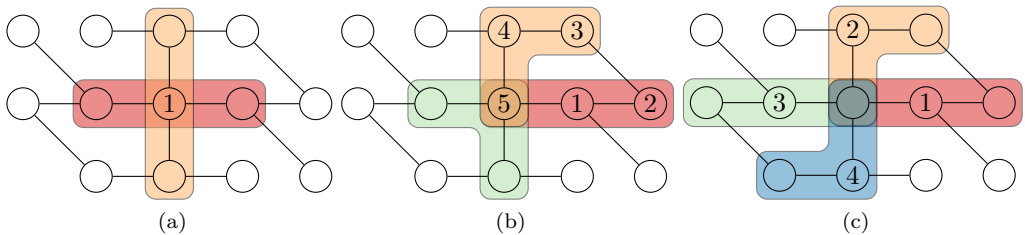
\begin{figure}[!t]
    \centering
    \input{figures/grammar_example}
    \caption{Example of a hyperedge replacement grammar.}
    \label{fig:grammar_example}
\end{figure}
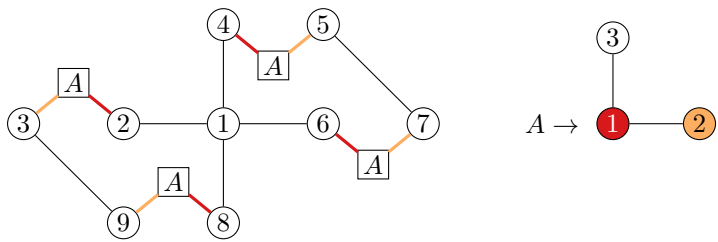
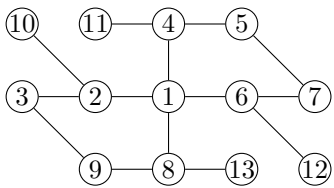
\begin{figure}[!t]
    \centering
    \input{figures/graph_derivation_order_example}
    \caption{Unique result of the derivation of the grammar in Figure~\ref{fig:grammar_example} when
    ordering the nonterminal edges by their $\att$-relation as described in Section~\ref{sss:nodeIDs}.}
    \label{fig:derivation_order_example}
\end{figure}
Our implementation offers a choice of four different orders, as explained in
Section~\ref{ss2:node_order}. Implementation details of this step are given in 
%While doing
%this we not only count how many we find for every digram, but also keep track of the occurrences,
%active digrams, and their frequencies using several data structures. Recall that there are $O(k^2)$
%possible occurrences centered on a node $v$ of degree $k$. We only want to visit $O(k)$ of them. We
%describe how we deal with this issue in 
Section~\ref{ss2:compression}.
\subsubsection{Updating Occurrence Lists (Step~\ref{item:update})}\label{ss2:updating}
Let $o$ be an occurrence of $d$ that is being replaced and let $F$ be the set of edges in $g$ that
are incident with the attachment nodes of $o$. Removing $o$ from the graph can only affect the
occurrence lists of digrams that have occurrences using edges in $F$. In particular, for the two
edges in $o$ ($e_1$ and $e_2$) we need to remove every occurrence that either $e_1$ or $e_2$
participate in from the corresponding occurrence list. To do this efficiently, we store a
list of digrams with every edge for which the edge appears in an occurrence. After the replacement
let $e'$ be the new $A$-labeled nonterminal edge in $g$. Then every pair of edges $\{e', e\}, e \in
F$ is an occurrence of a digram, and is thus inserted into the appropriate occurrence list (and the
frequency counts are updated accordingly).
The last step has again complexity issues. Let $k$ be the sum of degrees of all attachment nodes of
$o$. Then there are $O(k)$ pairs of edges to be considered as occurrences with the new nonterminal
edge. This is not a problem in itself, but consider the following situation: let there be an
attachment node $v$ with degree $k$. Further, let every one of the $k$ edges around $v$ be part of a
distinct occurrence of the digram $d$ being replaced. As explained, when replacing one of these
occurrences, the other $k-1$ edges are considered as occurrences with the new nonterminal. Now
however, when replacing the next one, the remaining $k-2$ edges have to be considered again. Thus,
during all the replacement steps, we would again need to consider $O(k^2)$ occurrences. To solve
this issue, we consider first the case of a graph without edge labels and directions. Let node $v$ have
degree $k$, i.e., there are $O(k^2)$ edge pairs that are occurrences of some digram. After
inserting one of them into the occurrence list, we take the two edges involved out of
further consideration, because every other occurrence using one of them would overlap with this one.
Let $E$ be the edges incident with $v$ that have not been added into occurrence lists yet. We
partition $E$ into two sets $E_1 =  \{e_1,\ldots, e_n\}$ and $E_2 = \{f_1,\ldots,f_{m}\}$ where $m-n
\in \{0,1\}$. We then add $\Occ(E_1,E_2) = \{\{e_i,f_i\}\mid 1 \leq i \leq n\}$ as the occurrences
around $v$ to the list. Note that only if all occurrences around $v$ are occurrences of the
\emph{same} digram, this procedure guarantees to produce a maximum non-overlapping set of
occurrences around $v$.

From here, adding labels (or directions, which can be viewed as labels) is straightforward. For two
labels $\sigma_1$ and $\sigma_2$ let $E_{\sigma_1,\sigma_2}(v)$ be the set of edges incident with $v$
labeled $\sigma_1$ and not yet added to an occurrence list for a digram with an edge labeled
$\sigma_2$. Then for distinct symbols $\sigma_1$ and $\sigma_2$ add the occurrences
$\Occ(E_{\sigma_1,\sigma_2}(v), E_{\sigma_2,\sigma_1}(v))$ and for every $\sigma$ add the
occurrences where both edges have the same symbol by splitting $E_{\sigma,\sigma}(v)$ as above. This
takes $O(|\Sigma|^2)$ time (where $|\Sigma|$ is expected to be small).

The above method works during the initial traversal, however every time an occurrence of a digram is
replaced, the occurrence lists of affected digrams need to be updated. After inserting the new
nonterminal edge $e'$ we need to select an edge $e$ from the set of neighboring edges $F$ to make an
occurrence of a new digram. This can again lead to a situation where time quadratic in the node
degree is necessary, unless this selection is done in constant time. Our implementation does this by
storing a list of available edges for every pair of edge labels attached to every node of the graph.
For every edge label the first edge $e$ in the respective list is selected to create the occurrence
$\{e',e\}$. This takes $O(|\Sigma|)$ time.
\subsubsection{Pruning}
%(Step~\ref{item:prune})
%The grammar resulting from the replacement step can still be reduced. Due to the size definition,
%the most frequent digram might actually result in a grammar that is bigger than the graph we started
%with. Replacing it may however still be necessary to make further replacement steps possible which
%then compress the graph.
%
%%The pruning step is a post-processing step on the grammar after the replacement step finished. It
%%simply 
As mentioned in Section~\ref{sss:StringTreeRePair}, pruning can reduce the size of a grammar. For
string grammars, pruning removes every nonterminal that is referenced only once in the grammar. The
presence of parameters for tree grammars, or external nodes for HR grammars, complicates the pruning
step. While, every nonterminal only referenced once can be removed just like before, it is then
possible to still have nonterminals which are referenced more than once but do not contribute to the
compression (this is \emph{not} the case for strings). To see this, consider the following tree
grammar (from~\cite[Example 8]{Lohrey13_treeRePair}):
\begin{align*}
    S &\rightarrow f(A(a,a), B(A(a,a))) \\
    A(y_1,y_2) &\rightarrow f(B(y_1), y_2) \\
    B(y_1) &\rightarrow f(y_1,a)
\end{align*}
In this grammar $A$ and $B$ are each referenced twice, but do not ``contribute'' to
compression. To see this, note that the grammar above has size 12 and consider the grammars obtained
by removing either nonterminal (left: after removing $A$, right: after removing $B$):
\begin{align*}
    S &\rightarrow f(f(B(a),a), B(f(B(a),a))) & S &\rightarrow f(A(a,a), f(A(a,a),a))\\
    B(y_1) &\rightarrow f(y_1,a) & A(y_1,y_2) &\rightarrow f(f(y_1,a), y_2)
\end{align*}
These grammars now have sizes 11 and 12, respectively. Clearly, both rules where therefore not
contributing in the full grammar above, as neither grammar is now larger than before. However,
removing $A$ first did make the grammar \emph{smaller} than it was before. Further, if we were to
remove one more nonterminal, we would in both cases get the original tree
\[f(f(f(a,a),a),f(f(f(a,a),a),a)),\]
which is still of size $12$. This shows that the $A$ rule was still not contributing after removal
of the $B$-rule, but the $B$-rule turns into a contributing rule if the $A$-rule is removed first.
Thus the order in which the nonterminals are removed may also affect the quality of the pruning.
Finding an optimal order is a complex optimization problem as mentioned
in~\cite[Section~3.2]{Lohrey13_treeRePair}.

Let us consider pruning for SL-HR grammars. For a nonterminal $A$ of rank $n$ we define
$\HG(A) = (\{v_1,\ldots,v_n\}, \{e\}, \lbl(e) = A, \att(e) = v_1\cdots v_n, \ext=v_1\cdots v_n)$.
The \emph{contribution of $A$} is defined as
\[\sav(A) = |\refs(A)| \cdot (|\rhs(A)| - |\HG(A)|) - |\rhs(A)|,\]
where $\refs(A) = |\{e \in E_S \mid \lbl(e) = A\}| + \sum_{B \in N}|\{e \in E_{\rhs(B)} \mid \lbl(e)
= A\}|$ is the number of edges labeled $A$ in the grammar. The contribution of $A$ counts by how
much the size of the grammar changes when every instance of the nonterminal is derived, i.e., it
measures how much $A$ contributes towards compression. If $\sav(A)>0$ then we say that \emph{$A$
contributes towards compression}. The grammar in Figure~\ref{fig:grammar_example} represents the
graph of Figure~\ref{fig:counting_example}. Here, the $A$-rule has $\sav(A)= 4 \cdot (5 - 3) - 5 =
3$ and thus contributes to the compression. The reader may verify that the sizes of this grammar and
the graph (given in Figure~\ref{fig:derivation_order_example}, with the IDs assigned as explained in
Section~\ref{sss:nodeIDs}) differ by exactly three. Note that, as we remove rules, the contribution
of other nonterminals might change as edges are added or deleted. Therefore, the effectiveness of
pruning depends on the order in which the nonterminals are considered. For TreeRePair, a bottom-up
hierarchical order works well in practice. We use a similar approach. First every nonterminal $A$
with $\refs(A) = 1$ is removed, because, by definition, they do not contribute towards compression.
The order does not matter for this step. To remove $A$ we apply its rule to each $A$-edge in the
grammar and remove the $A$-rule. Then we traverse the nonterminals in bottom-up $\SLOrder$-order
(see Preliminaries), removing each nonterminal with $\sav(A) \leq 0$.
%If $\sav(A) \leq 0$ then every $A$-labeled edge in
%the grammar is derived and the $A$-rule is removed.
%\begin{lemma}
%    For a given graph $g$ and a node order $\ord$, \graphrepair computes a compressed representation
%    $G$ of $g$ in $O(?)$ time.
%    \label{lem:runtime}
%\end{lemma}
%\begin{proof}
%    \begin{itemize}
%        \item Preferably $?$ would be linear in $|g|$.
%        \item This may need additional assumptions about the time needed to find occurrences around
%            a center node.
%        \item Also, it is possible that the maximum rank of a digram would be in the exponent.
%    \end{itemize}
%\end{proof}
\subsection{Relation to Compression of Strings and Trees}\label{ss2:StringTreeRelation}
Our algorithm is a generalization of previous RePair-variants for
strings~\cite{Larsson00_rePair} and trees~\cite{Lohrey13_treeRePair}. Let $w$ be a string, and
$d_1,\ldots,d_n$ be the digrams that string-RePair would replace during its run on $w$. If we run
\graphrepair on $\sgraph(w)$, and use a left-to-right node order, then the digrams
$d_1',\ldots,d_n'$ that \graphrepair replaces are exactly $\sgraph(d_1)$ to $\sgraph(d_n)$.
Therefore, we obtain a graph-encoding of the same string-grammar that the original RePair
generates. The comparison with TreeRePair is less clear. TreeRePair compresses
ordered trees where the nodes, not the edges, are labeled from a finite alphabet. For \graphrepair
to compress trees in the same way, we need to represent the tree as a tree-graph as given in
Section~\ref{sss:stringTreeDef} and use a similar postOrder-traversal
as TreeRePair does. We can experimentally confirm that \graphrepair achieves
comparable (within the same order of magnitude) compression ratios as TreeRePair.

An SL tree grammar achieves compression by representing only once repeating \emph{connected
subgraphs} of the given tree. An SL-HR grammar is more general, because it can share repeating
\emph{disconnected} subgraphs. Does this allow, for an SL-HR grammar, to compress a tree more
effectively than any SL tree grammar? We show that this is \emph{not} the case by proving that every
SL-HR grammar $G$ that represents a tree (string), can be converted to an SL-HR grammar, that is
only slightly larger than $G$ and for which every right-hand side is a tree (string). This is not
trivial, as the expressive power of graph grammars is higher than the one of context-free string
grammars. For example, it is possible to generate the language $\{\sgraph(w) \mid w = a^nb^nc^n,
n\in \N\}$ using HR grammars, but the string-language represented here is not context-free. We call
an SL-HR grammar $G$ \emph{tree generating}, if $\val(G)$ is a tree-graph as defined in
Section~\ref{sss:stringTreeDef}. Similarly, $G$ is called \emph{string-generating}, if $\val(G)$ is
a string-graph. Recall that a path from a node $u$ to a node $v$ in a hypergraph is a tuple
$(e_1,\ldots,e_n)$ of edges, which connect $u$ to $v$, such that an edge of rank $k$ is considered
to have one source (the first attached node) and $k-1$ target nodes, and that we call a path
internal, if it uses only internal nodes (with the possible exception of $u$ and $v$). For example,
in Figure~\ref{fig:lineGraphExample} the right-hand side of $A$ has a path $(C,B)$ from the green to
the red node, but it is \emph{not} an internal path, because it crosses the orange node.
\begin{definition}
    Let $g$ be a graph with the set $X$ of external nodes. Then the \emph{line-structure} $\ls(g) =
    (X,E)$ is the directed, unlabeled graph such that for all $u,v \in X, (u,v) \in E$ if and only
    if there is an internal path from $u$ to $v$ in $\val(g)$.
    \label{def:lineGraph}
\end{definition}
For a nonterminal $A$ of an SL-HR grammar we denote $\ls(\rhs(A))$ by $\ls(A)$. For an example of a
line-structure, see Figure~\ref{fig:lineGraphExample}. Note that if a grammar $G$ is
tree-generating, $\ls(A)$ is always a rooted forest for every nonterminal $A$ in $G$ (otherwise
there would be a cycle in $\val(A)$, because $G$ is tree-generating and $A$ is useful). We refer to
this property by (LG). However, the right-hand side of $A$ may not be a tree-graph itself, as in the
$A$-rule in Figure~\ref{fig:lineGraphExample}. There are two ways in which $\rhs(A)$ for a rule in a
tree-generating grammar can not be a tree-graph: it can contain cycles, or nodes with multiple
parents. Both of these cases can only happen with nonterminal edges, however. The line-structure is
a tool to split these nonterminal edges, such that the resulting graphs are tree-graphs. 
\begin{figure}[!t]
    \centering
    \input{figures/line_graph_example}
    \caption{A set of tree-generating rules, the tree they represent, and their line-structures.}
    \label{fig:lineGraphExample}
\end{figure}
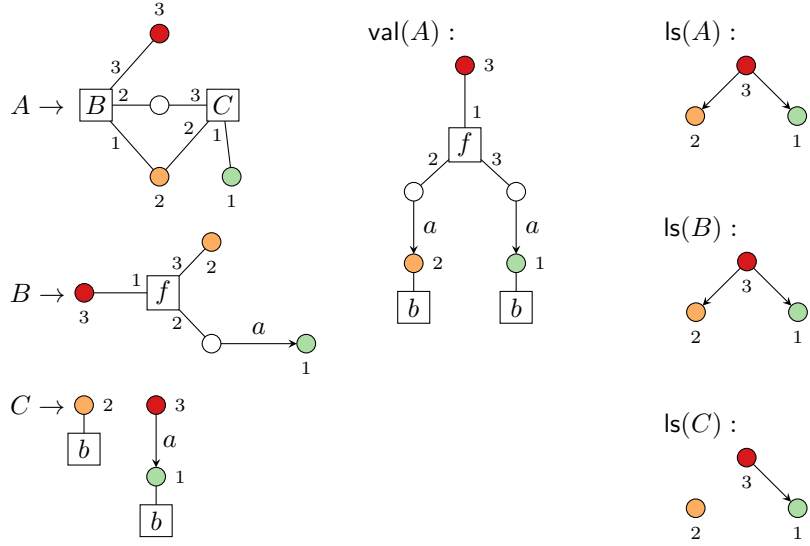
\begin{theorem}
    Let $G = (N, P, S)$ be a tree-generating SL-HR grammar. We can construct an SL-HR grammar $G' =
    (N', P', S')$
    such that
    \begin{enumerate}
        \item $\val(G') = \val(G)$,
        \item $|G'| \leq 2|G|$, and
        \item $\rhs_{G'}(A)$ is a tree-graph for every $A \in N' \cup \{S'\}$.
    \end{enumerate}
    \label{thm:treeGeneratingToTree}
\end{theorem}
\begin{proof}
We describe three transformations on $G$, which yield the desired grammar $G'$. In the first
transformation, the set of nonterminals remains unchanged. We use the
line-structure to order the external nodes in a top-down way, i.e., if $u$ is an ancestor of $v$
within the line-structure, then $u$ should be ordered before $v$. Let
$\ls(A) = (V,E)$ for some nonterminal $A$, and let $u,v \in V$. We define the partial order
$\prec_A$ such that $u \prec_A v$ if and only if there is a directed path from $u$ to $v$ in
$\ls(A)$. Due to the property (LG) mentioned above, $\prec_A$ is well-defined. Let $(A,g)$ be a
rule, and $u,v$ two of its external nodes, such that $u \prec_A v$, but $v$ appears before $u$ in
$\ext_g$. We reorder $\ext_g$ (and by extension, the attachment relation of every nonterminal edge
labeled $A$ in the entire grammar) such that they are sorted according to some (arbitrary, but
fixed) total order consistent with $\prec_A$. This step does not change the structure of the grammar
or affect its size. As an example, consider $\ls(X)$ for $X \in \{A,B,C\}$ in
Figure~\ref{fig:lineGraphExample}. In all three cases, the red external node, needs to come first to
be consistent with $\prec_X$, but it is currently third. The orange and green nodes can be ordered
arbitrarily after that. Permissible orders for $\ext$ thus are red $<$ orange $<$ green or red $<$
green $<$ orange (see Figure~\ref{fig:conversionExample} for the final grammar).

In the next two steps, additional nonterminals are (in general) added to the set of nonterminals.
We iterate over $N$ in a bottom-up fashion, i.e., in reverse $\SLOrder$-order. If we
encounter an $A \in N$ such that $\ls(A)$ has at least two connected components $g_1,\ldots,g_k$
then we add rules $(A_i, h_i)$ for $i \in [k]$, where $A_i$ is a new nonterminal of rank
$|V_{g_i}|$, and $h_i$ is a subgraph of $\rhs(A)$. Note that $\sum_{i\in [k]}A_i = \rank(A)$. The
set of nodes of $h_i$ is defined as 
\[U_i = \{u \in \rhs(A) \mid \exists v \in V_{g_i}: \text{ there is a path from } v\text{ to
}u\text{ in } \rhs(A)\},\]
i.e., $U_i$ contains the nodes of $\rhs(A)$ that can be reached from the external nodes in $g_i$.
The set of edges of $h_i$ contains every edge of $\rhs(A)$, which is attached only to nodes in $U_i$.
Let $w = a_1\cdots a_n$ be a string and $X = \{x_1,\ldots,x_m\} \subseteq [n]$ a set of indices. By
$w|X$ we denote the string $a_{x_1}\cdots a_{x_m}$ such that $x_1 < x_2 < \cdots < x_m$. We define
the set of $i$-indices as
\[\idx[i] = \{j \mid \exists v \in V_{g_i} \text{ such that } v \text{ is at position } j \text{
within } \ext_{\rhs(A)}\}.\]
Using this, we define $\ext[i] = \ext_{\rhs(A)}|\idx[i]$, and for some edge $e$ labeled $A$ we
define $\att[i](e) = \att(e)|\idx[i]$. Now, we replace every occurrence of an edge $e$ labeled $A$
within the grammar by $k$ edges, such that for $i \in [k]$ the edge $e_i$ is labeled $A_i$ and
attached to the nodes $\att[i](e)$. We remove the nonterminal $A$ and its rule from the grammar. Not
that now all right-hand sides are tree graphs. Note further that $\sum_{i\in [k]} \rank(e_i) =
\rank(e)$, and $\sum_{i\in [k]} |h_i| \leq |\rhs(A)|$ and therefore the size of the new grammar is
at most $|G|$.

In the third step
we eliminate rules that have external nodes, which are neither root nor leafs. To do so, we again
split the rules up. Let $(A,g)$ be a rule which has at least three external nodes $u,v,w$ such that
there is a path from $u$ to $v$ and from $v$ to $w$, i.e., $v$ is an inner node and not root or
leaf. Let $g_1,\ldots,g_k$ be the maximally sized subgraphs of $g$, such that their union is $g$,
and for every $i \in [k]$, $g_i$ is a tree-graph which \emph{does not} have a triple of external
nodes $u,v,w$ such that paths exist from $u$ to $v$ and $v$ to $w$, i.e., every external node of
$g_i$ is either a root or a leaf. Let $e$ be an $A$-labeled edge of the grammar and let $\ext[i]$
and $\att[i](e)$ be defined as above. Now we again replace $(A,g)$ by the rules
$(A_1,g_1),\ldots,(A_k,g_k)$ and every nonterminal edge $e$ with $\lbl(e) = A$ by the edges
$e_1,\ldots,e_k$, where $\lbl(e_i) = A_i$ and $\att(e_i) = \att[i](e)$. In this case, an external
node that is neither root nor leaf will appear in more than one of the graphs $g_1,\ldots,g_k$, thus
this step increases the size of the grammar. This can be bounded by 
\[\sum_{i\in [k]} \rank(A_i) \leq 2 \rank(A) \text{ and } \sum_{i\in [k]} |g_i|_V \leq 2|g|_V - 1.\]
The factor $2$ comes from the fact that, in the worst case, we may split at every
external node. Thus the node- and edge-sizes are at most doubled. Therefore, for the grammar $G'$ we
have $|G'| \leq 2|G|$.
\end{proof}
\begin{figure}[!t]
    \centering
    \input{figures/conversionTreeGrammarEx}
    \caption{The rules of Figure~\ref{fig:lineGraphExample} after converting them to tree-graphs.}
    \label{fig:conversionExample}
\end{figure}
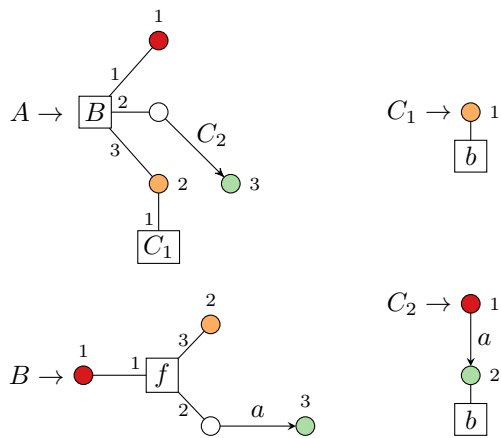
As an example, we provide in Figure~\ref{fig:conversionExample} the conversion of the rules
presented in Figure~\ref{fig:lineGraphExample}. We first adjust the order of the external nodes to
one consistent with $\prec_X$ for $X \in \{A,B,C\}$ by using red $<$ orange $<$ green as the order
of the external nodes in all three cases. Thus, the first external node becomes the third one, and
the third external node becomes the first one. In the same way, the appropriate nonterminal edges
($B$ and $C$ in $\rhs(A)$) have their attachment relation reordered. We then apply
the second step, splitting the $C$-rule into two rules $C_1$ and $C_2$. In this case the conversion
was already finished after the second step, thus the new rules are not larger than the original ones
(in fact, they are smaller, because the $C_2$-edge now is of rank $2$, which counts as size 1).

Consider again Theorem~\ref{thm:treeGeneratingToTree}. It clearly follows from the construction of
$G'$ that $V_{S'} = V_S$. This allows us to easily transfer the result also to string-graphs. Let
$G$ be a string-generating SL-HR grammar. Recall that $G$'s start graph $S$ has two external nodes
(because it generates a string-graph), and note that every nonterminal of $G$ has rank $\geq 2$. We
convert $G$ into a tree-generating SL-HR grammar by making the second external node $v$ in $S$
internal. The resulting grammar generates a monadic tree, and we
can use the result of Theorem~\ref{thm:treeGeneratingToTree} to convert it into a grammar, where
every right-hand side is a monadic tree-graph (of rank 2). Finally, this grammar can be converted
back into a string-generating grammar, by again making $v$ external.
\begin{corollary}
    Let $G = (N, P, S)$ be a string-generating SL-HR grammar. We can construct an SL-HR grammar $G' =
    (N', P', S')$
    such that
    \begin{enumerate}
        \item $\val(G') = \val(G)$,
        \item $|G'| \leq 2|G|$, and
        \item $\rhs_{G'}(A)$ is a string-graph for every $A \in N' \cup \{S'\}$.
    \end{enumerate}
    \label{cor:stringGeneratingToString}
\end{corollary}
%\begin{proposition}
%    A tree-generating (string-generating) SL-HR grammar $G$ into a context-free tree (string)
%    grammar $G'$ generating the tree (string) represented by $G$, with $|G'| \leq |G|$.
%\end{proposition}
%As every right-hand side of the grammar created with Lemma~\ref{lem:treeGeneratingToTree} is a graph
%encoding a tree (string) in the manner of the descriptions in Section~\ref{sss:stringTreeDef}, they
%can all be replaced by their respective trees/strings (substituting external nodes for parameters in
%the first case) to get such a grammar.
\subsection{Important Parameters}
In this section we describe some parameters of our algorithm that influence the compression
ratio. Their effect is evaluated experimentally in Section~\ref{sse:experiments}.
\subsubsection{Node Order}\label{ss2:node_order}
As discussed before the node order heavily influences the digram counting, which in turn influences
the compression behavior. A node order $<_V$ is given by a bijective function $f:V \rightarrow [|V|]$
such that $u <_V v$ if $f(u) < f(v)$. Some of the orders below are modeled by functions $g: V
\rightarrow [k]$ for some $k < |V|$, i.e., there exist nodes $u,v$ such that $g(u) = g(v)$, but a
strict order is needed to run the algorithm. For this reason, we consider sets of functions $F =
\{f: V \rightarrow [|V|] \mid \forall u,v \in V: f(u) < f(v) \Rightarrow g(u) \leq g(v)\}$ to be an
admissible set of orders. This means, that whenever we state that a node order was used, we actually
use an arbitrary order out of the set of admissible orders.
%As we cannot guarantee to find a maximal
%set of non-overlapping occurrences for every digram, the node order becomes the main factor in the
%quality of this set.
We evaluate these orders:
\begin{enumerate}
    \item \emph{natural order (nat)} uses the node IDs as given in the source graph,
    \item \emph{BFS} order follows a breadth-first traversal,
    \item \emph{\FP} computes a fixpoint on the node neighborhoods starting from the degrees, and
    \item \emph{\FP{0}}, which is a degree order.
\end{enumerate}
Formally, the natural order is just the identity. In the case of BFS there exists an additional
element of nondeterminism. We choose as the first node $v$ (i.e., the one with $f(v) = 1$) any node
of lowest degree. For any other node $u$ the function $g(u)$ evaluates to the length of a shortest
path from $v$ to $u$. This forms the basis for an admissible set of orders. Note that, for graphs
with more than one connected component, one node $v$ has to be picked for every connected component.
All of these initial nodes evaluate to $1$ by $g$. We now define the \FP and \FP{0} orders. For a
graph $g$ let $c_i: V_g \rightarrow \N$ be a family of functions that color every node with an
integer. We first define $c_0(v) = d(v)$, where $d(v)$ is the degree of $v$. This is the order
\FP{0}. Now we map every node $v$ to the tuple $f_0(v) = (c_0(v), c_0(v_1),\ldots,c_0(v_n))$, where
$v_1,\ldots,v_n$ are the neighbors of $v$ ordered by their values in $c_0$. We sort these tuples
lexicographically and let $c_1(v)$ be the position of $f_0(v)$ in this lexicographical order. This
process is iterated until $c_{i+1} = c_i$. Now $c_i$ can be used as a basis for an admissible set of
orders.
%Now $c_i$ implies an order of the nodes by defining $v <
%u$ iff $c_i(v) < c_i(u)$.
This computation of the order works for undirected, unlabeled graphs, but
can be straightforwardly extended to directed labeled graphs. We call this order $\FP$.
Figure~\ref{fig:node_order_example} shows an example of the \FP-order. The graph on the left is
annotated by $c_0$, the graph in the middle shows $f_0$, which is then ordered lexicographically to
get $c_1$ on the right. This is the fixpoint for this graph.
\begin{figure}[!t]
    \centering
    \input{figures/node_order_example}
    \caption{Computing the \FP-order of a small graph.}
    \label{fig:node_order_example}
\end{figure}
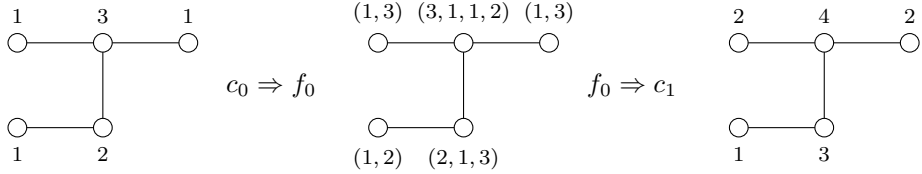
Note that it is not necessarily a strict order and thus also implies an equivalence relation on the
nodes ($v \FPequiv u$ if and only if $c_i(v) = c_i(u)$). The number of equivalence classes of
$\FPequiv$ has an interesting correlation with the compression ratio of \graphrepair, as discussed
in Section~\ref{ss2:nodeOrder}.
\subsubsection{Maximal Rank}\label{ss2:maxRank}
The maximal rank is a user defined parameter of \graphrepair that specifies the maximal rank of a
digram (and thus the maximal rank of a nonterminal edge) that the compressor considers. Digrams with
a higher rank are ignored and not counted. It was shown already for TreeRePair~\cite[Theorems 9 and
10]{Lohrey13_treeRePair}  that choosing this parameter too high or too small can have strong effects
on compression (in both directions). The two families of trees given there, can be converted into
families of graphs showing the same relation to the maximal rank for \graphrepair. We briefly recap
the arguments, but as the proof details are almost identical to the case for TreeRePair, we do not
repeat them here.
\paragraph{High Rank} Let $g_{n,k}$ be the graph $\tgraph(s_{n,k})$, where $s_{n,k}$ is a tree
consisting of a right comb of $2^n$ nodes labeled $f_k$, a symbol of rank $k$. The first $k-1$
children of these nodes are leaves labeled $a$, the last child is the next $f_k$-labeled node
(except for the last of these, which only has $a$-children). This is the same tree as given
in~\cite[Theorem 9]{Lohrey13_treeRePair}. The size of $g_{n,k}$ depends on $k$: for $k = 1$ the
$f$-edges are of rank 2 (i.e., size 1) and thus we get $|g_{n,1}| = 2^n + 2^n + 1$. For $k > 1$ the
size of each $f$-edge is $k+1$ and thus $|g_{n,k}| = 2^n\cdot (k+1) + 2^n + 1$.
\begin{lemma}
    Given $g_{n,k}$ \graphrepair produces a grammar of size $O(k^2+n)$ if the maximal rank allowed is
    at least $k$, and does not compress at all if the maximal rank is less than $k$.
    \label{lem:unbounded_maxrank}
\end{lemma}
It is easy to see that there can be no compression with a maximal rank less than $k$: $g_{n,k}$ does
not contain occurrences of any digram with a rank smaller than $k$. Once the allowed rank is at
least $k$, \graphrepair reduces the width of the tree by $1$ with every iteration by combining one
of the $a$-leaves with its $f_k$-/nonterminal parent. Finally a line of $2^n$ nonterminal edges
remains, which all have the same label and can be compressed exponentially. This argument is identical
to the one in the proof of~\cite[Theorem 9]{Lohrey13_treeRePair}, only the notation changes to the
one used for graphs.
%\begin{figure}[!t]
%    \centering
%    \subfloat[The graph $g_{2,3}$]{\input{figures/highRankEx}\label{fig:highRankEx}}\qquad
%    \subfloat[The graph $T_{3}$]{\input{figures/lowRankEx}\label{fig:lowRankEx}}
%    \caption{Examples for graphs compressing best with unbounded \protect\subref{fig:highRankEx} rank, and
%    with $\text{maxRank}=2$ \protect\subref{fig:lowRankEx}.}
%    \label{fig:RankEx}
%\end{figure}
\paragraph{Small Rank}
We give an example of graphs which \graphrepair compresses best, if the maximum rank is limited to
$2$. Our example is similar to the one TreeRePair~\cite[Theorem 10]{Lohrey13_treeRePair}. The main
difference is that we use a graph with $2^n$ edges labeled $f$, instead of nodes. Let $T_n$ be a
graph over $\Sigma = \{f,a_0,a_1,a_2,a_3,a_4\}$ with $2^{n+1}+2$ nodes $\{0,\ldots, 2^{n+1}+1\}$ and
$2^{n+1}$ edges $\{e_1,\ldots,e_{2^{n+1}}\}$. The attachment relation is defined as $\att(e_i) =
i-1\cdot i$ for $i \in \{1,\ldots,2^n\}$, and $\att(e_i) = i-2^n \cdot i$ for $i \in
\{2^n+1,\ldots,2^{n+1}+1\}$. They are labeled by $\lbl(e_i) = f$ for $i \in \{1,\ldots,2^n\}$ and
\[\lbl(e_i) = 
    \begin{cases}
        a_0 & \text{if } i \equiv 0 \pmod 5 \\
        a_1 & \text{if } i \equiv 1 \pmod 5 \\
        a_2 & \text{if } i \equiv 2 \pmod 5 \\
        a_3 & \text{if } i \equiv 3 \pmod 5 \\
        a_4 & \text{if } i \equiv 4 \pmod 5
    \end{cases},
\]
for $i \in \{2^{n}+1,\ldots,2^{n+1}\}$. This graph is a tree consisting of a path of $2^n$ edges
labeled $f$, and every node attached to one of these $f$-edges, is also attached to an edge labeled
$a_0$, $a_1$, $a_2$, $a_3$, or $a_4$. These edges appear in this order, i.e., if a node is attached
to an edge labeled $a_i$, then the previous node (one $f$-edge up in the tree) is labeled
$a_{i-1}$ ($\bmod 5$), and the next node (one $g$-edge down in the tree) is labeled $a_{i+1}$
($\bmod 5$). The size of $T_n$ is $2^{n+2}+2$.
\begin{lemma}
    Given $T_{n}$, \graphrepair produces a grammar of size $O(n)$ if the maximal rank is restricted
    to $2$ and compresses at best to $75\%$ of the original size if the maximal rank is unbounded.
    \label{lem:bounded_maxrank}
\end{lemma}
The argument for this is mostly identical to~\cite[Theorem 10]{Lohrey13_treeRePair}. With a maximal
rank of $2$ it compresses well, because \graphrepair then only has limited choice in digrams. It
will first replace the pairs of $f$- and $a_l$-labeled ($0 \leq l \leq 4$) edges, because every
other possible digram has a rank greater than $2$. After this is done, a string-graph remains that can
now be compressed exponentially. If the maximal rank is unrestricted, the digrams of higher rank
(pairs of $2$ $f$-edges in the first iteration) occur more frequently and are thus replaced by
\graphrepair. This will replace all the $f$-edges, but all of the nodes remain in the start graph.
Furthermore, in the next step, pairs of nonterminal edges will be replaced as the most frequent
digram, again without sharing any nodes. This continues until, after the last iteration, the start
graph still has all $2^{n+1}+2$ nodes, and all the edges labeled $a_l$ for $0 \leq l \leq 4$.
Furthermore, it has 2 nonterminal edges of rank $2^{n-1}$. Thus, 
\[|S| = 2^{n+1} + 2^{n} + 2\cdot 2^{n-1} + 4 = 2^{n+2}+4.\]
Regardless of the rules and the pruning step, $|S|$ is already larger than $|T_n|$. In particular
note that, not counting the nonterminal edges, the size of $S$ is still $x = 2^{n+1}+2^n +4$.
Therefore 
\[\frac{x}{|T_n|} = \frac{2^{n+1} + 2^n + 4}{2^{n+1} + 2^{n+1} + 2} \overset{n\rightarrow
\infty}{\longrightarrow} \frac{3}{4}. \]
The
pruning step reduces the grammar's size, but the nodes and terminal edges in the
start graph remain. Thus, the compression ratio cannot be better than $3/4$.
Note that for TreeRePair the compression with unbounded rank is at best $50\%$ instead of $75\%$.
The reason for this is, that the size definition for trees only counts edges, whereas we consider
nodes and edges in graphs.

%% file: figures/all_digrams.tex
\begin{tikzpicture}
    \begin{scope}[every node/.style={circle, draw, inner sep=0pt, node distance=20pt,
        minimum size=7pt}]
        \node[fill=black] (1c) {};
        \node[fill=black, left=of 1c] (11) {};
        \node[fill=black, right=of 1c] (12) {};
%    \end{scope}
    \draw (11) -- (1c) -- (12);
        \node[below=.623 of 1c, fill=black] (1cb) {};
        \node[fill=black, left=of 1cb] (11b) {};
        \node[fill=black, right=of 1cb] (12b) {};
    \draw (11b) -- (1cb) -- (12b);

%    \begin{scope}[every node/.style={circle, draw, inner sep=2pt, node distance=25pt}]
        \node[below=.623 of 1cb] (2c) {};
        \node[left=of 2c, fill=black] (21) {};
        \node[right=of 2c, fill=black] (22) {};
%    \end{scope}
    \draw (21) -- (2c) -- (22);

%    \begin{scope}[every node/.style={circle, draw, inner sep=2pt, node distance=25pt}]
        \node[below=.623 of 2c] (4c) {};
        \node[left=of 4c, fill=black] (41) {};
        \node[right=of 4c] (42) {};
%    \end{scope}
    \draw (41) -- (4c) -- (42);

%    \begin{scope}[every node/.style={circle, draw, inner sep=2pt, node distance=25pt}]
        \node[right=3 of 12, fill=black] (3c) {};
        \node[left=of 3c, fill=black] (31) {};
        \node[right=of 3c] (32) {};
%    \end{scope}
    \draw (31) -- (3c) -- (32);
        \node[below=.623 of 3c, fill=black] (3cb) {};
        \node[left=of 3cb, fill=black] (31b) {};
        \node[right=of 3cb] (32b) {};
    \draw (31b) -- (3cb) -- (32b);

%    \begin{scope}[every node/.style={circle, draw, inner sep=2pt, node distance=25pt}]
        \node[below=.623 of 31b] (51) {};
        \node[right=of 51, fill=black] (5c) {};
        \node[right=of 5c] (52) {};
%    \end{scope}
    \draw (51) -- (5c) -- (52);

%    \begin{scope}[every node/.style={circle, draw, inner sep=2pt, node distance=25pt}]
        \node[below=.623 of 51] (61) {};
        \node[right=of 61, fill=black] (6c) {};
%    \end{scope}
    \draw (61) -- (6c);
%    \path (61) edge[loop above] (61);
    \path (6c) edge[loop right, out=320, in=50, min distance=8mm, looseness=10] (6c);

%    \begin{scope}[every node/.style={circle, draw, fill=black, inner sep=2pt, node distance=25pt}]
        \node[right=3 of 32, fill=black] (7c) {};
        \node[left=of 7c, fill=black] (71) {};
%    \end{scope}
%    \end{scope}
    \draw (71) -- (7c);
%    \path (71) edge[loop above] (71);
    \path (7c) edge[loop right, out=320, in=50, min distance=8mm, looseness=10] (7c);

        \node[below=.623 of 71, fill=black] (71b) {};
        \node[right=of 71b, fill=black] (7cb) {};
%    \end{scope}
%    \end{scope}
    \draw (71b) -- (7cb);
%    \path (71) edge[loop above] (71);
    \path (7cb) edge[loop right, out=320, in=50, min distance=8mm, looseness=10] (7cb);
%    \begin{scope}[every node/.style={circle, draw, inner sep=2pt, node distance=25pt}]

        \node[below=.623 of 71b, fill=black] (81) {};
        \node[right=of 81] (8c) {};

        \node[below=.623 of 8c, fill=black] (9c) {};
        \node[draw=none, left=of 9c] (91) {};
    \end{scope}
    \draw (81) -- (8c);
    \path (8c) edge[loop right, out=320, in=50, min distance=8mm, looseness=10] (8c);

    \path (9c) edge[loop right, out=320, in=50, min distance=8mm, looseness=10] (9c);
    \path (9c) edge[loop left, out=220, in=130, min distance=8mm, looseness=10] (9c);

    \begin{scope}[every node/.style={inner sep=0pt, node distance=20pt, minimum size=7pt}]
        \node[below=-0pt of 1c] {\scriptsize 2};
        \node[below=-0pt of 11] {\scriptsize 1};
        \node[below=-0pt of 12] {\scriptsize 3};
        \node[below=-0pt of 1cb] {\scriptsize 1};
        \node[below=-0pt of 11b] {\scriptsize 2};
        \node[below=-0pt of 12b] {\scriptsize 3};
        \node[below=-0pt of 21] {\scriptsize 1};
        \node[below=-0pt of 22] {\scriptsize 2};
        \node[below=-0pt of 41] {\scriptsize 1};
        \node[below=-0pt of 3c] {\scriptsize 2};
        \node[below=-0pt of 31] {\scriptsize 1};
        \node[below=-0pt of 3cb] {\scriptsize 1};
        \node[below=-0pt of 31b] {\scriptsize 2};
        \node[below=-0pt of 5c] {\scriptsize 1};
        \node[below=-0pt of 6c] {\scriptsize 1};
        \node[below=-0pt of 71] {\scriptsize 1};
        \node[below=-0pt of 7c] {\scriptsize 2};
        \node[below=-0pt of 71b] {\scriptsize 2};
        \node[below=-0pt of 7cb] {\scriptsize 1};
        \node[below=-0pt of 81] {\scriptsize 1};
        \node[below=-0pt of 9c] {\scriptsize 1};
    \end{scope}

    \draw ($ (11.west) + (-.125,.125) $) to [square right brace] ($ (11b.west) + (-.125,-.125) $);
    \draw ($ (31.west) + (-.125,.125) $) to [square right brace] ($ (31b.west) + (-.125,-.125) $);
    \draw ($ (71.west) + (-.125,.125) $) to [square right brace] ($ (71b.west) + (-.125,-.125) $);

    \begin{scope}[every node/.style={inner sep=0pt, node distance=20pt, minimum size=7pt}]
        \node[below=.25 of 11] (x1) {};
        \node[left=12pt of x1] (a) {\textit{a})};
        \node[left=12pt of 21] (b) {\textit{b})};
        \node[left=12pt of 41] (c) {\textit{c})};

        \node[below=.25 of 31] (x2) {};
        \node[left=12pt of x2] (d) {\textit{d})};
        \node[left=.40 of 51] (e) {\textit{e})};
        \node[left=.40 of 61] (f) {\textit{f})};

        \node[below=.25 of 71] (x3) {};
        \node[left=12pt of x3] (g) {\textit{g})};
        \node[left=.40 of 81] (h) {\textit{h})};
        \node[left=.40 of 91] (i) {\textit{i})};
    \end{scope}

\end{tikzpicture}

%% file: figures/treeRePair_example.tex
\begin{tikzpicture}%[auto]
    \node (1) {$c$};
    \node[below right=.25 and .5 of 1] (2) {$a$};
    \node[below left=.25 and .5 of 1] (3) {$c$};
    \node[below left=.25 and .5 of 3] (4) {$a$};
    \node[below right=.25 and .5 of 3] (5) {$c$};
    \node[below left=.25 and .5 of 5] (6) {$a$};
    \node[below right=.25 and .5 of 5] (7) {$a$};

    \begin{scope}[every node/.style={coordinate}]
        \draw (1) --node (c1) {} (2);
        \draw (1) --node (c2) {} (3);
        \draw (3) --node (c3) {} (4);
        \draw (3) --node (c4) {} (5);
        \draw (5) --node (c5) {} (6);
        \draw (5) --node (c6) {} (7);
    \end{scope}

    \begin{pgfonlayer}{background}
        \begin{scope}[every node/.style={draw, rectangle, rounded corners, minimum width=4.5em, minimum height=1.5em, opacity=.5}]
            \node[rotate=145, fill=color2] at (c1) {};
            \node[rotate=35, fill=color1] at (c3) {};
            \node[rotate=35, fill=color1] at (c5) {};
            \node[rotate=145, fill=color2] at (c6) {};
        \end{scope}
        \begin{scope}[every node/.style={draw, rectangle, rounded corners, minimum width=4.5em, minimum height=1.5em}]
            \node[rotate=35, thick, dashed, color=color1] at (c2) {};
            \node[rotate=145, thick, dashed, color=color2] at (c4) {};
        \end{scope}
    \end{pgfonlayer}

%    \draw ($ (3.south) + (0,-.1) $) -- ($ (1.east) + (.1,0) $) -- ($ (1.north) + (0,.1) $) -- ($ (3.west) + (-.1,0) $) -- ($ (3.south) + (0,-.1) $);
\end{tikzpicture}

%% file: figures/treeRePair_example_productions.tex
\begin{tikzpicture}
    \node (A) {$A$};
    \node[below=.25 of A] (par) {$y_1$};
    \node[below right=-0.1 and .125 of A] (to) {$\rightarrow$};
    \node[right=1.5 of A] (c) {$c$};
    \node[below left=.25 and .125 of c] (y) {$y_1$};
    \node[below right=.25 and .125 of c] (a) {$a$};

    \draw (A) -- (par);
    \draw (c) -- (y);
    \draw (c) --node[coordinate] (x1) {} (a);
    
    \node[below=1 of A] (B) {$B$};
    \node[below left=.25 and .075 of B] (par1) {$y_1$};
    \node[below=.25 of B] (par2) {$y_2$};
    \node[below right=.25 and .075 of B] (par3) {$y_3$};
    \node[below right=-0.1 and .35 of B] (to1) {$\rightarrow$};

    \node[right=2 of B] (c1) {$c$};
    \node[below left=.25 and .125 of c1] (c2) {$c$};
    \node[below right=.25 and .125 of c1] (y1) {$y_3$};
    \node[below left=.25 and .125 of c2] (y2) {$y_1$};
    \node[below right=.25 and .125 of c2] (y3) {$y_2$};

    \draw (B) -- (par1);
    \draw (B) -- (par2);
    \draw (B) -- (par3);
    \draw (c1) --node[draw, thick, rectangle, rounded corners, color1, rotate=50, minimum
    height=1.5em, minimum width=4em, dashed] {} (c2);
    \draw (c1) -- (y1);
    \draw (c2) -- (y2);
    \draw (c2) -- (y3);

    \begin{pgfonlayer}{background}
        \begin{scope}[every node/.style={draw, rectangle, rounded corners, minimum height=1.5em,
            minimum width=4em, opacity=.5}]
            \node[fill=color2, rotate=130] at (x1) {};
        \end{scope}
    \end{pgfonlayer}
\end{tikzpicture}

%% file: figures/size_example.tex
\begin{tikzpicture}
    \begin{scope}[every node/.style={circle, draw, inner sep=0pt, minimum size=7pt, node
        distance=20pt}]
        \node (o1) {};
        \node[right=of o1] (o2) {};
        \node[right=of o2] (o3) {};
    \end{scope}
    \begin{pgfonlayer}{background}
        \node[draw, rectangle, rounded corners, minimum width=7.5em, minimum height=1.5em, fill=color1,
        opacity=.5] at (o2.center) {};
    \end{pgfonlayer}
    \draw (o1) -- (o2);
    \draw (o2) -- (o3);
    \draw (o3) .. controls +(0,-1) and +(0,-1) .. (o1);

    \node[right=of o3] (sS) {$S = $};
    \node[below=1.0 of sS] (sA) {$A \rightarrow$};
    \begin{scope}[every node/.style={circle, draw, inner sep=0pt, minimum size=7pt, node
        distance=20pt}]
        \node[right=.1 of sS] (s1) {};
        \node[right=of s1] (s2) {};
        \node[right=of s2] (s3) {};
        \node[above=.5 of s2, rectangle,inner sep=2pt] (A) {\normalsize $A$};
        \node[right=.1 of sA, fill=color1] (sA1) {};
        \node[right=of sA1, fill=color3] (sA2) {};
        \node[right=of sA2, fill=color4] (sA3) {};
    \end{scope}
%    \draw (s1) -- (s2);
%    \draw (s2) -- (s3);
    \draw[very thick, color1] (A) -- (s1);
    \draw[very thick, color3] (A) -- (s2);
    \draw[very thick, color4] (A) -- (s3);
    \draw (s3) .. controls +(0,-1) and +(0,-1) .. (s1);
    \draw (sA1) -- (sA2);
    \draw (sA2) -- (sA3);

    \node[right=of s3] (gS) {$S = $};
    \node[below=1.0 of gS] (gA) {$A \rightarrow$};
    \begin{scope}[every node/.style={circle, draw, inner sep=0pt, minimum size=7pt, node
        distance=20pt}]
        \node[right=.1 of gS] (g1) {};
        \node[right=of g1, rectangle,inner sep=2pt] (g2) {\normalsize $A$};
        \node[right=of g2] (g3) {};
        \node[right=.1 of gA, fill=color1] (gA1) {};
        \node[right=of gA1] (gA2) {};
        \node[right=of gA2, fill=color2] (gA3) {};
    \end{scope}
    \draw[very thick, color1] (g1) -- (g2);
    \draw[very thick, color2] (g2) -- (g3);
    \draw (g3) .. controls +(0,-1) and +(0,-1) .. (g1);
    \draw (gA1) -- (gA2);
    \draw (gA2) -- (gA3);
\end{tikzpicture}

%% file: figures/counting_example_center.tex
\begin{tikzpicture}
    \begin{scope}[every node/.style={circle, draw, inner sep=0pt, minimum size=12pt, node
        distance=15pt}]
        \node (c) {\normalsize 1};
        \node[left=of c] (l) {};
        \node[left=of l] (l2) {};
        \node[above=of l2] (l1) {};
        \node[above=of c] (a) {};
        \node[left=of a] (a1) {};
        \node[right=of a] (a2) {};
        \node[right=of c] (r) {};
        \node[right=of r] (r1) {};
        \node[below=of r1] (r2) {};
        \node[below=of c] (b) {};
        \node[left=of b] (b1) {};
        \node[right=of b] (b2) {};
    \end{scope}

    \begin{pgfonlayer}{background}
        \begin{scope}[every node/.style={draw, rectangle, minimum height=18pt, minimum width=72pt, rounded corners, opacity=.5}]
            \node[fill=color1] at (c.center) {};
            \node[rotate=90, fill=color2] at (c.center) {};
        \end{scope}
    \end{pgfonlayer}

    \draw (c) -- (l);
    \draw (l) -- (l1);
    \draw (l) -- (l2);
    \draw (c) -- (a);
    \draw (a) -- (a1);
    \draw (a) -- (a2);
    \draw (c) -- (r);
    \draw (r) -- (r1);
    \draw (r) -- (r2);
    \draw (c) -- (b);
    \draw (b) -- (b1);
    \draw (b) -- (b2);
    \draw (a2) -- (r1);
    \draw (b1) -- (l2);
\end{tikzpicture}

%% file: figures/counting_example_dfs.tex
\begin{tikzpicture}
    \begin{scope}[every node/.style={circle, draw, inner sep=0pt, minimum size=12pt, node distance=15pt}]
        \node (c) {\normalsize 5};
        \node[left=of c] (l) {};
        \node[left=of l] (l2) {};
        \node[above=of l2] (l1) {};
        \node[above=of c] (a) {\normalsize 4};
        \node[left=of a] (a1) {};
        \node[right=of a] (a2) {\normalsize 3};
        \node[right=of c] (r) {\normalsize 1};
        \node[right=of r] (r1) {\normalsize 2};
        \node[below=of r1] (r2) {};
        \node[below=of c] (b) {};
        \node[left=of b] (b1) {};
        \node[right=of b] (b2) {};
    \end{scope}

    \begin{pgfonlayer}{background}
        \begin{scope}[every node/.style={draw, rectangle, minimum height=18pt, minimum
            width=72pt, rounded corners, opacity=.5}]
            \node[fill=color1] at (r.center) {};
%            \node[rotate=90,fill=color2] at (a.center) {};
        \end{scope}
%    \begin{scope}[every path/.style={dashed}]
%        \draw (l.west) .. controls (-.5,1.3) and (1.3,-.5) .. (b.south) ..
%        controls +(-.6,-.2) and +(.2,-.2) .. ($ (c.south west) + (-.1,-.1) $) ..
%        controls +(-.2,.2) and +(-.2,-.6) .. (l.west);
%    \end{scope}
        \begin{scope}[every path/.style={rounded corners, opacity=.5}]
            \draw[fill=color3] ($ (c.center) + (0,9pt) $) -- +(9pt,0) --
            +(9pt,-45pt) -- +(-9pt,-45pt) --
            +(-9pt,-18pt) -- +(-36pt,-18pt) -- +(-36pt,0) --  ($ (c.center) + (0,9pt) $);
        \end{scope}
        \begin{scope}[every path/.style={rounded corners, opacity=.5}]
            \draw[fill=color2, rotate=90] ($ (a.center) + (0,9pt) $) -- +(9pt,0) --
            +(9pt,-45pt) -- +(-9pt,-45pt) --
            +(-9pt,-18pt) -- +(-36pt,-18pt) -- +(-36pt,0) --  ($ (a.center) + (0,9pt) $);
        \end{scope}
    \end{pgfonlayer}

    \draw (c) -- (l);
    \draw (l) -- (l1);
    \draw (l) -- (l2);
    \draw (c) -- (a);
    \draw (a) -- (a1);
    \draw (a) -- (a2);
    \draw (c) -- (r);
    \draw (r) -- (r1);
    \draw (r) -- (r2);
    \draw (c) -- (b);
    \draw (b) -- (b1);
    \draw (b) -- (b2);
    \draw (a2) -- (r1);
    \draw (b1) -- (l2);
\end{tikzpicture}

%% file: figures/grammar_example.tex
\begin{tikzpicture}[]
    \begin{scope}[every node/.style={circle, draw, inner sep=0pt, minimum size=12pt, node
        distance=25pt}]
%        \node (c) {};
%        \node[left=of c] (l) {};
%%        \node[above=of l] (l1) {};
%        \node[left=of l] (l2) {};
%        \node[above=of c] (a) {};
%%        \node[above=of a] (a1) {};
%        \node[right=of a] (a2) {};
%        \node[right=of c] (r) {};
%        \node[right=of r] (r1) {};
%%        \node[below=of r] (r2) {};
%        \node[below=of c] (b) {};
%        \node[left=of b] (b1) {};
%%        \node[below=of b] (b2) {};
        \node (c) {\normalsize 1};
        \node[left=of c] (l) {\normalsize 2};
%        \node[above=of l] (l1) {};
        \node[left=of l] (l2) {\normalsize 3};
        \node[above=of c] (a) {\normalsize 4};
%        \node[above=of a] (a1) {};
        \node[right=of a] (a2) {\normalsize 5};
        \node[right=of c] (r) {\normalsize 6};
        \node[right=of r] (r1) {\normalsize 7};
%        \node[below=of r] (r2) {};
        \node[below=of c] (b) {\normalsize 8};
        \node[left=of b] (b1) {\normalsize 9};
    \end{scope}

%    \begin{scope}[every node/.style={draw, rectangle, minimum height=1.5em, minimum width=6.5em, opacity=.2, rounded corners}]
%        \node[fill=color1!30] at (c.center) {};
%        \node[rotate=90, fill=color2!30] at (c.center) {};
%    \end{scope}

    \draw (c) -- (l);
%    \draw (l) -- (l1);
    \path (l) --node[rectangle, draw, fill=white, inner sep=2pt, outer sep=0pt, above=10pt] (1) {\normalsize $A$} (l2);
    \draw[color=color1, very thick] (l) -- (1);
    \draw[color=color2, very thick] (1) -- (l2);
    \draw (c) -- (a);
%    \draw (a) -- (a1);
    \path (a) --node[rectangle, draw, fill=white, inner sep=2pt, outer sep=0pt, below=10pt] (2) {\normalsize $A$} (a2);
    \draw[color=color1, very thick] (a) -- (2);                                             
    \draw[color=color2, very thick] (2) -- (a2);                                            
    \draw (c) -- (r);                                                               
    \path (r) --node[rectangle, draw, fill=white, inner sep=2pt, outer sep=0pt, below=10pt] (3) {\normalsize $A$} (r1);
    \draw[color=color1, very thick] (r) -- (3);                                             
    \draw[color=color2, very thick] (3) -- (r1);                                            
%    \draw (r) -- (r2);                                                             
    \draw (c) -- (b);                                                               
    \path (b) --node[rectangle, draw, fill=white, inner sep=2pt, outer sep=0pt, above=10pt] (4) {\normalsize $A$} (b1);
    \draw[color=color1, very thick] (b) -- (4);
    \draw[color=color2, very thick] (4) -- (b1);
%    \draw (b) -- (b2);
    \draw (a2) -- (r1);
    \draw (b1) -- (l2);

    \node[right=of r1] (A) {$A \rightarrow$};
    \begin{scope}[every node/.style={circle, draw, inner sep=0pt, minimum size=12pt, node
        distance=20pt}]
%        \node[right=.1 of A, fill=color1] (A1) {};
%        \node[right=of A1, fill=color2] (A2) {};
%        \node[above=of A1] (A3) {};
        \node[right=.1 of A, fill=color1] (A1) {\normalsize \color{white} 1};
        \node[right=of A1, fill=color2] (A2) {\normalsize 2};
        \node[above=of A1] (A3) {\normalsize 3};
    \end{scope}
    \draw (A1) -- (A2);
    \draw (A1) -- (A3);
%    \node[below=.1pt of A1] {\normalsize 1};
%    \node[below=.1pt of A2] {\normalsize 2};
\end{tikzpicture}

%% file: figures/graph_derivation_order_example.tex
\begin{tikzpicture}
    \begin{scope}[every node/.style={circle, draw, inner sep=0pt, minimum size=12pt, node distance=15pt}]
        \node (c) {\normalsize 1};
        \node[left=of c] (l) {\normalsize 2};
        \node[left=of l] (l2) {\normalsize 3};
        \node[above=of l2] (l1) {\normalsize 10};
        \node[above=of c] (a) {\normalsize 4};
        \node[left=of a] (a1) {\normalsize 11};
        \node[right=of a] (a2) {\normalsize 5};
        \node[right=of c] (r) {\normalsize 6};
        \node[right=of r] (r1) {\normalsize 7};
        \node[below=of r1] (r2) {\normalsize 12};
        \node[below=of c] (b) {\normalsize 8};
        \node[left=of b] (b1) {\normalsize 9};
        \node[right=of b] (b2) {\normalsize 13};
    \end{scope}

    \draw (c) -- (l);
    \draw (l) -- (l1);
    \draw (l) -- (l2);
    \draw (c) -- (a);
    \draw (a) -- (a1);
    \draw (a) -- (a2);
    \draw (c) -- (r);
    \draw (r) -- (r1);
    \draw (r) -- (r2);
    \draw (c) -- (b);
    \draw (b) -- (b1);
    \draw (b) -- (b2);
    \draw (a2) -- (r1);
    \draw (b1) -- (l2);
\end{tikzpicture}

%% file: figures/line_graph_example.tex
\begin{tikzpicture}
    \node (Arule) {$A \rightarrow$};
    \node[below=2 of Arule] (Brule) {$B \rightarrow$};
    \node[below=of Brule] (Crule) {$C \rightarrow$};
    \begin{scope}[every node/.style={circle, draw, inner sep=0pt, outer sep=0pt, minimum size=7pt,
        node distance=20pt}]
        \node[right=1 of Arule] (A2) {};
        \node[above=of A2, fill=color1] (A1) {};
        \node[below=of A2, fill=color2] (A3) {};
        \node[right=of A3, fill=color3] (A4) {};
    \end{scope}
    \begin{scope}[every node/.style={rectangle, draw, inner sep=2pt, outer sep=0pt, minimum size=12pt,
        node distance=20pt}]
        \node[left=.5 of A2] (AB) {$B$};
        \node[right=.5 of A2] (AC) {$C$};
    \end{scope}
    \begin{scope}[every path/.style={}]
        \draw (AB) --node[pos=.3,above left=-4pt] {\scriptsize $3$} (A1);
        \draw (AB) --node[pos=.3,above=-2pt] {\scriptsize $2$} (A2);
        \draw (AB) --node[pos=.3,below left=-4pt] {\scriptsize $1$} (A3);
        \draw (AC) --node[pos=.3,above=-2pt] {\scriptsize $3$} (A2);
        \draw (AC) --node[pos=.3,above left=-4pt] {\scriptsize $2$} (A3);
        \draw (AC) --node[pos=.3,left=-2pt] {\scriptsize $1$} (A4);
    \end{scope}
    \node[above=0pt of A1] {\scriptsize $3$};
    \node[below=0pt of A3] {\scriptsize $2$};
    \node[below=0pt of A4] {\scriptsize $1$};

    \begin{scope}[every node/.style={circle, draw, inner sep=0pt, outer sep=0pt, minimum size=7pt,
        node distance=20pt}]
        \node[right=0pt of Brule, fill=color1] (B1) {};
        \node[above right=.5 and 1.5 of B1, fill=color2] (B2) {};
        \node[below right=.5 and 1.5 of B1] (B3) {};
        \node[right=1 of B3, fill=color3] (B4) {};
    \end{scope}
    \begin{scope}[every node/.style={rectangle, draw, inner sep=2pt, outer sep=0pt, minimum size=12pt,
        node distance=20pt}]
        \node[right=of B1] (Bf) {$f$};
    \end{scope}
    \draw (Bf) --node[pos=.2,above=-1pt] {\scriptsize 1} (B1);
    \draw (Bf) --node[pos=.2,above left=-4pt] {\scriptsize 3} (B2);
    \draw (Bf) --node[pos=.2,below left=-4pt] {\scriptsize 2} (B3);
    \draw[-stealth] (B3) --node[above] {$a$} (B4);
    \node[below=0pt of B1] {\scriptsize $3$};
    \node[below=0pt of B2] {\scriptsize $2$};
    \node[below=0pt of B4] {\scriptsize $1$};

    \begin{scope}[every node/.style={circle, draw, inner sep=0pt, outer sep=0pt, minimum size=7pt,
        node distance=20pt}]
        \node[right=0pt of Crule, fill=color2] (C1) {};
        \node[right=of C1, fill=color1] (C2) {};
        \node[below=of C2, fill=color3] (C3) {};
    \end{scope}
    \begin{scope}[every node/.style={rectangle, draw, inner sep=2pt, outer sep=0pt, minimum size=12pt,
        node distance=20pt}]
        \node[below=.25 of C1] (Cb1) {$b$};
        \node[below=.25 of C3] (Cb2) {$b$};
    \end{scope}
    \draw (Cb1) -- (C1);
    \draw (Cb2) -- (C3);
    \draw[-stealth] (C2) --node[right=-1pt] {$a$} (C3);
    \node[right=0pt of C1] {\scriptsize $2$};
    \node[right=0pt of C2] {\scriptsize $3$};
    \node[right=0pt of C3] {\scriptsize $1$};

    \node[right=2.5 of A1] (valA) {$\val(A):$};
    \begin{scope}[every node/.style={circle, draw, inner sep=0pt, outer sep=0pt, minimum size=7pt,
        node distance=20pt}]
        \node[below=0pt of valA.south east,fill=color1] (v1) {};
        \node[below left=1.5 and 0.5 of v1] (v2) {};
        \node[below right=1.5 and 0.5 of v1] (v3) {};
        \node[below=of v2, fill=color2] (v4) {};
        \node[below=of v3, fill=color3] (v5) {};
    \end{scope}
    \begin{scope}[every node/.style={rectangle, draw, inner sep=2pt, outer sep=0pt, minimum size=12pt,
        node distance=20pt}]
        \node[below=of v1] (vf) {$f$};
        \node[below=.25 of v4] (vb1) {$b$};
        \node[below=.25 of v5] (vb2) {$b$};
    \end{scope}
    \draw (vf) --node[pos=.3,right=-1pt] {\scriptsize 1} (v1);
    \draw (vf) --node[pos=.3,above left=-4pt] {\scriptsize 2} (v2);
    \draw (vf) --node[pos=.3,above right=-4pt] {\scriptsize 3} (v3);
    \draw[-stealth] (v2) --node[right] {$a$} (v4);
    \draw[-stealth] (v3) --node[right] {$a$} (v5);
    \draw (vb1) -- (v4);
    \draw (vb2) -- (v5);
    \node[right=0pt of v1] {\scriptsize $3$};
    \node[right=0pt of v4] {\scriptsize $2$};
    \node[right=0pt of v5] {\scriptsize $1$};

    \node[right=2.5 of valA] (lsA) {$\ls(A):$};
    \begin{scope}[every node/.style={circle, draw, inner sep=0pt, outer sep=0pt, minimum size=7pt,
        node distance=20pt}]
        \node[below=0pt of lsA.south east,fill=color1] (l1) {};
        \node[below left=of l1,fill=color2] (l2) {};
        \node[below right=of l1,fill=color3] (l3) {};
    \end{scope}
    \draw[-stealth] (l1) -- (l2);
    \draw[-stealth] (l1) -- (l3);
    \node[below=0pt of l1] {\scriptsize $3$};
    \node[below=0pt of l2] {\scriptsize $2$};
    \node[below=0pt of l3] {\scriptsize $1$};

    \node[below=2 of lsA] (lsB) {$\ls(B):$};
    \begin{scope}[every node/.style={circle, draw, inner sep=0pt, outer sep=0pt, minimum size=7pt,
        node distance=20pt}]
        \node[below=0pt of lsB.south east,fill=color1] (l1B) {};
        \node[below left=of l1B,fill=color2] (l2B) {};
        \node[below right=of l1B,fill=color3] (l3B) {};
    \end{scope}
    \draw[-stealth] (l1B) -- (l2B);
    \draw[-stealth] (l1B) -- (l3B);
    \node[below=0pt of l1B] {\scriptsize $3$};
    \node[below=0pt of l2B] {\scriptsize $2$};
    \node[below=0pt of l3B] {\scriptsize $1$};

    \node[below=2 of lsB] (lsC) {$\ls(C):$};
    \begin{scope}[every node/.style={circle, draw, inner sep=0pt, outer sep=0pt, minimum size=7pt,
        node distance=20pt}]
        \node[below=0pt of lsC.south east,fill=color1] (l1C) {};
        \node[below left=of l1C,fill=color2] (l2C) {};
        \node[below right=of l1C,fill=color3] (l3C) {};
    \end{scope}
    \draw[-stealth] (l1C) -- (l3C);
    \node[below=0pt of l1C] {\scriptsize $3$};
    \node[below=0pt of l2C] {\scriptsize $2$};
    \node[below=0pt of l3C] {\scriptsize $1$};
\end{tikzpicture}

%% file: figures/conversionTreeGrammarEx.tex
\begin{tikzpicture}
    \node (Arule) {$A \rightarrow$};
    \node[below=3 of Arule] (Brule) {$B \rightarrow$};
    \node[right=4 of Arule] (C1rule) {$C_1 \rightarrow$};
    \node[below=2 of C1rule] (C2rule) {$C_2 \rightarrow$};
    \begin{scope}[every node/.style={circle, draw, inner sep=0pt, outer sep=0pt, minimum size=7pt,
        node distance=20pt}]
        \node[right=1 of Arule] (A2) {};
        \node[above=of A2, fill=color1] (A1) {};
        \node[below=of A2, fill=color2] (A3) {};
        \node[right=of A3, fill=color3] (A4) {};
    \end{scope}
    \begin{scope}[every node/.style={rectangle, draw, inner sep=2pt, outer sep=0pt, minimum size=12pt,
        node distance=20pt}]
        \node[left=.5 of A2] (AB) {$B$};
%        \node[right=.5 of A2] (AC2) {$C_2$};
        \node[below=.5 of A3] (AC1) {$C_1$};
    \end{scope}
    \begin{scope}[every path/.style={}]
        \draw (AB) --node[pos=.3,above left=-4pt] {\scriptsize $1$} (A1);
        \draw (AB) --node[pos=.3,above=-2pt] {\scriptsize $2$} (A2);
        \draw (AB) --node[pos=.3,below left=-4pt] {\scriptsize $3$} (A3);
        \draw (AC1) --node[pos=.3,left=-2pt] {\scriptsize $1$} (A3);
%        \draw (AC2) --node[pos=.3,above=-2pt] {\scriptsize $1$} (A2);
%        \draw (AC2) --node[pos=.3,left=-2pt] {\scriptsize $2$} (A4);
        \draw[-stealth'] (A2) --node[above right=-4pt] {$C_2$} (A4);
    \end{scope}
    \node[above=0pt of A1] {\scriptsize $1$};
    \node[right=0pt of A3] {\scriptsize $2$};
    \node[right=0pt of A4] {\scriptsize $3$};

    \begin{scope}[every node/.style={circle, draw, inner sep=0pt, outer sep=0pt, minimum size=7pt,
        node distance=20pt}]
        \node[right=0pt of Brule, fill=color1] (B1) {};
        \node[above right=.5 and 1.5 of B1, fill=color2] (B2) {};
        \node[below right=.5 and 1.5 of B1] (B3) {};
        \node[right=1 of B3, fill=color3] (B4) {};
    \end{scope}
    \begin{scope}[every node/.style={rectangle, draw, inner sep=2pt, outer sep=0pt, minimum size=12pt,
        node distance=20pt}]
        \node[right=of B1] (Bf) {$f$};
    \end{scope}
    \draw (Bf) --node[pos=.2,above=-1pt] {\scriptsize 1} (B1);
    \draw (Bf) --node[pos=.2,above=-1pt] {\scriptsize 3} (B2);
    \draw (Bf) --node[pos=.2,below=-1pt] {\scriptsize 2} (B3);
    \draw[-stealth] (B3) --node[above] {$a$} (B4);
    \node[above=0pt of B1] {\scriptsize $1$};
    \node[above=0pt of B2] {\scriptsize $2$};
    \node[above=0pt of B4] {\scriptsize $3$};

    \begin{scope}[every node/.style={circle, draw, inner sep=0pt, outer sep=0pt, minimum size=7pt,
        node distance=20pt}]
        \node[right=0pt of C1rule, fill=color2] (C11) {};
    \end{scope}
    \begin{scope}[every node/.style={rectangle, draw, inner sep=2pt, outer sep=0pt, minimum size=12pt,
        node distance=20pt}]
        \node[below=.25 of C11] (Cb1) {$b$};
    \end{scope}
    \draw (Cb1) -- (C11);
    \node[right=0pt of C11] {\scriptsize $1$};

    \begin{scope}[every node/.style={circle, draw, inner sep=0pt, outer sep=0pt, minimum size=7pt,
        node distance=20pt}]
        \node[right=0pt of C2rule, fill=color1] (C21) {};
        \node[below=of C21, fill=color3] (C22) {};
    \end{scope}
    \begin{scope}[every node/.style={rectangle, draw, inner sep=2pt, outer sep=0pt, minimum size=12pt,
        node distance=20pt}]
        \node[below=.25 of C22] (Cb2) {$b$};
    \end{scope}
    \draw (Cb2) -- (C22);
    \draw[-stealth] (C21) --node[right=-1pt] {$a$} (C22);
    \node[right=0pt of C21] {\scriptsize $1$};
    \node[right=0pt of C22] {\scriptsize $2$};
\end{tikzpicture}

%% file: figures/node_order_example.tex
\begin{tikzpicture}[node distance=25pt]
    \begin{scope}[every node/.style={inner sep=0pt, outer sep=0pt, minimum size=7pt, draw, circle}]
        \node (1c) {};        
        \node[left=of 1c] (1ol) {};        
        \node[right=of 1c] (1or) {};        
        \node[below=of 1c] (1b1) {};        
        \node[left=of 1b1] (1b2) {};        
    \end{scope}
    \draw (1c) -- (1ol);
    \draw (1c) -- (1or);
    \draw (1c) --node[coordinate] (x1) {} (1b1);
    \draw (1b1) -- (1b2);
    {\footnotesize
    \node[above=1pt of 1ol] {$1$};
    \node[above=1pt of 1or] {$1$};
    \node[above=1pt of 1c] {$3$};
    \node[below=1pt of 1b1] {$2$};
    \node[below=1pt of 1b2] {$1$};
    }

    \begin{scope}[every node/.style={inner sep=0pt, outer sep=0pt, minimum size=7pt, draw, circle}]
        \node[right=4.5 of 1c] (2c) {};        
        \node[left=of 2c] (2ol) {};        
        \node[right=of 2c] (2or) {};        
        \node[below=of 2c] (2b1) {};        
        \node[left=of 2b1] (2b2) {};        
    \end{scope}
    \draw (2c) -- (2ol);
    \draw (2c) -- (2or);
    \draw (2c) --node[coordinate] (x2) {} (2b1);
    \draw (2b1) -- (2b2);
    {\footnotesize
    \node[above=1pt of 2ol] {$(1,3)$};
    \node[above=1pt of 2or] {$(1,3)$};
    \node[above=1pt of 2c] {$(3,1,1,2)$};
    \node[below=1pt of 2b1] {$(2,1,3)$};
    \node[below=1pt of 2b2] {$(1,2)$};
    }

    \begin{scope}[every node/.style={inner sep=0pt, outer sep=0pt, minimum size=7pt, draw, circle}]
        \node[right=4.5 of 2c] (3c) {};        
        \node[left=of 3c] (3ol) {};        
        \node[right=of 3c] (3or) {};        
        \node[below=of 3c] (3b1) {};        
        \node[left=of 3b1] (3b2) {};        
    \end{scope}
    \draw (3c) -- (3ol);
    \draw (3c) -- (3or);
    \draw (3c) -- (3b1);
    \draw (3b1) -- (3b2);
    {\footnotesize
    \node[above=1pt of 3ol] {$2$};
    \node[above=1pt of 3or] {$2$};
    \node[above=1pt of 3c] {$4$};
    \node[below=1pt of 3b1] {$3$};
    \node[below=1pt of 3b2] {$1$};
    }

    \node[right=1.5 of x1] {$c_0 \Rightarrow f_0$};
    \node[right=1.5 of x2] {$f_0 \Rightarrow c_1$};
\end{tikzpicture}

%% file: implementation.tex
\subsection{Implementation Details}
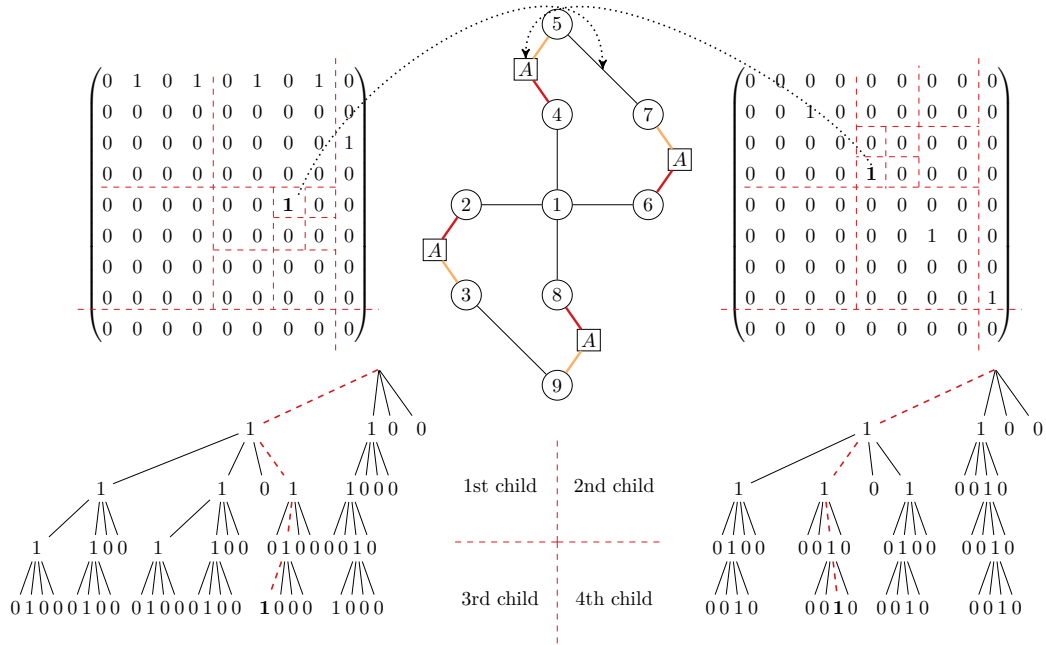
\begin{figure*}[!t]
    \centering
    \input{figures/output_example_ktree}
    \caption{Start graph (middle): terminal edges (left) and nonterminal edges
      (right) and their \ktree representations (below) with $k=2$.}
    \label{fig:output_example_ktree}
\end{figure*}
In this section we describe some of the technical details of our implementation. 
We outline the involved data structures, and describe our
output format.% which uses \ktree{s}.
\subsubsection{Data Structures}\label{ss2:compression}
Our data structures are a direct generalization to graphs of the
data structures used for strings~\cite{Larsson00_rePair} and
trees~\cite[Figure~11]{Lohrey13_treeRePair}. The occurrences are managed using doubly linked lists
for every active digram. Of importance is a priority queue, which uses the frequency of a digram as
the priority. Following Larsson and Moffat~\cite{Larsson00_rePair} the length of this queue is
chosen as $\sqrt{n}$, where $n$ is the number of edges of the original input graph to \graphrepair.
\subsubsection{Grammar Representation}
We encode the start graph and the productions in different ways. As an example, consider again the
grammar in Figure~\ref{fig:grammar_example}. The start graph is encoded using
\ktree{s}~\cite{DBLP:journals/is/BrisaboaLN14}, using $k=2$ as this provides the best compression.
This data structure partitions the adjacency matrix into $k^2$
squares and represents it in a $k^2$-ary tree. Consider the left adjacency matrix in
Figure~\ref{fig:output_example_ktree}. The $9\times 9$-matrix is first expanded with 0-values to the
next power of two; i.e., $16 \times 16$. If one partition has only $0$-entries, a leaf labeled $0$ is
added to the tree. This happens for the 3rd and 4th partition in this case (the partitions are
numbered left to right, top to bottom as indicated in the bottom center of the figure). Thus the 3rd
and 4th child of the root are $0$-leafs. The other two have at least one $1$-entry, therefore inner
nodes labeled $1$ are added and the square is again partitioned into $k^2$ squares. This is
continued at most until every square covers exactly one value. At this point the values are added to
the tree as leafs. As we need to consider edges with different labels, we use a method similar
to the representation of RDF graphs proposed in~\cite{DBLP:journals/kais/Alvarez-GarciaB15}. Let
$E_{\sigma} \subseteq E$ be the set of all edges labeled $\sigma$. For every label $\sigma$
appearing in $S$ we encode the subgraph $(V, E_\sigma)$. If $\rank(\sigma) = 2$, then this is
encoded as an adjacency matrix. Otherwise we use an incidence matrix, i.e., a matrix that has one
row for every edge and a column for every node. Thus, a 1 in row $i$, column $j$ of the incidence
matrix means, that edge $i$ is attached to node $j$. All of these matrices are encoded as
\ktree{s}. Figure~\ref{fig:output_example_ktree} is an example with two edge labels. Note that this
example only uses edges of rank $2$. For a hyperedge $e$, the incidence matrix only provides
information on the set of nodes attached to $e$, but not the specific order of $\att(e)$. For this
reason we also store a permutation for every edge to recover $\att(e)$. We count the number of
distinct such permutations appearing in the grammar and assign a number to each. Then we store the
list encoded in a $\lceil\log n\rceil$-fixed length encoding, where $n$ is the number of distinct
permutations.
%Note that we also rename the nodes of the startgraph, such that their IDs correspond
%to the order used, when initially traversing the graph to find the digram occurrences. This
%experimentally improved the compression in almost all the graphs in our data set.

For the productions we use a different format, as we expect the right-hand-sides to be very small
graphs (due to pruning, they may be larger than just a digram). We store an edge list for every
production, encoding the nodes using a variable-length
$\delta$-code~\cite{DBLP:journals/tit/Elias75}. One more bit per node is used to mark external
nodes. As the order of the external nodes is also important, we make sure that the order induced by
the IDs of the external nodes is the same as the order of the external nodes. Every production
begins with the edge count (again, using $\delta$-codes). For every edge we first use one bit to
mark terminal/nonterminal edges, then store the number of attached nodes, followed by the
$\delta$-codes of the list of IDs. Finally, we also use a $\delta$-code for the edge label. For the
production in Figure~\ref{fig:grammar_example} this leads to the following encoding:
\[
    \begin{array}{ll}
        \delta(2) & \text{two edges} \\
        0\delta(2) & \text{edge is terminal ($0$), has two nodes} \\
        1\delta(1)1\delta(2)\delta(1) & \text{nodes 1 (external) and 2 (external), label 1} \\
        0\delta(2) & \text{next edge is terminal, has two nodes} \\
        1\delta(1)0\delta(3)\delta(1) & \text{nodes 1 (external) and 3 (internal), label 1}
    \end{array}
\]
This is a bit sequence of length 28.
%
%Note that, as mentioned in Section~\ref{sss:nodeIDs} the grammar only produces an isomorphic copy of
%the original graph, but always the same copy. For derivation, the order of the edges is taken as
%lexicographically by their $\att$-relation. This is true, the implementation makes it somewhat
%easier, as this  For derivation, the order of the edges is taken as
%they are given in the file, i.e., rank 2 edges are ordered using the source node and the target
%node, if the source node is the same. Edges of higher rank are ordered by their row number in the
%incidence matrix, and edges in a rule are ordered by their position within that rule.
\subsection{On the Choice of Grammar Formalism}\label{sss:formalism_choice}
We discuss briefly the reasoning for some of our choices regarding the grammar formalism used. This
includes the choice of hyperedge replacement over a different replacement method, and the
restrictions we further enforce on hypergraphs.
\subsubsection{Hyperedge vs. Node Replacement}
There are two well-known types of context-free graph grammars:
\begin{itemize}
    \item context-free hyperedge replacement grammars (HR grammars for short,
        see~\cite{Engelfriet:1997:CGG:267871.267874}), and
    \item context-free node replacement grammars (NR grammars for short,
        see~\cite{DBLP:conf/gg/EngelfrietR97}).
\end{itemize}
In terms of graph language generating power, NR grammars are strictly more expressive than HR
grammars. For instance, NR grammars can describe the set of all complete graphs, whereas this is not
possible with HR grammars. If we use NR grammars to produce bipartite graphs that encode
hypergraphs, then the resulting hypergraph language generating power is exactly the same as that of
HR grammars (see~\cite[Theorem 4.28]{Engelfriet:1997:CGG:267871.267874}). 
Note that for any given SL-HR grammar one can construct an equivalent SL-NR grammar of similar size.

The difference in expressive power comes due to node replacement grammars using a different
formalism for derivation. Instead of merging external nodes with the nodes in the nonterminals
neighborhood, every rule in an NR-grammar includes a \emph{connection relation}, which specifies,
which nodes in the rule are to be connected to which nodes in the nonterminal's neighborhood.
Consider as an example the grammar given in Figure~\ref{fig:NRcomplete}: the initial graph $S$ has
two nonterminal nodes, both labeled $A$. The rule for $A$ has two nonterminal nodes labeled $B$, and
the connection relation $\{(A,B), (C,B), (D,B), (a,B)\}$. When deriving one of the $A$-nodes in $S$
we note that the neighborhood of this node consists of only one node with label $A$. The tuple
$(A,B)$ in the connection relation means that every $A$ node from the nonterminals neighborhood is
connected with every $B$-node in the rule. Figure~\ref{fig:twoStepsNR} shows this, and one more
derivation step as an example.

The natural choice to define a digram for use of RePair with NR-grammars would be two neighboring
nodes and the edge between them. However, due to the connection relation, we would also need to
include the neighborhood in this definition, as different occurrences of the same two nodes may need
different connection relations. These may be incompatible with each other in some cases, or could be
merged into one rule in others. This could make digrams too specific, and occurrences somewhat rare.
It is therefore unclear, whether the use of node replacement yields a successful RePair variant, but
we note that it is interesting future work, particularly as we believe, that complete graphs can be
compressed much better using NR grammars.
\begin{conjecture}
    Let $C_n$ be a complete graph with $2^n$ nodes. Then there exists an SL-NR grammar $G_n$ such
    that $\val(G_n) = C_n$, and $|G_n| \in O(n)$, but there is no SL-HR grammar $G_n'$ such that
    $\val(G_n') = C_n$ and $|G_n| \in O(n)$.
\end{conjecture}
The previously mentioned grammar in Figure~\ref{fig:NRcomplete} is an example for an SL-NR grammar
that compresses a complete graph as in the conjecture and can be extended to generate $C_n$ for any
$n$. The connection relation for this grammar could also use a shorthand (i.e., $\{(*,a)\}$ could be
used for the connection relation of the $D$-rule instead of $\{(A,a), (B,a), (C,a), (D,a)\}$). Doing
so yields a grammar with $n$ rules, each having two nodes, one edge, and one tuple in the connection
relation. We currently have no proof for the second part of the conjecture.
\begin{figure}[!t]
    \centering
    \input{figures/NRComplete}
    \caption{NR grammar that generates a complete graph with 32 nodes.}
    \label{fig:NRcomplete}
\end{figure}
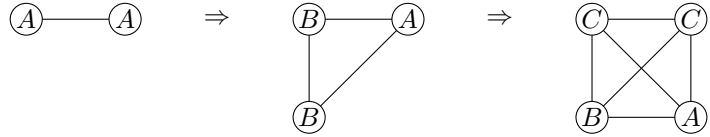
\begin{figure}[!t]
    \centering
    \input{figures/twoStepsNR}
    \caption{Two derivation steps of the grammar from Figure~\ref{fig:NRcomplete}.}
    \label{fig:twoStepsNR}
\end{figure}
\subsubsection{On the Conditions \ref{condition1} and \ref{condition2}}
Recall, we also enforce two conditions on hypergraphs, which are not always found in the literature:
\begin{itemize}
    \item[\ref{condition1}] for every edge $e$: $\att(e)$ contains no node twice, and
    \item[\ref{condition2}] the string $\ext$ of external nodes contains no node twice.
\end{itemize}
We refer to these as att-, and ext-distinctness, respectively. We also call a rule att-distinct
(ext-distinct), if its right-hand side is att-distinct (ext-distinct). A grammar is att-distinct
(ext-distinct), if every rule and the start graph are att-distinct (ext-distinct). As already
mentioned in Section~\ref{sse:preliminaries} these conditions have no effect on the graph language
generating power of HR grammars (provided the graphs do not require edges/external nodes that are
not att/ext-distinct). For efficiency reasons \graphrepair does not consider or
produce edges that are not att-distinct. This potentially weakens the compression, as nonterminal
edges that are not att-distinct could be used to represent occurrences of different digrams using
only one rule. For example, consider the digrams \textit{g}) in Figure~\ref{fig:all_digrams}. Occurrences
of them could also be replaced by a rule using the digrams \textit{a}) from that same figure, by merging
the center and right-most external nodes. This is not necessarily always an
improvement: in this case we are using a rank-3 nonterminal edge, where a rank-2 edge would have
sufficed, thus storing larger nonterminals. It can be shown however, that there are grammars, which
are not att-distinct, but smaller than the smallest att-distinct grammar for the same graph.
Figure~\ref{fig:attDistinct_example} is an example for such a graph. We leave it as an exercise to
calculate the maximal impact of att-distinctness with respect to compression. For
ext-distinctness, on the other hand, we can show that the condition not only has no effect on
compression, but enforcing it makes the grammar \emph{smaller}.
\begin{figure}[!t]
    \centering
    \input{figures/att_example}
    \caption{Example for a grammar that is not att-distinct (left), but smaller than the
        att-distinct grammar produced by \graphrepair (right) representing the same graph
    (middle).}
    \label{fig:attDistinct_example}
\end{figure}
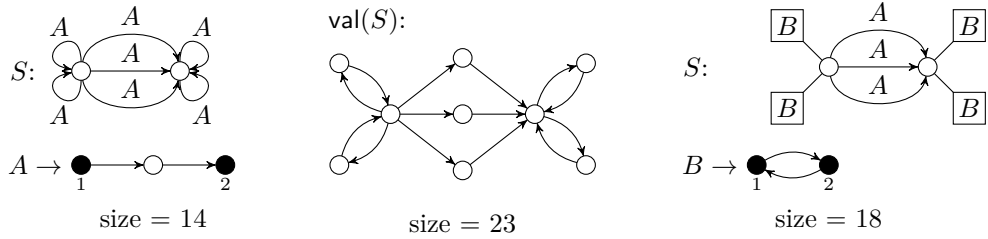
\begin{lemma}
    Given a non-ext-distinct SL-HR grammar $G = (N, P, S)$, an ext-distinct SL-HR grammar $G'$ can
    be constructed (in linear time), such that $L(G') = L(G)$ and $|G'| < |G|$.
    \label{lem:distinct_ext}
\end{lemma}
\begin{proof}
    For a string $w$ let $\symb(w)$ be the set of symbols appearing in $w$, and $\dist(w)$ the
    string $v$ derived from $w$ by only keeping the first occurrence of every distinct symbol in
    $w$. For example, for $w = acbac$, $\symb(w) = \{a,b,c\}$, and $\dist(w) = acb$. Let $(A,g)$ be
    a rule that is not ext-distinct. We first add a new rule $(A', g')$ with $\rank(A') =
    |\dist(\ext_g)|$, and $g'$ is the same graph as $g$ except with $\ext_{g'} = \dist(\ext_g)$.
    Then we replace the $A$-rule by $(A, h)$ where $h = (\symb(\ext_g), \{e\}, \lbl(e) = A', \att(e)
    = \dist(\ext_g), \ext = \ext_g)$. We can now derive every occurrence of $A$ in the grammar and
    remove the $A$-rule altogether. Doing this merges some of the nodes in the grammar, as they were
    referenced by the same external node. Note that this process reduces the size of the grammar.
    Let $u$ be a node that occurs $k \geq 2$ times in $\ext_g$. Then, applying the $A$-rules
    decreases the node size by $k-1$. The edge-size also decreases by $k-1$, as $A'$ has a
    smaller rank.
%    
%    
%    Let
%    $(A,g)$ be a rule. Then we define the two graphs 
%    \begin{align*}
%        \extHandle(A) &= (\symb(\ext_g), \{e\}, \att(e) = \dist(\ext_g), \lbl(e) = A, \ext_g) \\
%        \distGraph(g) &= (V_g, E_g, \att_g, \lbl_g, \dist(\ext_g)).
%    \end{align*}
%    We now acquire $G'$ by doing the following for every rule $(A,g) \in P$:
%    \begin{enumerate}
%        \item Remove $(A,g)$ from $P$.
%        \item Relabel every nonterminal edge labeled $A$ with $A'$.
%        \item Add the two rules $(A', \extHandle(A))$, and $(A, \distGraph(g))$.
%        \item Derive every nonterminal edge labeled $A'$, using the semantics from~\cite{Engelfriet:1997:CGG:267871.267874}.
%        \item Remove the nonterminal $A'$ with its associated rule, it is no longer needed.
%    \end{enumerate}
%    In effect, this procedure merges the nodes where necessary in the grammar (the $A'$-derivation
%    makes sure of this) and thus replaces every rule that is not ext-distinct with an equivalent
%    rule that is.
\end{proof}

%% file: figures/output_example_ktree.tex
\scalebox{.75}{
\begin{tikzpicture}[
        emph/.style={edge from parent/.style={color1,draw,thick, dashed}},
        norm/.style={edge from parent/.style={black,draw,thin, solid}}
    ]
    \node (termMatrix) {
%        \footnotesize
        $
            \begin{pmatrix}
               0 & 1 & 0 & 1 & 0 & 1 & 0 & 1 & 0 \\
               0 & 0 & 0 & 0 & 0 & 0 & 0 & 0 & 0 \\
               0 & 0 & 0 & 0 & 0 & 0 & 0 & 0 & 1 \\
               0 & 0 & 0 & 0 & 0 & 0 & 0 & 0 & 0 \\
               0 & 0 & 0 & 0 & 0 & 0 & \mathbf{1} & 0 & 0 \\
               0 & 0 & 0 & 0 & 0 & 0 & 0 & 0 & 0 \\
               0 & 0 & 0 & 0 & 0 & 0 & 0 & 0 & 0 \\
               0 & 0 & 0 & 0 & 0 & 0 & 0 & 0 & 0 \\
               0 & 0 & 0 & 0 & 0 & 0 & 0 & 0 & 0 \\
            \end{pmatrix}
        $
    };
    \node[coordinate] at ($ (termMatrix.center) + (3.5em, 5pt) $) (m1) {};
    \node[right=6 of termMatrix] (AMatrix) {
%        \footnotesize
        $
            \begin{pmatrix}
               0 & 0 & 0 & 0 & 0 & 0 & 0 & 0 & 0 \\
               0 & 0 & 1 & 0 & 0 & 0 & 0 & 0 & 0 \\
               0 & 0 & 0 & 0 & 0 & 0 & 0 & 0 & 0 \\
               0 & 0 & 0 & 0 & \mathbf{1} & 0 & 0 & 0 & 0 \\
               0 & 0 & 0 & 0 & 0 & 0 & 0 & 0 & 0 \\
               0 & 0 & 0 & 0 & 0 & 0 & 1 & 0 & 0 \\
               0 & 0 & 0 & 0 & 0 & 0 & 0 & 0 & 0 \\
               0 & 0 & 0 & 0 & 0 & 0 & 0 & 0 & 1 \\
               0 & 0 & 0 & 0 & 0 & 0 & 0 & 0 & 0 \\
            \end{pmatrix}
        $
    };
    \node[coordinate] at ($ (AMatrix.center) + (0em, 2em) $) (m2) {};

    \begin{scope}[every node/.style={circle, draw, inner sep=0pt, minimum size=15pt, node
        distance=30pt}]
        \node[right=5.5 of termMatrix.center] (c) {1};
        \node[left=of c] (l) {2};
        \node[below=of l] (l2) {3};
        \node[above=of c] (a) {4};
        \node[above=of a] (a2) {5};
        \node[right=of c] (r) {6};
        \node[above=of r] (r1) {7};
        \node[below=of c] (b) {8};
        \node[below=of b] (b1) {9};
    \end{scope}

    \draw (c) -- (l);
    \path (l) --node[rectangle, draw, fill=white, inner sep=2pt, outer sep=0pt, left=10pt] (1) {$A$} (l2);
    \draw[color=color1, very thick] (l) -- (1);
    \draw[color=color2, very thick] (1) -- (l2);
    \draw (c) -- (a);
    \path (a) --node[rectangle, draw, fill=white, inner sep=2pt, outer sep=0pt, left=10pt] (2) {$A$} (a2);
    \draw[color=color1, very thick] (a) -- (2);                                             
    \draw[color=color2, very thick] (2) -- (a2);                                            
    \draw (c) -- (r);                                                               
    \path (r) --node[rectangle, draw, fill=white, inner sep=2pt, outer sep=0pt, right=10pt] (3) {$A$} (r1);
    \draw[color=color1, very thick] (r) -- (3);                                             
    \draw[color=color2, very thick] (3) -- (r1);                                            
    \draw (c) -- (b);                                                               
    \path (b) --node[rectangle, draw, fill=white, inner sep=2pt, outer sep=0pt, right=10pt] (4) {$A$} (b1);
    \draw[color=color1, very thick] (b) -- (4);
    \draw[color=color2, very thick] (4) -- (b1);
    \draw (a2) --node[coordinate] (e1) {} (r1);
    \draw (b1) -- (l2);

    \node[coordinate, below=2.5 of b1] (x) {};
    \draw[dashed,color1] (x) -- +(1.8,0);
    \draw[dashed,color1] (x) -- +(-1.8,0);
    \draw[dashed,color1] (x) -- +(0,1.8);
    \draw[dashed,color1] (x) -- +(0,-1.8);

    \node at ($ (x) + (-1,1) $) {1st child};
    \node at ($ (x) + (1,1) $) {2nd child};
    \node at ($ (x) + (-1,-1) $) {3rd child};
    \node at ($ (x) + (1,-1) $) {4th child};

    \draw[-stealth',dotted,thick] (m1) .. controls +(1,2.5) and +(0,2.5) .. (e1);
    \draw[-stealth',dotted,thick] (m2) .. controls +(-1,1.5) and +(0,2.5) .. (2);

    \begin{scope}[every node/.style={inner sep=0, outer sep=0,coordinate}]
        \draw[color1, dashed] (termMatrix.54) --node[pos=.44] (c1) {} (termMatrix.306);
        \draw[color1, dashed] (termMatrix.215) --node[pos=.45] (c2) {} (termMatrix.325);

        \draw[color1, dashed] (c1) --node[pos=.26] (c3) {} +(-4.2,0);
        \draw[color1, dashed] (c2) --node[pos=.25] (c4) {} +(0, 4.2);

        \draw[color1, dashed] (c3) --node[pos=.25] (c5) {} +(0,-2.15);
        \draw[color1, dashed] (c4) --node[pos=.75] (c6) {} +(2.15, 0);

        \draw[color1, dashed] (c5) -- +(1.1, 0);
        \draw[color1, dashed] (c6) -- +(0, 1.1);
    \end{scope}

    \begin{scope}[every node/.style={inner sep=0, outer sep=0,coordinate}]
        \draw[color1, dashed] (AMatrix.54) --node[pos=.44] (c1) {} (AMatrix.306);
        \draw[color1, dashed] (AMatrix.215) --node[pos=.45] (c2) {} (AMatrix.325);

        \draw[color1, dashed] (c1) --node[pos=.25] (c3) {} +(-4.2,0);
        \draw[color1, dashed] (c2) --node[pos=.77] (c4) {} +(0, 4.2);

        \draw[color1, dashed] (c3) --node[pos=.25] (c5) {} +(0,2.15);
        \draw[color1, dashed] (c4) --node[pos=.24] (c6) {} +(2.15, 0);

        \draw[color1, dashed] (c5) -- +(-1.1, 0);
        \draw[color1, dashed] (c6) -- +(0, -1.1);
    \end{scope}

    \tikzset{level 1/.style={sibling distance=1.5cm},
        level 2/.style={sibling distance=2.5cm}, 
        level 3/.style={sibling distance=1cm}, 
        level 4/.style={sibling distance=.25cm}, 
        level distance=30pt
    }
%    \footnotesize
    \node[below=.3 of termMatrix.south east,outer sep=0, inner sep=0] (rterm) {}
        child[emph] {node {1}
            child[norm, sibling distance=1.75cm] {node {1}
                child[sibling distance=0.75cm] {node {1}
                    child {node {0}}
                    child {node {1}}
                    child {node {0}}
                    child {node {0}}
                }
                child[sibling distance=0.25cm] {node {1}
                    child {node {0}}
                    child {node {1}}
                    child {node {0}}
                    child {node {0}}
                }
                child[sibling distance=.25cm] {node {0}}
                child[sibling distance=.25cm] {node {0}}
            }
            child[norm, sibling distance=1.0cm] {node {1}
                child[sibling distance=0.75cm] {node {1}
                    child {node {0}}
                    child {node {1}}
                    child {node {0}}
                    child {node {0}}
                }
                child[sibling distance=0.25cm] {node {1}
                    child {node {0}}
                    child {node {1}}
                    child {node {0}}
                    child {node {0}}
                }
                child[sibling distance=.25cm] {node {0}}
                child[sibling distance=.25cm] {node {0}}
            }
            child[norm, sibling distance=.5cm] {node {0}}
            child[emph, sibling distance=.5cm] {node {1}
                child[norm, sibling distance=.25cm] {node {0}}
                child[emph, sibling distance=.25cm] {node {1}
                    child[emph] {node {\textbf{1}}}
                    child[norm] {node {0}}
                    child[norm] {node {0}}
                    child[norm] {node {0}}
                }
                child[norm, sibling distance=.25cm] {node {0}}
                child[norm, sibling distance=.25cm] {node {0}}
            }
        }
        child[sibling distance=.25cm] {node {1}
            child[sibling distance=0.25cm] {node {1}
                child[sibling distance=.25cm] {node {0}}
                child[sibling distance=.25cm] {node {0}}
                child[sibling distance=.25cm] {node {1}
                    child {node {1}}
                    child {node {0}}
                    child {node {0}}
                    child {node {0}}
                }
                child[sibling distance=.25cm] {node {0}}
            }
            child[sibling distance=.25cm] {node {0}}
            child[sibling distance=.25cm] {node {0}}
            child[sibling distance=.25cm] {node {0}}
        }
        child[sibling distance=.5cm] {node {0}}
        child[sibling distance=.5cm] {node {0}};
    \node[below=.3 of AMatrix.310,outer sep=0, inner sep=0] (rnterm) {}
        child[emph] {node {1}
            child[norm,sibling distance=1.5cm] {node {1}
                child[sibling distance=.25cm] {node {0}}
                child[sibling distance=.25cm] {node {1}
                    child {node {0}}
                    child {node {0}}
                    child {node {1}}
                    child {node {0}}
                }
                child[sibling distance=.25cm] {node {0}}
                child[sibling distance=.25cm] {node {0}}
            }
            child[emph, sibling distance=1.5cm] {node {1}
                child[norm, sibling distance=.25cm] {node {0}}
                child[norm, sibling distance=.25cm] {node {0}}
                child[emph, sibling distance=.25cm] {node {1}
                    child[norm] {node {0}}
                    child[norm] {node {0}}
                    child[emph] {node {\textbf{1}}}
                    child[norm] {node {0}}
                }
                child[norm, sibling distance=.25cm] {node {0}}
            }
            child[norm, sibling distance=.25cm] {node {0}}
            child[norm, sibling distance=.5cm] {node {1}
                child[sibling distance=.25cm] {node {0}}
                child[sibling distance=.25cm] {node {1}
                    child {node {0}}
                    child {node {0}}
                    child {node {1}}
                    child {node {0}}
                }
                child[sibling distance=.25cm] {node {0}}
                child[sibling distance=.25cm] {node {0}}
            }
        }
        child[sibling distance=.5cm] {node {1}
            child[sibling distance=.25cm] {node {0}}
            child[sibling distance=.25cm] {node {0}}
            child[sibling distance=0.25cm] {node {1}
                child[sibling distance=.25cm] {node {0}}
                child[sibling distance=.25cm] {node {0}}
                child[sibling distance=.25cm] {node {1}
                    child {node {0}}
                    child {node {0}}
                    child {node {1}}
                    child {node {0}}
                }
                child[sibling distance=.25cm] {node {0}}
            }
            child[sibling distance=.25cm] {node {0}}
        }
        child[sibling distance=.5cm] {node {0}}
        child[sibling distance=.5cm] {node {0}};
\end{tikzpicture}
}

%% file: figures/NRComplete.tex
\begin{tikzpicture}
    \begin{scope}[node distance=10pt]
        \node (S) {$S:$};
        \node[below=of S] (A) {$A\rightarrow$};
        \node[below=of A] (B) {$B\rightarrow$};
        \node[below=of B] (C) {$C\rightarrow$};
        \node[below=of C] (D) {$D\rightarrow$};
    \end{scope}

    \begin{scope}[every node/.style={circle, draw, inner sep=0pt, outer sep=0pt, minimum size=12pt,
        node distance=25pt}]
        \node[right=2pt of S] (s1) {$A$};
        \node[right=of s1] (s2) {$A$};
        \node[right=0pt of A] (a1) {$B$};
        \node[right=of a1] (a2) {$B$};
        \node[right=0pt of B] (b1) {$C$};
        \node[right=of b1] (b2) {$C$};
        \node[right=0pt of C] (c1) {$D$};
        \node[right=of c1] (c2) {$D$};
        \node[right=0pt of D] (d1) {$a$};
        \node[right=of d1] (d2) {$a$};
    \end{scope}
    \node[right=of a2] (crA) {$\{(A,B), (C,B), (D,B), (a,B)\}$};
    \node[right=of b2] (crB) {$\{(A,C), (B,C), (D,C), (a,C)\}$};
    \node[right=of c2] (crC) {$\{(A,D), (B,D), (C,D), (a,D)\}$};
    \node[right=of d2] (crD) {$\{(A,a), (B,a), (C,a), (D,a)\}$};

    \draw (s1) -- (s2);
    \draw (a1) -- (a2);
    \draw (b1) -- (b2);
    \draw (c1) -- (c2);
    \draw (d1) -- (d2);
\end{tikzpicture}

%% file: figures/twoStepsNR.tex
\begin{tikzpicture}
    \begin{scope}[every node/.style={circle, draw, inner sep=0pt, outer sep=0pt, minimum size=12pt,
        node distance=25pt}]
        \node (A11) {$A$};
        \node[right=of A11] (A21) {$A$};

        \node[right=2 of A21] (B12) {$B$};
        \node[below=of B12] (B22) {$B$};
        \node[right=of B12] (A12) {$A$};

        \node[right=2 of A12] (C13) {$C$};
        \node[right=of C13] (C23) {$C$};
        \node[below=of C13] (B13) {$B$};
        \node[right=of B13] (A13) {$A$};
    \end{scope}

    \draw (A11) -- (A21);

    \draw (B12) -- (B22);
    \draw (B12) -- (A12);
    \draw (B22) -- (A12);

    \draw (C13) -- (C23);
    \draw (C13) -- (B13);
    \draw (C13) -- (A13);
    \draw (C23) -- (B13);
    \draw (C23) -- (A13);
    \draw (B13) -- (A13);

    \node[right=.7 of A21] {$\Rightarrow$};
    \node[right=.7 of A12] {$\Rightarrow$};
\end{tikzpicture}

%% file: figures/att_example.tex
\begin{tikzpicture}
    \begin{scope}[every node/.style={circle, draw, minimum size=7pt, outer sep=0pt, inner sep=0pt,
        node distance=30pt}]
        \node (S1) {};
        \node[right=of S1] (S2) {};
    \end{scope}
    \begin{scope}[every node/.style={circle, draw, minimum size=7pt, outer sep=0pt, inner sep=0pt,
        node distance=20pt}]
        \node[below=1 of S1, fill=black] (A1) {};
        \node[right=of A1] (A2) {};
        \node[right=of A2, fill=black] (A3) {};
    \end{scope}
    \node[left=1em of S1] (S) {$S$:};
    \node[left=0pt of A1] (A) {$A \rightarrow$};
    \begin{scope}[every path/.style={-stealth'}]
        \draw (S1) -- node[above] {$A$} (S2);
        \path (S1) edge[bend left=80] node[above] {$A$} (S2);
        \path (S1) edge[bend right=80] node[above] {$A$} (S2);
        \path (S1) edge[loop above,min distance=15pt, in=170, out=90, looseness=10] node[above] {$A$} (S1);
        \path (S2) edge[loop above,min distance=15pt, in=10, out=90, looseness=10] node[above] {$A$} (S2);
        \path (S1) edge[loop below,min distance=15pt, in=190, out=270, looseness=10] node[below] {$A$} (S1);
        \path (S2) edge[loop below,min distance=15pt, in=350, out=270, looseness=10] node[below] {$A$} (S2);
        \draw (A1) -- (A2);
        \draw (A2) -- (A3);
    \end{scope}
    \node[below=10pt of A2] (size1) {size = $14$};
    \node[below=-2pt of A1] {\scriptsize $1$};
    \node[below=-2pt of A3] {\scriptsize $2$};

    \begin{scope}[every node/.style={circle, draw, minimum size=7pt, outer sep=0pt, inner sep=0pt,
        node distance=20pt}]
        \node[above right=0.5 and 2 of A3] (val1) {};
        \node[above left=0.5 and 0.5 of val1] (val2) {};
        \node[below left=0.5 and 0.5 of val1] (val3) {};
        \node[right=of val1] (val6) {};
        \node[above=0.5 of val6] (val4) {};
        \node[below=0.5 of val6] (val5) {};
        \node[right=of val6] (val7) {};
        \node[above right=0.5 and 0.5 of val7] (val8) {};
        \node[below right=0.5 and 0.5 of val7] (val9) {};
    \end{scope}
    \begin{scope}[every path/.style={-stealth'}]
        \path (val1) edge[bend left] (val2);
        \path (val2) edge[bend left] (val1);
        \path (val1) edge[bend left] (val3);
        \path (val3) edge[bend left] (val1);
        \draw (val1) -- (val4);
        \draw (val1) -- (val5);
        \draw (val1) -- (val6);
        \draw (val4) -- (val7);
        \draw (val5) -- (val7);
        \draw (val6) -- (val7);
        \path (val7) edge[bend left] (val8);
        \path (val8) edge[bend left] (val7);
        \path (val7) edge[bend left] (val9);
        \path (val9) edge[bend left] (val7);
    \end{scope}
    \node[above right=5pt and -10pt of val2] (val) {$\val(S)$:};
    \node[below=10pt of val5] (size2) {size = $23$};

%    \begin{scope}[every node/.style={circle, draw, minimum size=7pt, outer sep=0pt, inner sep=0pt,
%        node distance=20pt}]
%        \node[right=1.5 of val9, fill=black] (BA1) {};
%        \node[right=of BA1] (BA2) {};
%        \node[right=of BA2, fill=black] (BA3) {};
%        \node[above=0.5 of BA1] (BS1) {};
%        \node[above=of BS1] (BS2) {};
%        \node[right=of BS1] (BS3) {};
%        \node[right=of BS3] (BS4) {};
%    \end{scope}
    \begin{scope}[every node/.style={circle, draw, minimum size=7pt, outer sep=0pt, inner sep=0pt,
        node distance=20pt}]
        \node[right=2 of val9, fill=black] (B1) {};
        \node[right=of B1, fill=black] (B2) {};
    \end{scope}
    \begin{scope}[every node/.style={circle, draw, minimum size=7pt, outer sep=0pt, inner sep=0pt,
        node distance=30pt}]
        \node[above=of B2] (BS1) {};
        \node[right=of BS1] (BS2) {};
    \end{scope}
    \begin{scope}[every node/.style={rectangle, draw, minimum size=12pt, outer sep=0pt, inner sep=0pt,
        node distance=20pt}]
        \node[above left=0.25 and 0.25 of BS1] (BE1) {$B$};
        \node[below left=0.25 and 0.25 of BS1] (BE2) {$B$};
        \node[above right=0.25 and 0.25 of BS2] (BE3) {$B$};
        \node[below right=0.25 and 0.25 of BS2] (BE4) {$B$};
    \end{scope}
    \node[left=3.8em of BS1] (BS) {$S$:};
    \node[left=0pt of B1] (B) {$B \rightarrow$};
    \begin{scope}[every path/.style={-stealth'}]
        \path (BS1) edge[bend left=80] node[above] {$A$} (BS2);
        \path (BS1) edge[bend right=80] node[above] {$A$} (BS2);
        \path (BS1) edge node[above] {$A$} (BS2);
        \path (B1) edge[bend left] (B2);
        \path (B2) edge[bend left] (B1);
    \end{scope}
    \draw (BS1) -- (BE1);
    \draw (BS1) -- (BE2);
    \draw (BS2) -- (BE3);
    \draw (BS2) -- (BE4);
    \node[below=10pt of B2] (size3) {size = $18$};
    \node[below=-2pt of B1] {\scriptsize $1$};
    \node[below=-2pt of B2] {\scriptsize $2$};
\end{tikzpicture}

%% file: experimental_results.tex
\section{Experimental results}\label{sse:experiments}
We implemented a prototype in Scala (version 2.11.7) using the Graph for
Scala library\footnote{http://www.scala-graph.org/} (version 1.9.4). The experiments are conducted
on a machine running Scientific Linux 6.6 (kernel version 2.6.32), with 2 Intel Xeon E5-2690 v2
processors at 3.00~GHz and 378~GB memory. As we are only evaluating a prototype, we do not mention
runtime or peak memory performance, as these can be largely improved by a more careful
implementation. We compare to the following compressors:
\begin{itemize}
    \item \ktree, for which we use our own Scala-implementation following the description
        in~\cite{DBLP:journals/is/BrisaboaLN14}, using the same binary format, however without the
        further optimization on compressing the leaf-level described there.
    \item The list merge (LM) algorithm by Grabowski and
        Bieniecki~\cite{DBLP:journals/dam/GrabowskiB14}. We use $64$ for their chunk size parameter,
        as in their paper.
%        as this provides the best compression ratio.
    \item The combination of dense substructure removal \cite{Buehrer08_scalableWebGraphCompression}
        and \ktree by Hern\'andez and Navarro~\cite{DBLP:journals/kais/HernandezN14} (HN). For the
        parameters to the algorithm we use $T=10$, $P= 2$, and $ES=10$, which are the parameters
        their experiments show to provide the best compression. Note that this compressor uses a
        \ktree implementation with all the optimizations.
\end{itemize}
The latter two implementations were provided by the authors. We also experimented with RePair on
adjacency lists by Claude and Navarro~\cite{Claude10_compactWebGraphRep}, but omit the results here,
because on all graphs we tested, stronger compression was achieved by another compared compressor.

As common in graph compression, we present the compression ratios in \emph{bpe} (bits per
edge). Note that our compressor reorders the nodes. We omit the space required to retain the
original node IDs, because we assume that they represent arbitrary data values and it is possible to
update this mapping. This is particularly true for RDF graphs, as explained in
Section~\ref{ss2:rdf_results}.
%Note that our results are not perfectly comparable, as we do not include the space required to
%retain the original node IDs, whereas the other three methods can do so. However, in particular for
%RDF we do not consider this a limitation, as explained in Section~\ref{ss2:rdf_results}.

Note that it is known that using different node orders will change the results for all of these
compressors, including \graphrepair. We discussed the importance of the traversal order, but another
area where the node order will have an effect is the final encoding of the start graph using the
\ktree-structure. To get comparable results, we use the natural order in all cases, when it comes to
the encoding. We ran experiments, where the start graph of the grammar produced by \graphrepair was
reordered according to the initially computed \FP-order. This improves some of the results, but we
consider it future work to test the impact of different orders on the encoding of the start graph
(cf the conclusions).
\subsection{Datasets}
We use three different types of graphs:
%\begin{enumerate}
%    \item
        network graphs (Table~\ref{tbl:networks_desc}),
%    \item
        RDF graphs (Table~\ref{tbl:rdf_desc}), and
%    \item
        version graphs (Table~\ref{tbl:versions_desc}).
%\end{enumerate}
Each table lists the numbers of nodes and edges and the number of equivalence classes of $\FPequiv$
(see Section~\ref{ss2:node_order}) of each graph. For RDF graphs we also list the number of 
edge labels (i.e., predicates) of each graph. Two of the version graphs also have labeled edges.
\begin{table}[!t]
\centering
\tblFontSize
\begin{tabularx}{\textwidth}{XRRR}\toprule
    Graph & $|V|$ & $|E|$ & $|[\FPequiv]|$ \\\midrule
    CA-AstroPh & \numprint{18772} & \numprint{396160} & \numprint{14742} \\
    CA-CondMat & \numprint{23133} & \numprint{186936} & \numprint{17135} \\
    CA-GrQc & \numprint{5242} & \numprint{28980} & \numprint{3394} \\
    Email-Enron & \numprint{36692} & \numprint{367662} & \numprint{5805} \\
    Email-EuAll & \numprint{265214} & \numprint{420045} & \numprint{28895} \\
    NotreDame & \numprint{325729} & \numprint{1497134} & \numprint{118264} \\
    Wiki-Talk & \numprint{2394385} & \numprint{5021410} & \numprint{566846} \\
    Wiki-Vote & \numprint{7115} & \numprint{103689} & \numprint{5806}\\\bottomrule
%    CA-DBLP & \numprint{1544071} & \numprint{6748226} & \numprint{1083970}
\end{tabularx}
\caption{Network graphs}
\label{tbl:networks_desc}
\end{table}
We give a short description of each graph: the network graphs are from the Stanford Large
Network Dataset Collection\footnote{\url{http://snap.stanford.edu/data/index.html}} and are
unlabeled. They are communication networks (Email-EuAll, Wiki-Vote, Wiki-Talk),
a web graph (NotreDame) and Co-Authorship networks (CA-AstroPh, CA-CondMat, CA-GrQc). Even if
they were advertised as undirected, we considered all of them to be lists of directed edges, to
improve the comparability with the other compressors, as these would also assume the input graph to
be directed. 
%The exception is the CA-DBLP co-authorship graph, which
%we created ourselves from the DBLP XML\footnote{\url{http://dblp.uni-trier.de/xml/} (Release from
%August 1st, 2015)} by using the author ids as nodes and creating an edge between two authors if there is
%any publication they collaborated on. 
\begin{table}[!t]
\centering
\tblFontSize
\begin{tabularx}{\textwidth}{p{10em}RRRR}\toprule
    Graph & $|V|$ & $|E|$ & $|\Sigma|$ & $|[\FPequiv]|$\\\midrule
    1 Specific properties en & \numprint{609014} & \numprint{819764} & \numprint{71} &
    \numprint{236235} \\
    2 Types ru & \numprint{642340} & \numprint{642364} & \numprint{1} & \numprint{79} \\
    3 Types es & \numprint{818657} & \numprint{819780} & \numprint{1} & \numprint{336} \\
    4 Types de with en & \numprint{618708} & \numprint{1810909} & \numprint{1} & \numprint{335} \\
    5 Identica & \numprint{16355} & \numprint{29683} & \numprint{12} & \numprint{14588} \\
    6 Jamendo & \numprint{438975} & \numprint{1047898} & \numprint{25} &
    \numprint{396725}\\\bottomrule
\end{tabularx}
\caption{RDF graphs}
\label{tbl:rdf_desc}
\end{table}

The RDF graphs mostly come from the DBPedia
project\footnote{\url{http://wiki.dbpedia.org/Downloads2015-04}}, which is an effort of representing
ontology information from Wikipedia. We evaluate on \emph{specific
mapping-based properties} (English), which contains infobox data from the English Wikipedia and
\emph{mapping-based types}, which contains the rdf:types for the instances extracted from the
infobox data. We use three different versions of the latter: types for instances extracted from the
Spanish and Russian Wikipedia pages that do not have an equivalent English page, and types for
instances extracted from the German Wikipedia pages that do have an equivalent English page. The
Identica-dataset\footnote{\url{http://www.it.uc3m.es/berto/RDSZ/}} represents messages from the
public stream of the microblogging site identi.ca. Its triples map a notice or user with predicates
such as creator (pointing to a user), date, content, or name. 
The Jamendo-dataset\footnote{\url{http://dbtune.org/jamendo/}} is a linked-data representation of
the Jamendo-repository for Creative Commons licensed music. Subjects are artists, records, tags,
tracks, signals, or albums. The triples connect them with metadata such as names, birthdate,
biography, or date.
\begin{table}[!t]
\centering
\tblFontSize
\begin{tabularx}{\textwidth}{XRRRR}\toprule
    Graph & $|V|$ & $|E|$ & $|\Sigma|$ & $|[\FPequiv]|$ \\\midrule
    Tic-Tac-Toe & \numprint{5634} & \numprint{10016} & \numprint{3} & \numprint{9} \\
    Chess & \numprint{76272} & \numprint{113039} & \numprint{12} & \numprint{74592} \\
    DBLP60-70 & \numprint{24246} & \numprint{23677} & \numprint{1} & \numprint{2739} \\
    DBLP60-90 & \numprint{658197} & \numprint{954521} & \numprint{1} &
    \numprint{207305}\\\bottomrule
\end{tabularx}
\caption{Version graphs}
\label{tbl:versions_desc}
\end{table}

Version graphs are disjoint unions of multiple versions of the same graph. Here, Tic-Tac-Toe
represents winning positions, and Chess represents legal moves\footnote{Both from
\url{http://ailab.wsu.edu/subdue/download.htm}}. The files contain node labels from a
finite alphabet, which we ignore here. DBLP60-70 and DBLP60-90 are co-authorship networks from DBLP,
created from the XML\footnote{\url{http://dblp.uni-trier.de/xml/} (release from August 1st, 2015)}
file by using author IDs as nodes and creating an edge between two authors who appear as co-authors
of some entry in the file. To make version graphs, we created graphs containing the disjoint union
of yearly snapshots of the co-authorship network.
\subsection{Influence of Parameters}
We evaluate how the different parameters for our compressor affect compression. For
these experiments every parameter except the one being evaluated is fixed for the runs. Note that
this sometimes leads to situations where none of the results in a particular experiment represents
the best compression our compressor is able to achieve for the given graph. The parameters evaluated
are
%\begin{itemize}
%    \item 
        the maximum rank of a nonterminal and
%    \item
        the node order.
%\end{itemize}
\subsubsection{Maximum Rank} We test maxRank values from 2 up to 8. The average results for the
three types of graphs tested are given in Table~\ref{tbl:maxRank_results}, as compression in bpe.
We did some tests for higher values (up to 16) but only got worse results. The best results are
marked in bold. For each class of graphs, a value of $4$ achieved the best results on average. We
therefore conclude that a value of $4$ is a good compromise for our data sets.
%On average the best results are achieved with a value of $3$ for network graphs, $2$
%for RDF graphs, and $4$ for version graphs. In all cases, the average of a value of $4$ is very
%close to the best results. We therefore conclude that a value of $4$ is a good compromise for our
%data sets.
\begin{table}[!t]
    \centering
    \nprounddigits{2}
    \tblFontSize
    \begin{tabularx}{\textwidth}{lRRRRRRR}\toprule
        &2&3&4&5&6&7&8\\\midrule
        Network Graphs & \numprint{11.74} & \numprint{11.59} & \textbf{\numprint{11.44}} &
        \numprint{11.94} & \numprint{12.51} & \numprint{13.67} & \numprint{12.92} \\
        RDF	Graphs & \numprint{5.17} & \numprint{5.40} & \textbf{\numprint{4.85}} & \numprint{4.92} &
        \numprint{5.60} & \numprint{6.12} & \numprint{6.25} \\
        Version	Graphs & \numprint{7.99} & \numprint{8.21} & \textbf{\numprint{5.57}} & \numprint{5.93} &
        \numprint{6.07} & \numprint{6.12} & \numprint{6.11}\\\bottomrule
%        Network Graphs & \numprint{9.28} & \textbf{\numprint{9.24}} & \numprint{9.34} & \numprint{9.65} &
%        \numprint{10.01} & \numprint{10.64} & \numprint{10.23} \\
%        RDF Graphs & \textbf{\numprint{4.37}} & \numprint{4.61} & \numprint{4.60} & \numprint{4.62} &
%        \numprint{4.99} & \numprint{5.24} & \numprint{5.39} \\
%        Version Graphs & \numprint{8.66} & \numprint{9.14} & \textbf{\numprint{5.68}} & \numprint{6.03} &
%        \numprint{6.09} & \numprint{6.12} & \numprint{6.12} \\\bottomrule
    \end{tabularx}
    \caption{Average results for different values of maxRank (compression in bpe).}
    \label{tbl:maxRank_results}
\end{table}
\subsubsection{Node Order}\label{ss2:nodeOrder} Recall from Section~\ref{ss2:node_order} that the
\FP-order is a fixed point computation starting from the node degrees. As this is an iterative
process, it can be terminated at any point. We were interested how much difference a
fixpoint makes compared to using just the node degrees ($\FP0$).
Figure~\ref{fig:networks_orderComparison} shows the compression ratio of a selection of graphs under
the different node orders. The selection aims to be representative for the graphs of the types we
evaluated: CA-graphs behave similar to CA-AstroPh, version graphs similar to DBLP60-70, and the
RDF graphs similar to Specific properties en. The other graphs in the figure are chosen because they
are outliers in their respective category. Our \FP-order achieves the best result on almost all of the
graphs. On RDF graphs the order generally had only marginal impact: the best and worst results
usually are within 0.5~bpe of each other. The Jamendo graph presents an exception here, with the
natural order being about 1~bpe better than the closest other result and featuring the largest
difference between $\FP0$ and \FP of all our graphs. Version graphs however benefit
hugely from the \FP-order, as further discussed in Section~\ref{ss2:version_results}. This shows
that two or more versions of the same graph are similarly ordered in the \FP-order, increasing the
likelihood of the compressor of recognizing repeating structures.

There is another interesting observation about the \FP-order, or in particular the equivalence
relation $\FPequiv$. It is likely that nodes with a high similarity, i.e., the same neighborhood up
to a certain distance, are equivalent in this relation. This implies that graphs with a low number
of equivalence classes should compress well, as they would have many repeating substructures.
Figure~\ref{fig:bpevsequiv} shows this correlation. There is no graph in the lower right corner,
i.e., there is no graph with a low number of equivalence classes and low compression.
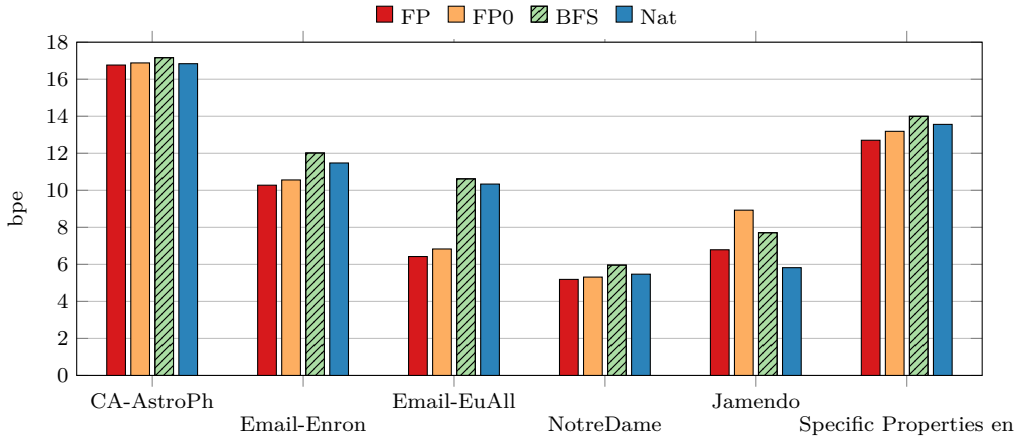
\begin{figure}[!t]
    \centering
    \input{figures/results_nodeorders}
    \caption{Performance of \graphrepair under different node orders.}
    \label{fig:networks_orderComparison}
\end{figure}
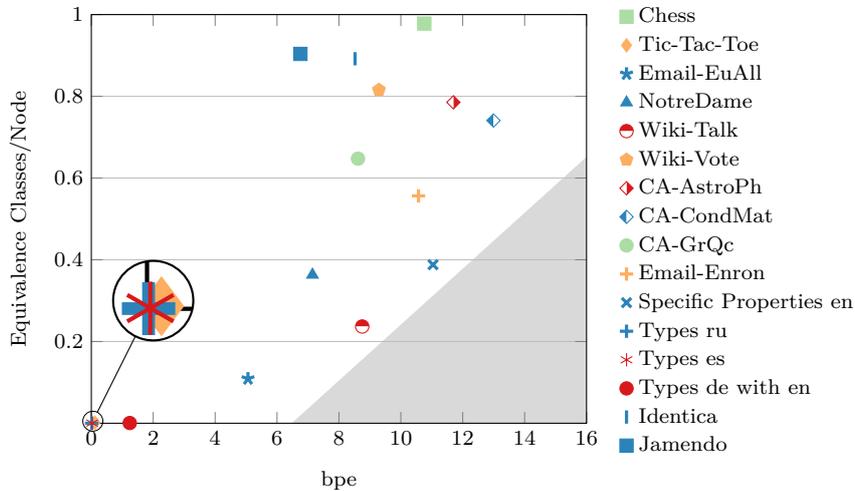
\begin{figure}[!t]
    \centering
    \input{figures/results_eqclasses}
    \caption{Correlation between equivalence classes of $\FPequiv$ and compression.}
    \label{fig:bpevsequiv}
\end{figure}
%\subsubsection{Pruning} The effects of pruning are unpredictable. As the definition of the $\sav$-value
%is entirely dependent on the size definition for hypergraphs, pruning also reduces the size of the
%grammar by this definition. It can do so quite significantly: about a third of the productions are
%commonly removed again in the pruning step, reducing the size by up to 40\%. However, this can
%cause adding nodes back into the start graph that were previously removed, and hence can lead to
%deeper \ktree representations as before, i.e., to a \emph{worse} compression ratio.
%
%Most network graphs compress best with pruning turned off, the NotreDame graph being the only
%exception. Pruning improves the results for some RDF graphs (all instance types graphs and
%identica) and worsens the results for the others. It tends to be beneficial for version graphs,
%though only with small differences.
\subsection{Comparison with other Compressors}\label{sss:comparison}
We compare \graphrepair with the compressors \ktree, LM, and HN listed at the beginning of
Section~\ref{sse:experiments}. Note that we compare RDF graph compression
only against the \ktree-method, as LM and HN have not been extended to RDF graphs. While these
algorithms all work as in-memory data structures, they produce outputs with file sizes comparable to
the in-memory representations. We measure the compression performance based on these file
sizes. Where applicable, we furthermore use the dense substructure removal done as a
first step in the HN-method (see Section~\ref{sse:related}) in combination with our compressor,
marked as ``\graphrepair{}+DSR''. To do so, we added specially labeled rank-1 edges to the virtual
nodes created by the dense substructure removal, to ensure that the original graph could be
restored. Usually, the virtual nodes are identified by having an ID greater than the largest ID
occurring in the original graph. As our compressor reorders the nodes, this is no longer feasible,
and marking them with additional edges was the most space-efficient method we found.

Let us give an idea of the compression ratio for the graphs according to our size definition. Using
\graphrepair (without DSR), we achieve, on average, a compression ratio ($\frac{|G|}{|g|}$) of
\begin{itemize}
    \item $68\%$ for network graphs, 
    \item $35\%$ for RDF, and 
    \item $24\%$ for version graphs.
\end{itemize}
The parameters we choose for \graphrepair are $\text{maxRank} = 4$ and the \FP-order, both being
generally the best choice for our dataset. We note that in most results the majority of the file
sizes of \graphrepair{'s} output ($>90\%$) is for the \ktree representation of the start graph.
\subsubsection{Network Graphs}
Our results on network graphs compared to \ktree, LM, and HN are summarized in
Figure~\ref{fig:compression_comparison_networks}. We improve on the plain \ktree-representation on
all graphs but NotreDame. However, our results are often slightly worse than LM and HN, with 
Email-EuAll and CA-GrQc being exceptions. That being said, dense substructure removal can
be combined with our compressor, using their dense substructure removal as a preprocessing step.
This generally improves on our results and achieves the smallest bpe-values for two of the three
CA-graphs.
%
%We also experimented with combining our approach with the dense substructure removal done by HN. The
%results are given as \graphrepair+DSR, and generally improve on \graphrepair without DSR and on HN,
%yielding the best result in all but two graphs in our dataset. As the node-IDs change in our
%approach and the virtual nodes added by dense substructure removal are only identified by their
%node-ID, and we cannot guarantee that these stay the same, we add additional rank-1 edges to the
%nodes to mark them as virtual.
%Doingso improves on
%their compression by up to 5~bpe.
%
%Keep in mind that none of the methods we compared to support labeled edges. Our compressor needs
%to support such by design and we therefore have to store additional information. Still, we have to
%conclude that our compressor is not the method of choice for unlabeled network graphs.
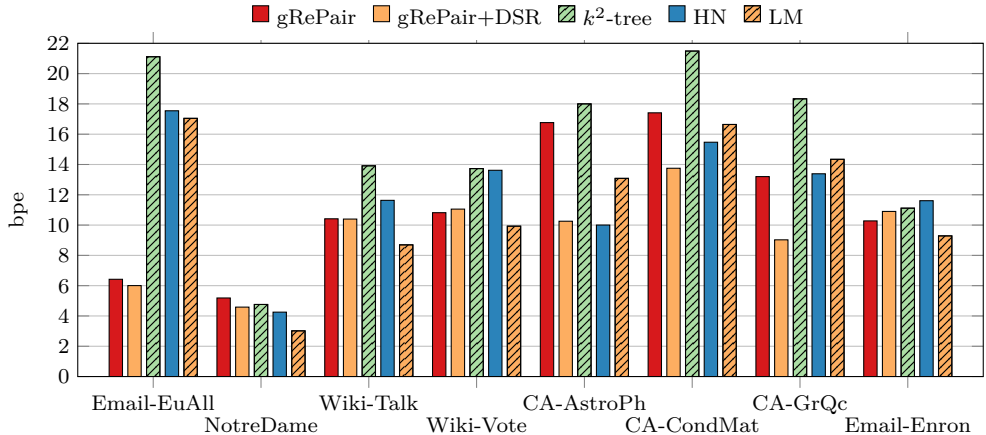
\begin{figure}[!t]
    \centering
    \input{figures/results_networks}
    \caption{Network graph comparison of \graphrepair with three other compressors.}
    \label{fig:compression_comparison_networks}
\end{figure}
\subsubsection{RDF Graphs}\label{ss2:rdf_results}
Recall from Section~\ref{sss:rdfGraphCompressionRelated} that the values for subject, predicate, and
object of RDF triples are commonly mapped to integers using a dictionary to represent the original
values. As in this way dictionary and graph are separate entities, we only focus on compressing the
graph. Any method for dictionary compression can be used to additionally compress the dictionary
(e.g.~\cite{DBLP:conf/sac/Martinez-PrietoFC12}) and we omit the space necessary for the dictionary.

Our results in comparison to \ktree are given in Table~\ref{tbl:compression_comparison_rdf}. We
greatly improve against this representation. For the graphs 2 -- 4 (in particular 2 and 3) we are able
to produce a representation that is orders of magnitude smaller than the \ktree-representation (note
that $|[\FPequiv]|$ is very low for these graphs). For these two graphs in particular, this is due
to the majority of their nodes being laid out in a star pattern: a few hub nodes of very high degree
are connected to nodes, most of which are only connected to the hub node. Furthermore, while not
acyclic, the graphs are also very tree-like. Structures like these are compressed well by
\graphrepair, because every iteration of the replacement-round of \graphrepair roughly halves the
number of edges around the hub node. 
\begin{table}[!t]
    \centering
    \nprounddigits{0}
    \tblFontSize
    \begin{tabularx}{\textwidth}{lRRRRRR}\toprule
                    RDF-Graph& 1 & 2 & 3 & 4 & 5 & 6 \\\midrule
                    \graphrepair & \numprint{1270.5859375} & \numprint{0.9833984375} &
                    \numprint{2.7880859375} & \numprint{266.9580078125} & \numprint{30.4091796875} &
                    \numprint{871.99609375} \\
                    \ktree  & \numprint{2731.23828125} & \numprint{590.345703125} &
                    \numprint{938.3798828125} & \numprint{1118.9970703125} &
                    \numprint{51.8330078125} & \numprint{988.1552734375}\\\bottomrule
    \end{tabularx}
    \caption{Results on RDF graphs (Size in KB)}
    \label{tbl:compression_comparison_rdf}
\end{table}
\subsubsection{Version Graphs}\label{ss2:version_results}
We describe several experiments over version graphs. First we study how the compressor behaves given a high
number of identical copies of the same simple graph. The graph in this case is a directed circle
with four nodes and one of the two possible diagonal edges.
Figure~\ref{fig:smallSquareGraphMultiply} shows the results of this experiment for
identical copies starting from 8 going in powers of 2 up to 4096. Clearly, \graphrepair
is able to compress much better in this case (``exponential compression''), while the file size of other methods rises with
roughly the same gradient as the size of the graph. Note that both axes in this graph use a
logarithmic scale: in this case, \graphrepair produces a representation that is orders of magnitude
smaller than the other compressors.

Except for identical copies of rather simple graphs, however, we cannot expect to achieve
exponential compression on version graphs. Every version has changes and it is not easy
to decide which parts of two versions remain the same and can thus be compressed using the same
nonterminals. Even if we can guarantee that the same (i.e., isomorphic) substructures are
consistently compressed in the same way, the changes between versions might be too big to allow for
exponential compression. Our \FP-order is inspired by the Weisfeiler-Lehman
method~\cite{WeisfeilerLehman68} (see also~\cite{DBLP:journals/combinatorica/CaiFI92}),
which approximates a test for isomorphism. The results on version graphs, when comparing different
orders (see also Section~\ref{ss2:nodeOrder} above), suggest that this is indeed exploited.
%Nor do we want this: if it were a
%true test of isomorphism it would not order two versions in the same way, unless they are actually
%isomorphic (and thus, have not changed at all).
Figure~\ref{fig:dblp60s_orderComparison} shows a
comparison on the compression of a version graph from the DBLP co-authorship network. We
started with a co-authorship network including publications from 1960 and older. To this graph we
then add versions with the publications from 1961, 1962,\dots until 1970 and compress the graphs
obtained in this way. The comparison shows that using the \FP-order our method achieves better
compression than using other orders.
%This indicates, that the \FP-ordering does tend to label the same structures in the
%same way.
Note that the results for BFS or random order are much closer to \ktree{s}. Our full
results for version graphs are given in Table~\ref{tbl:version_results}. Note that we compare
Tic-Tac-Toe and Chess only against \ktree, because these graphs have edge labels. The results show
that \graphrepair compresses version graphs well.
%Recall, that TTT and Chess are not compared
%against LM/HN, because these graphs have edge labels.
%Tic-Tac-Toe compresses particularly well.
%This is, because we ignore the node labels that are given in the file. Without them, this is an
%instance of a number of identical copies of the same graph.
\begin{table}
    \centering
    \nprounddigits{2}
    \tblFontSize
    \begin{tabularx}{\textwidth}{lRRRR}\toprule
        & TTT & Chess & DBLP60-70 & DBLP60-90 \\\midrule
        \graphrepair & \numprint{0.12} & \numprint{9.06} & \numprint{9.54} &
        \numprint{13.39} \\
        \ktree & \numprint{9.61581469} & \numprint{13.1029467} & \numprint{15.782404} & \numprint{20.79781586} \\
        LM & - & - & \numprint{16.4449888077} & \numprint{19.3223784495} \\
        HN & - & - & \numprint{16.6491} & \numprint{18.26} \\\bottomrule
    \end{tabularx}
    \caption{Results on version graphs (compression in bpe)}
    \label{tbl:version_results}
\end{table}
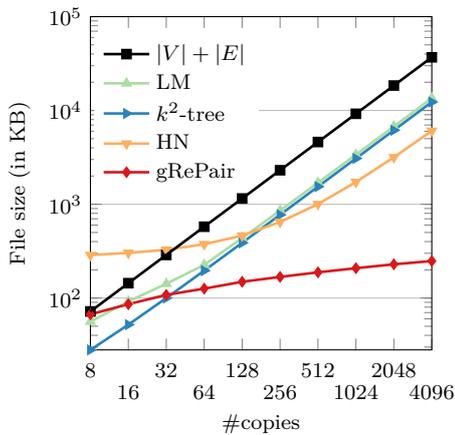
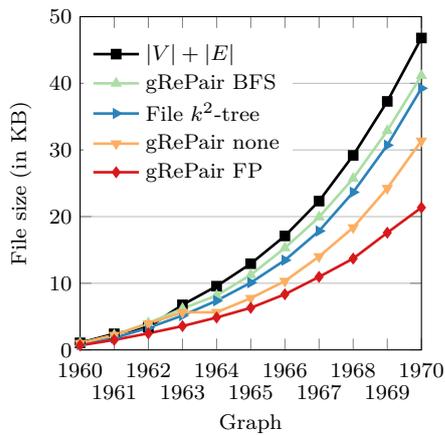
\begin{figure}[!t]
    \centering
    \subfloat[Compression of disjoint unions of the same synthetic graph with 4 nodes and 5
    edges.]{\input{figures/results_expComp}\label{fig:smallSquareGraphMultiply}}\qquad
    \subfloat[Using different node orders for the compression of a version graph.
        %(yearly snapshots
    %of the DBLP co-authorship network from 1960 to
    %1970).
    ]{\input{figures/results_60sProgression}\label{fig:dblp60s_orderComparison}}
    \caption{Results on version graphs}
\end{figure}
%\begin{figure}[!t]
%    \centering
%    \input{figures/results_expComp}
%    \caption{Compression of disjoint unions of the same synthetic graph with 4 nodes and 5
%    edges.}
%    \label{fig:smallSquareGraphMultiply}
%\end{figure}
%\begin{figure}[!t]
%    \centering
%    \input{figures/results_60sProgression}
%    \caption{Using different node orders for the compression of a version
%    graph (yearly snapshots of the DBLP co-authorship network from 1960 to 1970).}
%    \label{fig:dblp60s_orderComparison}
%\end{figure}
\subsection{Results on Synthetic Graphs}
We finally describe experiments on synthetic graphs. These show some of the effects described
earlier, namely that
\begin{enumerate}
    \item some graphs can be compressed exponentially,
    \item the maximal rank can have a big impact, and
    \item different node orders can have a big influence on the compression.
\end{enumerate}
We evaluate two different families of graphs: ``grid'' and ``triangle fractal''. Both have a
parameter $n \in \N$ to achieve different sizes. Let $\grid_n$ be an $n\times 2^n$-grid graph, i.e.,
a graph with $n\cdot 2^n$ nodes where node $i$ has edges to $i+1$ (unless $i \equiv 0 \mod{2^n}$),
and to $i+2^n$ (unless $i+2^n > n\cdot 2^n$). Following this construction, $\grid_n$ has
$n\cdot2^{n}$ nodes, and $(n-1)\cdot 2^n + n\cdot 2^{n-1}$ edges, yielding a size of $|\grid_n| =
n(2^n + 2^{n-1})+(n-1)\cdot 2^n$.
%For any $\grid_{n}$, it can be easily seen that it is
%possible to find an SL-HR grammar $G$, such that $\val(G) = \grid_{n}$, and $|G| = O(n)$. To do so,
%first replace every ``column'' of edges in the grid by a nonterminal of rank $2n$. Thus, we achieve
%a line of nonterminal edges, where every neighboring pair of nonterminal edges has $n$ nodes in
%common, which in turn are not adjacent to any node outside of those $n$. This line can therefore be
%compressed exponentially by sharing the two nonterminal edges, the $n$ nodes between them, and the
%(terminal) edges between the $n$ nodes. However, \graphrepair may not be able to find this grammar.
Furthermore we define a triangle fractal in the following way: initially $\tf_1$ is a complete graph
with $3$ nodes. To define $\tf_i$ with $i>1$, we start with $\tf_{i-1}$ and let $X$ be the set of
edges in $\tf_{i-1}$ that are incident to a node of degree $2$. For each edge $e$ in $X$ we add a new
node $v$ and add edges from both ends of $e$ to $v$, creating another triangle. Intuitively, we add
another triangle at every outer edge. Regarding the size, $t_n$ has $|t_n|_V = 2^n+2^{n-1}$ nodes,
and $|t_n|_E = 2|t_n|_V - 3$ edges. The natural order (i.e., the order of the node-IDs) is a
top-to-bottom, left-to-right order for $\grid_n$, while the natural order for $\tf_n$ is to start
with the three nodes of the innermost triangle and then go outward from there in the same way as
generated (i.e., in ``layers'').
\begin{figure}[!t]
    \centering
    \input{figures/triangleFractal}
    \includegraphics[width=0.35\textwidth]{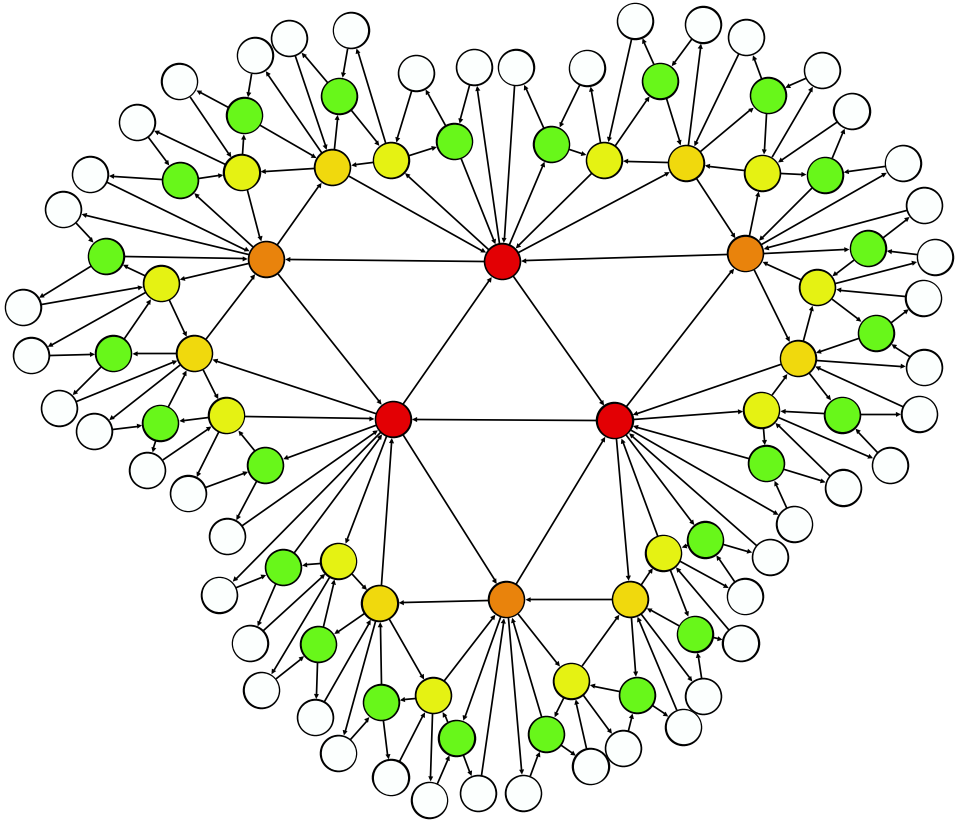}
    \caption{The $\tf_1$, $\tf_2$, $\tf_3$, and $\tf_6$ triangle fractals. The node colors in
    indicate the layer of the node: red nodes are $\tf_1$, dark orange is added in $\tf_2$,
    and every step towards green means the nodes are added in the next iteration, with white being
    the nodes added in $\tf_6$ from $\tf_5$.}
    \label{fig:triangleFractal}
\end{figure}

Table~\ref{tab:synthetic_results} lists the results for these graphs using various orders, graph
sizes, and values of maxRank. Here we state the compression performance in relation to the size
definition, i.e., compression ratio is given by $\frac{|G|}{|g|}$ (in percent), where $G$ is the
compressed representation of $g$.
\begin{table}
    \centering
    \tblFontSize
    \begin{tabularx}{\textwidth}{llRRRcRRR}\toprule
                 & & \multicolumn{3}{c}{Triangle fractal} &\phantom{ab}& \multicolumn{3}{c}{Grid} \\
              \cmidrule{3-5}\cmidrule{7-9}
        Order & MaxRank & $n = 4$    & $n = 8$   & $n = 12$  && $n = 4$    & $n = 8$   & $n = 12$
        \\\midrule
%        Order: FP \\
%\multirow{4}{*}{FP}    &
FP    &
                2        &  36.23\% &  4.61\% &  0.44\% &&  98.26\% & 99.95\% & 100.00\% \\
&                4        &  46.38\% &  5.40\% &  0.50\% && 100.00\% & 95.95\% &  97.15\% \\
&                15       &  85.51\% & 23.59\% &  9.43\% && 100.00\% & 72.31\% &  46.38\% \\
&                $\infty$ &  85.51\% & 24.80\% &  6.23\% && 100.00\% & 70.00\% &  33.99\% \\
%        Order: FP0 \\
%\multirow{4}{*}{FP0}   &
FP0   &
                2        &  36.23\% &  4.61\% &  0.44\% &&  98.26\% & 99.95\% & 100.00\% \\
&                4        &  46.38\% & 16.54\% &  5.70\% && 100.00\% & 92.23\% &  93.37\% \\
&                15       &  46.38\% & 61.10\% &  8.99\% && 100.00\% & 99.93\% &  96.99\% \\
&                $\infty$ &  46.38\% & 61.10\% &  8.99\% && 100.00\% & 99.93\% &  96.99\% \\
%        Order: Nat \\
%\multirow{4}{*}{Nat}   & 
Nat   & 
                2        &  39.13\% &  4.79\% &  0.45\% &&  98.26\% & 99.95\% & 100.00\% \\
&                4        & 100.00\% & 82.77\% & 80.69\% &&  99.42\% & 95.88\% &  97.15\% \\
&                15       &  95.65\% & 28.72\% &  5.19\% &&  94.77\% & 62.96\% &  75.48\% \\
&                $\infty$ &  95.65\% & 28.72\% &  4.51\% &&  94.77\% & 13.13\% &   1.23\% \\
%        Order: BFS \\
%\multirow{4}{*}{BFS}   & 
BFS   & 
                2        &  60.87\% & 18.28\% &  7.49\% &&  98.26\% & 99.95\% & 100.00\% \\
&                4        &  81.16\% & 54.22\% & 56.68\% && 100.00\% & 95.88\% &  97.15\% \\
&                15       &  81.16\% & 57.18\% & 72.33\% && 100.00\% & 47.81\% &  37.89\% \\
&                $\infty$ &  81.16\% & 52.13\% & 71.11\% && 100.00\% & 63.74\% &  (19.85\%) \\\bottomrule
    \end{tabularx}
    \caption{Results on triangle fractal- and grid-graphs for various values of $n$, four orders,
    and 4 different values of maxRank.}
    \label{tab:synthetic_results}
\end{table}
Some notable results: triangle fractal compresses best at $\text{maxRank}=2$. This is not
surprising, as in this case the construction of the graph is exactly reversed (with respect to the
recursive definition of $\tf_n$). However, note that at $\text{maxRank}=4$ the results are still
very close using the FP-order, whereas they get noticeably worse for every other order. Indeed, with
the BFS-order even for $\text{maxRank}=2$ the optimal result is not found any more. Using a higher
rank than 4 is generally detrimental for this graph, because digrams of higher rank occur more often
than the ones of rank 2. This is not a problem in itself, as carefully chosen occurrences can still
be reduced well, but there are many more sets of occurrences for rank-3 digrams, than there are for
rank-2 digrams. Which one is found by \graphrepair depends on the node order. This graph family
shows that
\begin{enumerate}
    \item maxRank can have a huge impact (compare 2 vs. 4 for the natural order), and
    \item the node order can have a huge impact (compare FP vs. BFS order).
\end{enumerate}

Grids are the opposite, when it comes to maxRank:
best compression is achieved for unbounded rank, while $\text{maxRank}=2$ gives almost no
compression (independent of the node order). For unbounded rank, $\FP0$ gives close to no
compression, \FP compresses to $34\%$, and natural order compresses to $12\%$ (for $n=12$). The
latter ``exponential compression'' is achieved,
%
%there is a grammar that compresses grids
%exponentially, but it requires a high rank. From the results we can indeed see that, while using the
%\FP-order yields good compression at unbounded rank, only the natural order achieves significant
%compression with unbounded maxRank. The grammar compressing grids exponentially, essentially does so
%
by treating the grid as $n$ identical lines of length $2^n$. Node orders that traverse the grid in
such a way (like the natural order, which goes row by row, or the BFS-order which follows the
rows ``in parallel''), are likely to find digram occurrences corresponding to this structure.
Accordingly, the BFS-order also achieves good compression, in particular, it has the best result 
for $\text{maxRank}=15$ with $n=8$ or $n=12$. Note however, that for $\grid_{12}$ with BFS-order and
unbounded maxRank the rank was actually bounded to $1500$, because the computation takes too long
otherwise.

%% file: figures/results_nodeorders.tex
\pgfplotstableread[col sep=comma]{figures/data/results_nodeorder.csv}\nodeOrderResults
\begin{tikzpicture}

\graphFontSize
\begin{axis}[
    ybar,
    ymin=0,
    ymax=18,
    ylabel=bpe,
    ylabel style={yshift={-1.5em}},
    xticklabels from table={figures/data/results_nodeorder.csv}{Graph},
    legend entries={FP,FP0,BFS,Nat},
    legend style={at={(0.5,1.01)}, draw=none, anchor=south,/tikz/every even column/.append style={column sep=0.5em},
    /tikz/every odd column/.append style={column sep=0.1em}},
    %ytick distance=2,
    ytick={0,2,4,6,8,10,12,14,16,18},
    ymajorgrids=true,
    xticklabel style={text height=1.5ex,yshift={-mod(\ticknum,2)*1em}},
    xtick=data,
    x tick label style={rotate=0,anchor=north},
    width=1.0\textwidth,
    height=60mm,
    bar width=7pt,
    legend cell align=left,
    legend columns=4
]
\addplot [fill=color1] table[x expr=\coordindex,y=FP]{\nodeOrderResults};
\addplot [fill=color2] table[x expr=\coordindex,y=FP0]{\nodeOrderResults};
\addplot [fill=color3,postaction={pattern=north east lines}] table[x
expr=\coordindex,y=BFS]{\nodeOrderResults};
\addplot [fill=color4] table[x expr=\coordindex,y=Nat]{\nodeOrderResults};
\end{axis}
\end{tikzpicture}

%% file: figures/results_eqclasses.tex
\pgfplotstableread[col sep=comma]{figures/data/results_bpeEquiv.csv}\resultsBpeEquiv
\begin{tikzpicture}[spy using outlines={circle, magnification=4, connect spies}]
\graphFontSize
\begin{axis}[
    ymin=0,
    ymax=1.0,
    xmin=0,
    xmax=16,
%    xtick distance=2,
%    ytick distance=0.2,
    xtick={0,2,4,6,8,10,12,14,16},
    ytick={0.2,0.4,0.6,0.8,1},
    ylabel=Equivalence Classes/Node,
    ylabel style={yshift={-1em}},
    xlabel=bpe,
%    xticklabels from table={figures/data/results_expComp.csv}{Copies},
    legend entries={
        Chess,%
        Tic-Tac-Toe,%
        Email-EuAll,%
        NotreDame,%
        Wiki-Talk,%
        Wiki-Vote,%
        CA-AstroPh,%
        CA-CondMat,%
        CA-GrQc,%
        Email-Enron,%
        Specific Properties en,%
        Types ru,%
        Types es,%
        Types de with en,%
        Identica,%
        Jamendo},
    legend style={at={(1.05,0.47)}, draw=none, anchor=west},
%
%    ytick distance=10,
    ymajorgrids=true,
%    xticklabel style={text height=1.5ex,yshift={-2pt}},
%    xtick=data,
%    x tick label style={rotate=0,anchor=north},
    width=0.6\textwidth,
    height=70mm,
    legend cell align=left,
    legend columns=1,
    enlargelimits=false,
    mark options={mark size=2.5pt},
    axis background/.style={postaction={path picture={\draw[color=black!15,fill=black!15] (axis
    cs:6.5,0) -- (axis cs:16,0.65) -- (axis cs:16,0) -- (axis cs:6.5,0);}}}
]

\addplot[
    scatter/classes={
        Chess={mark=square*,color=color3},%
        TTT={mark=diamond*,color=color2},%
        EuAll={mark=star,color=color4,very thick},%
        NotreDame={mark=triangle*,color=color4},%
        WikiTalk={mark=halfcircle*,color=color1},%
        WikiVote={mark=pentagon*,color=color2},%
        Astro={mark=halfsquare right*,color=color1},%
        Cond={mark=halfsquare left*,color=color4},%
        GrQc={mark=*,color=color3},%
        Enron={mark=+,color=color2,very thick},%
        SpecEn={mark=x,color=color4,very thick},%
        TypesRu={mark=+,color=color4,very thick},%
        TypesEs={mark=asterisk,color=color1},%
        TypesDeEn={mark=*,color=color1},%
        Identica={mark=|,color=color4,very thick},%
        Jamendo={mark=square*,color=color4}
    },
    scatter,
    scatter src=explicit symbolic,
    only marks]
    table[x=y, y=x, meta=Graph]{\resultsBpeEquiv};

    \coordinate (spypoint) at (axis cs:0.05,0.005);
    \coordinate (magnify) at (axis cs:2,0.3);

    \spy[black,size=30pt] on (spypoint) in node at (magnify);
%
%    \draw[color=black!15,fill=black!15] (axis cs:6.5,0.005) -- (axis cs:16,0.65) -- (axis cs:16,0.005) --
%    (axis cs:6.5,0.005);
\end{axis}
\end{tikzpicture}

%% file: figures/results_networks.tex
\pgfplotstableread[col sep=comma]{figures/data/results_networks.csv}\networkResults
\begin{tikzpicture}

\graphFontSize
\begin{axis}[
    ybar,
    ymin=0,
    ymax=22,
    ylabel=bpe,
    ylabel style={yshift={-1.5em}},
    xticklabels from table={figures/data/results_networks.csv}{Graph},
    legend entries={\graphrepair,\graphrepair{}+DSR,\ktree,HN,LM},
    legend style={at={(0.5,1.01)}, draw=none, anchor=south,/tikz/every even column/.append
    style={column sep=0.5em}, /tikz/every odd column/.append style={column sep=0.1em}},
%    ytick distance=2,
    ytick={0,2,4,6,8,10,12,14,16,18,20,22},
    ymajorgrids=true,
    xticklabel style={text height=1.5ex,yshift={-mod(\ticknum,2)*1em}},
    xtick=data,
    x tick label style={rotate=0,anchor=north},
    width=1.0\textwidth,
    height=60mm,
    bar width=5pt,
    legend cell align=left,
    legend columns=5
]
\addplot [fill=color1] table[x expr=\coordindex,y=gRePair]{\networkResults};
\addplot [fill=color2] table[x expr=\coordindex,y=gRePairDSR]{\networkResults};
\addplot [fill=color3,postaction={pattern=north east lines}] table[x
expr=\coordindex,y=kTree]{\networkResults};
\addplot [fill=color4] table[x expr=\coordindex,y=HN]{\networkResults};
\addplot [fill=color2,postaction={pattern=north east lines}] table[x expr=\coordindex,y=LM]{\networkResults};
\end{axis}
\end{tikzpicture}

%% file: figures/results_expComp.tex
\pgfplotstableread[col sep=comma]{figures/data/results_expComp.csv}\resultsExpComp
\begin{tikzpicture}

\graphFontSize
\begin{semilogyaxis}[
    set layers=standard,
    ymin=0,
    ymax=100000,
    ylabel=File size (in KB),
    ylabel style={yshift={-1em}},
    xlabel=\#copies,
    xlabel style={yshift={-.75em}},
    xticklabels from table={figures/data/results_expComp.csv}{Copies},
    legend entries={$|V|+|E|$, LM, \ktree, HN, \graphrepair},
    legend style={at={(0.01,0.95)}, draw=none, anchor=north west},
    ymajorgrids=true,
    xticklabel style={text height=1.5ex,yshift={-2pt}},
    xtick=data,
    x tick label style={rotate=0,anchor=north,yshift={-mod(\ticknum,2)*1em}},
    width=0.45\textwidth,
    height=60mm,
    legend cell align=left,
    legend columns=1,
    enlargelimits=false,
    mark options={mark size=1.5pt}
]
\addplot [line width=1pt,color=black, mark=square*, on layer=axis foreground] table[x expr=\coordindex,y=VE]{\resultsExpComp};
\addplot [line width=1pt,color=color3, mark=triangle*] table[x
expr=\coordindex,y=LM]{\resultsExpComp};
\addplot [line width=1pt,color=color4, mark=triangle*, every mark/.append style={rotate=270}]
table[x expr=\coordindex,y=kTree]{\resultsExpComp};
\addplot [line width=1pt,color=color2, mark=triangle*, every mark/.append style={rotate=180}]
table[x expr=\coordindex,y=HN]{\resultsExpComp};
\addplot [line width=1pt,color=color1, mark=diamond*] table[x
expr=\coordindex,y=gRePair]{\resultsExpComp};
\end{semilogyaxis}
\end{tikzpicture}

%% file: figures/results_60sProgression.tex
\pgfplotstableread[col sep=comma]{figures/data/60sprogression.csv}\resultsProgression
\begin{tikzpicture}

\graphFontSize
\begin{axis}[
    ymin=0,
    ymax=50,
    ylabel=File size (in KB),
    ylabel style={yshift={-1.5em}},
    xlabel=Graph,
    xlabel style={yshift={-.75em}},
    xticklabels from table={figures/data/60sprogression.csv}{Graph},
    legend entries={$|V|+|E|$, \graphrepair BFS, File \ktree, \graphrepair none,\graphrepair FP},
    legend style={at={(0.01,0.95)}, draw=none, anchor=north west},
%    ytick distance=10,
    ytick={0,10,20,30,40,50},
    ymajorgrids=true,
    xticklabel style={text height=1.5ex,yshift={-2pt}},
    xtick=data,
    x tick label style={rotate=0,anchor=north,yshift={-mod(\ticknum,2)*1em}},
    width=0.45\textwidth,
    height=60mm,
    legend cell align=left,
    legend columns=1,
    enlargelimits=false,
    mark options={mark size=1.5pt}
]
\addplot [line width=1pt,color=black, mark=square*] table[x expr=\coordindex,y=VE]{\resultsProgression};
\addplot [line width=1pt,color=color3, mark=triangle*] table[x expr=\coordindex,y=gRePairBFS]{\resultsProgression};
\addplot [line width=1pt,color=color4, mark=triangle*, every mark/.append style={rotate=270}] table[x expr=\coordindex,y=kTree]{\resultsProgression};
\addplot [line width=1pt,color=color2, mark=triangle*, every mark/.append style={rotate=180}] table[x expr=\coordindex,y=gRePairNone]{\resultsProgression};
\addplot [line width=1pt,color=color1, mark=diamond*] table[x expr=\coordindex,y=gRePairFP]{\resultsProgression};
\end{axis}
\end{tikzpicture}

%% file: figures/triangleFractal.tex
\scalebox{.8}{
    \begin{tikzpicture}
        \begin{scope}[every node/.style={circle, draw, inner sep=0pt, outer sep=0pt, minimum size=7pt,
            node distance=20pt}]
            \node[fill=l1] (11) {};
            \node[fill=l1,below left=1 and 0.5 of 11] (12) {};
            \node[fill=l1,below right=1 and 0.5 of 11] (13) {};

            \node[fill=l1,above right=0.5 and 3 of 11] (21) {};
            \node[fill=l1,below left=1 and 0.5 of 21] (22) {};
            \node[fill=l1,below right=1 and 0.5 of 21] (23) {};
            \node[fill=l2,above left=1 and 0.5 of 22] (24) {};
            \node[fill=l2,above right=1 and 0.5 of 23] (25) {};
            \node[fill=l2,below left=1 and 0.5 of 23] (26) {};

            \node[fill=l1,right=4 of 21] (31) {};
            \node[fill=l1,below left=1 and 0.5 of 31] (32) {};
            \node[fill=l1,below right=1 and 0.5 of 31] (33) {};
            \node[fill=l2,above left=1 and 0.5 of 32] (34) {};
            \node[fill=l2,above right=1 and 0.5 of 33] (35) {};
            \node[fill=l2,below left=1 and 0.5 of 33] (36) {};
            \node[fill=l3,above right=1 and 0.5 of 34] (37) {};
            \node[fill=l3,above right=1 and 0.5 of 31] (38) {};
            \node[fill=l3,below right=1 and 0.5 of 35] (39) {};
            \node[fill=l3,below right=1 and 0.5 of 33] (310) {};
            \node[fill=l3,below left=1 and 0.5 of 32] (311) {};
            \node[fill=l3,below left=1 and 0.5 of 34] (312) {};
        \end{scope}

        \begin{scope}[every path/.style={-stealth'}]
            \draw (11) -- (12);
            \draw (12) -- (13);
            \draw (13) -- (11);

            \draw (21) -- (22);
            \draw (22) -- (23);
            \draw (23) -- (21);
            \draw (22) -- (24);
            \draw (24) -- (21);
            \draw (21) -- (25);
            \draw (25) -- (23);
            \draw (23) -- (26);
            \draw (26) -- (22);

            \draw (31) -- (32);
            \draw (32) -- (33);
            \draw (33) -- (31);
            \draw (32) -- (34);
            \draw (34) -- (31);
            \draw (31) -- (35);
            \draw (35) -- (33);
            \draw (33) -- (36);
            \draw (36) -- (32);
            \draw (31) -- (37);
            \draw (37) -- (34);
            \draw (35) -- (38);
            \draw (38) -- (31);
            \draw (39) -- (35);
            \draw (33) -- (39);
            \draw (310) -- (33);
            \draw (36) -- (310);
            \draw (311) -- (36);
            \draw (32) -- (311);
            \draw (312) -- (32);
            \draw (34) -- (312);
        \end{scope}
    \end{tikzpicture}
}

%% file: query_evaluation.tex
\section{Query Evaluation}

In this section we investigate two types of queries
that can be performed over SL-HR grammars: 
neighborhood queries and speed-up queries.
\emph{Neighborhood queries} allow to traverse the edges of a graph (in any direction).
Using them, any arbitrary graph algorithm can be performed on the compressed
representation given by an SL-HR grammar.
However, this comes at a price: a considerable slow-down is to be expected
in comparison to running over an uncompressed graph representation, because a partial decompression
is required in order to obtain the neighboring nodes.
In contrast, \emph{speed-up queries}, as their name suggests, can run
faster on an SL-HR grammar than on an uncompressed graph representation.
Examples of speed-up queries are counting the number of connected components
of the graph, checking regular path properties in the graph, or checking
reachability between two nodes.
These queries can be evaluated in polynomial time over the grammar, 
and allow speed-ups proportional to the compression ratio.
The results in this section have not been implemented.
Over grammar-compressed trees, the performance of simple speed-up
queries is evaluated in~\cite{DBLP:journals/corr/abs-1012-5696}.

\subsection{Neighborhood Queries}\label{sss:neighborhood}
For a node $v \in V_g$ of a hypergraph $g$
we denote by $\neighbor(v) = \{u \in V_g \mid \exists e \in E_g: u, v \in \att_g(e)\}$ the
\emph{neighborhood of $v$}. For simple graphs we also define $\inNeighbor(v) = \{u \in V_g \mid
\exists e \in E_g: \att(e) = uv\}$ and $\outNeighbor(v) = \{u \in V_g \mid \exists e \in E_g:
\att(e) = vu\}$, the \emph{incoming} and \emph{outgoing neighborhoods of $v$}, respectively.
Furthermore we let $\incident(v) = \{e \in E_g \mid v \in \att(e)\}$ be the set of edges incident
with $v$.

Let $G=(N,P,S)$ be an SL-HR grammar.
We assume that every right-hand side in $P$ contains at
most two nonterminal edges (but arbitrarily many terminal edges). This can be achieved by replacing
pairs of nonterminal edges with a single new one: for $n>2$ let $e_1,\ldots,e_n$ be the nonterminal
edges in $\rhs(A)$ for some $A \in N$. Then replace $e_2,\ldots,e_n$ by a new edge $e_B$ with
$\att(e_B) = \att(e_2)\cdots\att(e_n)$, $\lbl(e_B) = B \notin N$ and $\rhs(e_B)$ is a graph
containing $n-1$ nonterminal edges such that replacing $e_B$ generates the original $\rhs(A)$. This
procedure can recursively be applied until every right-hand side has at most two nonterminal edges.
This increases the size of the grammar by at most a factor of 2.

Recall that the nodes in the start graph $S$ are numbered $1,..,m$
and that there is an order on the nonterminal edges $e_1,\dots,e_\ell$ in $S$
so that the nodes in $g_1=\val(e_1)$ are numbered $m+1,\dots,m+v_1$, where
$v_1=|g_1|_V$, and similarly, nodes in $g_i=\val(e_i)$ are numbered 
from $k=m+\sum_{j=1}^{i-1}v_j$ to $k+v_i$.
Given a node ID, i.e., a number in $k\in\{1,\dots,|\val(G)|_V\}$, computing
its outgoing neighbors consists of two steps:
\begin{enumerate}
    \item compute a \emph{grammar representation} (\emph{$G$-representation})
        of $k$, and
    \item from a $G$-representation, compute the outgoing neighbors (as IDs in the decompressed
        graph) of the represented node.
\end{enumerate}
A $G$-representation is a path in the derivation tree
of $G$ that ``derives the node $k$''. Such a path is of the form
$wv$ where $w$ is a (possibly empty) string of the form
$e_0e_1\cdots e_n$. If $w$ is empty, then $v$ must be a node in $S$. 
If not, then $e_0$ is a nonterminal edge of $S$.
If $A_1$ is the label of $e_0$, then $e_1$ is a nonterminal edge
in $\rhs(A_1)$ labeled $A_2$, etc. Finally, $v$ is an internal node
in $\rhs(A_n)$.
Let $\ell$ be the number of nonterminal edges in $S$
and $h = \height(G)$.
The $G$-representation of $k$ can be computed using Algorithm~\ref{alg:gRep}, where
$\#\mathsf{nodes}(A)$ refers to the number of internal nodes in $\val(A)$. The complexity of this
algorithm depends on the specific implementation of Steps~\ref{step:Ssearch} and~\ref{step:firstID},
and the computation of $\#\mathsf{nodes}(A)$. Specifically, it is possible to precompute
$\#\mathsf{nodes}(A)$ for every $A \in N$ and the first ID of every edge in $S$
(Step~\ref{step:firstID}). Doing so needs an additional $O(|N| + \ell)$ space, but now
Step~\ref{step:Ssearch} can be implemented using a binary search and only needs $O(\log\ell)$
time. Ignoring the preprocessing, a computation of Algorithm~\ref{alg:gRep} takes $O(\log(\ell)+h)$
time and $O(|N|+|E_S|)$ space. If we do not want to use the additional space, a computation of
$\#\mathsf{nodes}(A)$ takes $O(h)$ time (it can be done bottom-up). Steps~\ref{step:Ssearch}
and~\ref{step:firstID} can now be implemented using a linear search through the nonterminal edges of
$S$ in their derivation order as defined in Section~\ref{sss:stringTreeDef}. This leads to a total
running time of $O(|E_S|h+h^2)$. In the following we refer to the runtime of
Algorithm~\ref{alg:gRep} by $\gRepRun$ to reflect this ambiguity.
\begin{algorithm}[!t]
    \begin{algorithmic}[1]
%        \Require{Graph $g = (V,E,\att,\lbl,\ext)$}
%        \Ensure{Grammar $G$ with $\val(G) \cong g$ and $|G'| \leq |g|$}
        \Function{getGRep}{$v$: node-ID (of $\val(G)$)}
            \If{$v < |V_S|$}
                \State \textbf{return} $v$
            \Else
                \State $e \gets $ nonterminal edge within $S$ producing $v$ \label{step:Ssearch}
                \State $i \gets $ first ID assigned to nodes created from deriving $e$
                \label{step:firstID}
                \State \textbf{return} getGRep$(v,e,i)$
            \EndIf
        \EndFunction
        \Function{GetGRep}{$v$: node-ID, $e$: edge, $i$: Integer}
            \State $\mathit{int} \gets $ number of internal nodes in $\rhs(\lbl(e))$
            \If{$i + \mathit{int} \geq v$}
                \State $n \gets v-i$ 
                \State \textbf{return} $n$-th internal node of $\rhs(\lbl(e))$
            \Else
                \State $e_1,e_2 \gets $ nonterminal edges of $\rhs(\lbl(e))$ (if exists)
                \If{$i + \mathit{int} + \#\mathit{nodes}(\lbl(e_1)) \geq v$}
                    \State \textbf{return} $e_1 \cdot$ getGRep$(v, e_1, id+\mathit{int})$
                \Else
                    \State \textbf{return} $e_2 \cdot$ getGRep$(v, e_2, id+\mathit{int} +
                    \#\mathit{nodes}(\lbl(e_1)))$
                \EndIf
            \EndIf
        \EndFunction
    \end{algorithmic}
    \caption{Routine to compute a $G$-representation for a given node-ID $v$.}
    \label{alg:gRep}
\end{algorithm}

Given the $G$-representation $e_0e_1\cdots e_n v$, the outgoing neighbors are computed using
Algorithm~\ref{alg:getNeigh}. This algorithm uses the function getID, detailed in
Algorithm~\ref{alg:getID}. The complexity of the latter again depends on the implementation of
Step~\ref{step:comFirst} and $\#\mathsf{nodes}(A)$. By the same reasoning as above we get a runtime
of $\gRepRun$ for getID. Finally getOutNeighborhood uses $O(\gRepRun\cdot n)$ time, where $n$ is the
out-degree of the node represented by $e_0e_1\cdots e_n v$ within $\val(G)$. Note that for the
algorithms to be well defined, they have to work with $e_1\cdots e_n = \varepsilon$. In this case,
we implicitly assume $e_n = S$ and $\rhs(\lbl(e_n)) = S$, i.e., we are working on the startgraph.
%
%as follows.
%We return every internal node $u$ in $\rhs(A_n)$ that has a terminal edge to $v$.
%For every external node $u$ in $\rhs(A_n)$ that has a terminal edge to $v$,
%we need to compute its node ID, which is done by calling the method
%$\text{getID}(e_0e_1\cdots e_n u)$. Clearly, this is done in $O(h)$ time.
%For every nonterminal edge $e$ in $\rhs(A_n)$ that has an attachment to $v$,
%i.e., $v$ appears in $\att(e)$ at some position $p$,
%we compute the node IDs of all internal nodes produced by $A$ (the label of $e$)
%and having an edge to the $p$th external node. We use a recursive method
%$\text{getNeighboring}(e,p)$ which returns the set of node IDs that are neighbors to the $p$th
%external node within the subgraph derived from $e$. To do so, iterate over the edges incident with the
%$p$th external node. Let $e'$ be the current edge. If $e'$ is terminal and attached to an internal
%node, add the internal nodes ID to the result set. If it is terminal and attached to another
%external node $w$, we add the ID obtained by $\text{getID}(e_0e_1\cdots e_n e w)$ to the result set.
%If the edge is nonterminal, let $p'$ be the position of the $p$th external node in $\att(e')$ and
%add the result of $\text{getNeighboring}(e', p')$ to the result set. The runtime is $O(h)$ per node,
%and thus a total of $O(nh)$, where $n$ is the number of neighbors.
\begin{algorithm}[!t]
    \begin{algorithmic}[1]
        \Function{getID}{$e_1\cdots e_n v$: $G$-representation}
            \If{$v$ is external}
                \State $p \gets \varepsilon$
                \State $u \gets v$
                \For{$i \gets n, \ldots, 2$}
                    \State $j \gets $ position of $u$ in $\ext_{\rhs(\lbl(e_i))}$
                    \If{$\att(e_{i-1}[j])$ is internal}
                        \State $v' \gets \att(e_{i-1})[j]$
                        \State $n' \gets i-1$
                        \State $p \gets e_1\cdots e_{n'} v'$
                        \State \textbf{break} for loop
                    \EndIf
                \EndFor
            \EndIf
            \State $\mathit{id} \gets$ compute first ID created when deriving
            $e_1$\label{step:comFirst}
            \For{$i \gets 1,\ldots, n'-1$}
                \State $\mathit{id} \gets \mathit{id} + $ number of internal nodes in $\rhs(e_i)$
                \State $e' \gets $ nonterminal edge in $\rhs(\lbl(e_i))$ not equal to $e_i$ (if
                exists)
                \If{$e' < e_{i+1}$ (in derivation order)}
                    \State $\mathit{id} \gets \mathit{id}+\#\mathsf{nodes}(\lbl(e'))$
                \EndIf
            \EndFor
            \State $k \gets $ position of $v'$ within the internal nodes of $\rhs(\lbl(e_{n'}))$
            \State \textbf{return} $\mathit{id} + k$
        \EndFunction
    \end{algorithmic}
    \caption{Computes the ID of a given $G$-representation within $\val(G)$.}
    \label{alg:getID}
\end{algorithm}

\begin{algorithm}[!t]
    \begin{algorithmic}[1]
        \Function{getOutNeighborhood}{$e_1\cdots e_n v$: $G$-representation}
            \State $n \gets \emptyset$
            \If{$v$ is external}
                \State $n \gets $getNeighborhood(getGRep(getID($e_1\cdots e_n v$)))
            \Else
                \State $D \gets \{u \mid \exists e\in E_{\rhs(\lbl(e_n))}$ with $\att(e) = vu$ and
                        $e$ is terminal$\}$
                \State $n \gets n \cup \{$getID$(e_1\cdots e_n u) \mid u \in D\}$
                \ForAll{nonterminal edge $e$ attached to $v$ in $\rhs(\lbl(e_n))$}
                    \State $i \gets $ position of $v$ in $\att(e)$
                    \State $n \gets n \cup \text{getNeighbors}(\lbl(e),i)$
                \EndFor
            \EndIf
            \State \textbf{return} $n$
        \EndFunction
        \Function{getOutNeighbors}{$A$: nonterminal, $i$: Integer $(1 \leq i \leq \rank(A))$}
            \State $n \gets \emptyset$
            \State $v \gets i$-th external node of $\rhs(A)$
            \State $D \gets \{u \mid \exists e\in E_{\rhs(\lbl(e_n))}$ with $\att(e) = vu$ and
                    $e$ is terminal$\}$
            \State $n \gets n \cup \{$getID$(e_1\cdots e_n u) \mid u \in D\}$
            \ForAll{nonterminal edge $e$ attached to $v$ in $\rhs(\lbl(e_n))$}
                \State $i \gets $ position of $v$ in $\att(e)$
                \State $n \gets n \cup \text{getNeighbors}(\lbl(e),i)$
            \EndFor
            \State \textbf{return} $n$
        \EndFunction
    \end{algorithmic}
    \caption{Function getOutNeighborhood, which given a $G$-representation $p$ returns the set of IDs
    that are out-neighbors of the node represented by $p$ within $\val(G)$.}
    \label{alg:getNeigh}
\end{algorithm}

\begin{proposition}
Let $G$ be an SL-HR grammar and $k\in\{1,\dots,|\val(G)|_V\}$.
Let $n$ be the number of in (or out) neighbors of $k$ in $\val(G)$.
The node IDs of these $n$ nodes can be computed in time
$O(\gRepRun n)$.
\end{proposition}

Note that for string grammars, data structures have been presented that guarantee
\emph{constant time} per move from one letter to the next (or previous)~\cite{DBLP:conf/dcc/GasieniecKPS05}.
%We expect that the result can be extended to our SL-HR grammars.
This result has been extended to grammar-compressed trees~\cite{lmr16}, and we hope that
it can also be generalized to SL-HR grammar-compressed graphs. 

\subsubsection{Speed-Up Queries}
One attractive feature of straight-line context-free grammars
is the ability to execute finite automata over them without prior decompression.
This was first proved for strings (see~\cite{Lohrey12_stringSLPsurvey}) and was later extended to
trees (and various models of tree automata, see~\cite{DBLP:journals/jcss/LohreyMS12}).
The idea is to run the automaton in one pass, bottom-up, through the grammar.
As an example, consider the grammar
$S\to AAA$ and $A\to ab$ from the Introduction, and an automaton $\mathcal{A}$ that
accepts strings (over $\{a,b\}$) with an odd number of $a$'s.
Thus, $\mathcal{A}$ has states $q_0,q_1$ (where $q_0$ is initial and $q_1$ is final) and the transitions
$(q_0,a,q_1)$, $(q_1,a,q_0)$, and $(q,b,q)$ for $q\in\{q_0,q_1\}$.
Since the actual active states are not known during the bottom-up run
through the grammar, we need to run the automaton \emph{in every possible state}
over a rule.
For the nonterminal $A$ we obtain $(q_0,A,q_1)$ and
$(q_1,A,q_0)$, i.e., running in state $q_0$ over the string produced by $A$
brings us to state $q_1$, and starting in $q_1$ brings us to $q_0$.
Since $S$ is the start nonterminal, we are only interested in starting
the automaton in its initial state $q_0$.
We obtain the run $(q_0,A,q_1)(q_1,A,q_0)(q_0,A,q_1)$, i.e., the automaton
arrives in its final state $q_1$ and hence the grammar represents a string
with odd number of $a$'s.
It should be clear that the running time of this process
is $O(|Q||G|)$, where $Q$ is the set of states of the automaton,
and $G$ is the grammar.

Unfortunately, for graphs there does not exists an accepted notion
of finite-state automaton.
Nevertheless, properties that can be checked in one pass through
the derivation tree of a graph grammar have been studied
under various names: ``compatible'', ``finite'', and ``inductive'',
and it was later shown that these notions are essentially equivalent~\cite{DBLP:journals/tcs/HabelKL93}.
Courcelle and Mosbah~\cite{DBLP:journals/tcs/CourcelleM93} show that
all properties definable in ``counting monadic second-order logic'' (CMSO) 
belong to this class, and by their Proposition~3.1,
the complexity of evaluating a CMSO property $\psi$ over a derivation tree $t$ can
be done in $O(|t|\eta)$, where $\eta$ is an upper bound on the complexity of
evaluation on each right-hand side of the rules in $t$. This is done by a bottom-up computation on
the nodes of $t$, where every node $u$ can use the results of its children to achieve a correct
computation for the subtree rooted at $u$.
%How can we apply this result to SL-HR grammars $G$?
%Let $G=(N,S,P)$ so that every right-hand side in $P$ contains
%at most two nonterminal edges.
%We are interested in \emph{data complexity}, i.e., we assume $\psi$ to be fixed.
%Note that the size of the derivation tree $t$ of $G$ can be exponential
%in $|G|$. It is thus prohibitive to construct $t$, and instead their
%proposition above needs to be generalized to dags (directed acyclic graphs) $d$
%that represent $t$. Next, we eliminate $\eta$ as follows.
%We convert $G$ into a new SL-HR grammar $G'$ so that $G'$ is in Chomsky
%normal form, i.e., every right-hand side (including the start graph) has
%at most two edges
%%.
%%This can be done by replacing two edges in $S$ by one
%%new nonterminal edge, together with a rule producing the two edges
%(see, e.g., Proposition~3.13 of~\cite{Engelfriet:1997:CGG:267871.267874}).
%The derivation dag of $G'$ now has size $O(|G|)$. 
%Let $m$ be the maximum of the rank of $G$ and
%of $|S|_V$. %Since i
%In the worst-case 
%the start graph of $G'$
%%the repeated replacement of above
%has one 
%%generates 
%nonterminal edge that is incident with all nodes in $S$.
%Thus the maximal size of a right-hand side of $G'$ is in $O(m)$.
%%The $O(|d|\eta)$ time can now be stated as 
%%$O(|G|\cdot m)$.
%
\begin{proposition}\label{prop:mso}
Let $\psi$ be a fixed CMSO property.
For a given SL-HR grammar $G$ it can be decided in
$O(|G|\eta)$ time whether or not $\psi$ holds on
$\val(G)$, where $\eta$ is an upper bound on the time needed to evaluate one right-hand side of the
grammar.
\end{proposition}
Note that we state the time complexity based on $|G|$ instead of $|t|$. This needs an adjustment of
the proof, as the derivation tree $t$ cannot be explicitly constructed, because it may be
exponentially larger than $G$. However, a derivation DAG (directed acyclic graph) still retains all
information necessary for the algorithm stated in~\cite{DBLP:journals/tcs/CourcelleM93} to work. It
has a hierarchy, that makes the bottom-up computation possible, and the while there is ambiguity in
the parent relationship, the children for every node are well defined. A computation for one
right-hand side depending on the results of that computation on the children in the derivation tree,
can therefore be done in the derivation DAG as well.

Note that Proposition~\ref{prop:mso} is often stated under a fixed
tree decomposition $t$ of width $k$ of the graph $g$ and then simply 
becomes $O(|g|)$. 
The CMSO (or compatible or finite) graph properties have been extended to
functions from graphs to natural numbers, see e.g., Section~5 
of~\cite{DBLP:journals/tcs/CourcelleM93}.
They can be evaluated with a similar complexity as in Proposition~\ref{prop:mso}.
For the same explanation as above, this result can be applied to SL-HR grammars.
Without stating this result explicitly, we mention some of
the well-known CMSO functions:
%\begin{itemize}
%\item 
(1)~maximal and minimal degree,
%\item 
(2)~number of connected components,
%\item 
(3)~number of simple cycles,
%\item 
(4)~number of simple paths from a source to a target, and
%\item 
(5)~maximal and minimal length of a simple cycle.
%\end{itemize}

%More general graph properties and functions are considered, e.g., 
%by Drewes~\cite{DBLP:books/daglib/0084644}.
Beware that $\eta$ in Proposition~\ref{prop:mso} need not be linear in the size of a right-hand
side, but is rather a generic upper bound. Courcelle and
Mosbah~\cite{DBLP:journals/tcs/CourcelleM93} show (for their Proposition 3.1) linearity for
evaluations using a certain Boolean algebra, and cubic complexity for sets of cardinalities. For
universal evaluation they give an exponential upper bound. Lohrey~\cite{DBLP:journals/jcss/Lohrey12}
proves that the problem in Proposition~\ref{prop:mso} becomes \PSPACE-complete for grammars where both
the rank and the number of nonterminals per right-hand side are bounded by a constant. Without these
restrictions the problem is complete for \EXPTIME. Note that \PSPACE-completeness is already true
for explicitly represented graphs.

The specific complexity thus depends on the problem. We show for two problems that they can be
solved in linear time with grammar-compressed graphs as input.
\subsection{Reachability Queries}\label{sect:reachability}

An important class of queries are \emph{reachability queries}.
For a given graph $g$ and nodes $u$ and $v$ such a query asks if 
$v$ is reachable from $u$, i.e.,
if there exists a path from $u$ to $v$ in $g$.
It is well known that this problem can be solved in $O(g)$ time (e.g., by doing a BFS-traversal in
$O(|V|+|E|)$ time).
How can we solve this problem on an SL-HR grammar $G$?
Certainly, $(u,v)$-reachability is CMSO definable and therefore
Proposition~\ref{prop:mso} gives us an upper bound of $O(|G|^2)$. 
The following direct \emph{linear time} algorithm essentially uses the same method already applied
by Lengauer and Wanke~\cite{DBLP:journals/siamcomp/LengauerW88}. Their formalism is slightly
different from ours however, as it uses an encoding of the hypergraphs using bipartite graphs, and
the algorithm as stated only decides reachability in undirected graphs. We therefore restate the
algorithm using the following notion, which will again be used in the next section:
\begin{definition}
    Let $g$ be a graph with set $V$ of external nodes. Then the \emph{skeleton graph} $\sk(g) =
    (V,E)$ is the directed, unlabeled graph such that for all $v, v' \in V$, $(v,v') \in E$ if and
    only if $v'$ is reachable from $v$ in $\val(g)$.
%
%    Let $A$ be a nonterminal with right-hand side $g$. The \emph{skeleton graph} $\sk(A)$ is given
%    by the external nodes $\{v_1,\ldots,v_n\}$ of $g$ such that there is an edge from $v_i$ to $v_j$
%    (for $i,j \in [n]$) if $v_j$ is reachable from $v_i$.
    \label{def:skeleton}
\end{definition}
The edges of a skeleton graph are computed as follows.
First, assume that $g$ is a terminal graph.
We determine the strongly connected components of $g$
in linear time (e.g., using Tarjan's algorithm~\cite{DBLP:journals/siamcomp/Tarjan72}).
Let $g'$ be the corresponding graph which has
as nodes the strongly connected components of $g$. 
We remove from $g'$ each strongly connected component $C$ that does
not contain external nodes.
This is done by inserting for every pair of edges $e_1,e_2$
such that $e_1$ is an edge from a component $D\not= C$ into $C$ 
and $e_2$ is an edge from $C$ to a component $E\not =C$ (with $E\not= D$),
an edge from component $D$ to component $E$. 
Finally, we replace each component by a cycle of the external
nodes of that component, and, for an edge from a component $D$ to
a component $E$ we add an edge from an arbitrary external node of $D$
to one of $E$.

\begin{theorem}\label{theo:reachable}
Let $g$ be a graph and $G=(N,S,P)$ an SL-HR grammar with $\val(G)=g$.
Given nodes $u,v\in V_g$, it can be determined in $O(|G|)$ time whether or not
$v$ is reachable from $u$ in $g$. 
\end{theorem}
\begin{proof}
We first compute $G$-representations $u'$, $v'$ of $u$ and $v$ in $O(\log l + h)$ time, as described in
Section~\ref{sss:neighborhood}.
We traverse $G$ bottom-up with respect to $\SLOrder$ in one pass
and compute for each nonterminal $A$ its skeleton graph $\sk(A)$.
%The set of nodes of $\sk(A)$ is given by the
%external nodes $\{v_1,\dots,v_n\}$ of the right-hand side $h$ of $A$.
%The edges are computed as follows.
%First, assume that $h$ is a terminal graph.
%We determine the strongly connected components in $h$
%in linear time (e.g., using Tarjan's algorithm~\cite{DBLP:journals/siamcomp/Tarjan72}).
%Let $h'$ be the corresponding graph which has
%as nodes the strongly connected components of $h$. 
%We remove from $h'$ each strongly connected component $C$ that does
%not contain external nodes.
%This is done by inserting for every pair of edges $e_1,e_2$
%such that $e_1$ is an edge from a component $D\not= C$ into $C$ 
%and $e_2$ is an edge from $C$ to a component $E\not =C$ (with $E\not= D$),
%an edge from component $D$ to component $E$. 
%Finally, we replace each component by a cycle of the external
%nodes of that component, and, for an edge from a component $D$ to
%a component $E$ we add an edge from an arbitrary external node of $D$
%to one of $E$.
After having computed in $O(|G|)$ time the skeleta for
all nonterminals, we can solve a reachability query as follows.
Let $S'$ be the graph obtained from $S$ by replacing each nonterminal
edge by its skeleton graph; clearly, it can be obtained in $O(|G|)$ time.

Case 1: Assume that $u'$ and $v'$ are of the form $u''$ and $v''$,
i.e., both nodes are in the start graph. It should be clear that
$v$ is reachable from $u$ in $\val(G)$ if and only if
$v''$ is reachable from $u''$ in $S'$. The latter is checked 
in $O(|S'|)$ time. 

Case 2: Let $u'=e_0\cdots e_nu''$ and $v'=f_0\cdots f_mv''$. 
Let $A_i$ be the label of $e_{i-1}$ for $i\in\{1,\dots,n\}$ and
let $B_j$ be the label of $f_{j-1}$ for $j\in\{1,\dots,m\}$.
We determine the set $E_n$ of external nodes of the right-hand side $h_n$ of $A_n$ that
are reachable from $u''$ in $h_n$. 
This is done by replacing the (at most two) nonterminal edges in $h_n$ by
their skeleton graphs, and then running a standard reachability test.
We now move up the derivation tree (viz. to the left in $u'$), at
each step computing a subset $E_i$ of the external nodes of $\sk(A_i)$:
we locate the nodes corresponding $E_{i+1}$ in  
$\sk(A_i)$ and determine the set $E_i$ of external nodes reachable from these.
Finally, we obtain a set $E_0$ of nodes in $S'$ (all incident with the edge $e_0$).
In a similar way we compute a set $F_0$ of nodes in $S$ that are incident
with $f_0$ (and can reach $v'$).
Finally, we check if a node in $F_0$ is reachable from a node in $E_0$.
This is done by adding edges over $F_0$ that form a cycle, and
edges over $E_0$ that form a cycle. 
We now pick arbitrary nodes $f$ in $F_0$ and $e$ in $E_0$
and check if $f$ is reachable from $e$ in $S'$.
\end{proof}
Furthermore, by Lengauer and Wagner~\cite{DBLP:journals/jcss/LengauerW92} the reachability problem on
grammar-compressed graphs is $\P$-complete.
\subsection{Regular Path Queries}\label{sss:rpq_queries}
Regular path queries (RPQ) are a well-known (see, e.g.,~\cite{Wood12_querySurvey}) way to query
graph data: the query is given as a regular expression $\alpha$ over the edge alphabet of the graph
$g$. Given two nodes $u,v$ in $g$ the problem is to decide whether there exists a path from $u$ to
$v$ such that its edge labels, taken as a string, match the regular expression $\alpha$. Such
queries are of relevance in modern applications: version 1.1 of the SPARQL
query-language\footnote{\url{https://www.w3.org/TR/sparql11-query/}}, for RDF data,
introduces \emph{property paths}, which are a variant of regular path queries. It is well-known that
the problem whether a regular path query holds for two given nodes in a graph is decidable in time
$O(|g|\cdot |\A_\alpha|)$ where $\A_\alpha$ is an NFA deciding the language $L(\alpha)$ of $\alpha$.
Such an automaton $\A_\alpha$ can be constructed in linear time and with $|\A_\alpha| \in
O(|\alpha|)$, e.g., using Thompson's construction~\cite{DBLP:journals/cacm/Thompson68}.
%, if it is allowed
%to include the same edge twice. It is $\mathsf{NP}$-complete if the edges used to get from $u$ to
%$v$ have to be pairwise distinct (see~\cite{Mendelzon95_RPQ}). Therefore it is common to define the
%semantics of RPQs to allow multiple uses of the same edge. 
The algorithm works as follows:
\begin{itemize}
    \item Consider $g$ to be an NFA $\A_g$ with initial state $u$ and final state $v$.
    \item Construct the product automaton of $\A_g$ and $\A_\alpha$ deciding
        $L(\A_g) \cap L(\alpha)$.
    \item Test the product automaton for emptiness.
\end{itemize}
If the final test is true, then there is no such path. Otherwise it exists and we say $u$ and $v$
\emph{satisfy} $\alpha$.
We generalize this method to grammar-compressed graphs. Given a grammar $G$
and a regular expression $\alpha$, we first compute an SL-HR grammar $G_\alpha$ which generates
the graph-structure (ignoring initial and final states for now) of the product automaton of
$\A_{\val(G)}$ and $\A_\alpha$. Then, using the result from the previous section, we decide
reachability from the initial to the final state. To make this construction easier to read, we will
relax one condition we enforced on hypergraphs before: nodes in the grammar are named not just by
IDs, but by a tuple of ID and state. In fact, for a node-ID $i$ and a set of states $Q$, we generate
nodes $(i,q)$ for every $q \in Q$. We refer to $i$ as the node's ID, and $q$ as the node's state.
Consequently, this affects how nodes are named during a derivation step. The renaming $\nr$ for
nodes $(i,q)$ always only affects the ID portion of the node, i.e., for any $q \in Q$ only renamings
of the form $\nr((i,q)) = (j,q)$ will be allowed. The rules by which the ID is renamed still follow
the method laid out in Section~\ref{sss:nodeIDs}. However, as there may now be distinct nodes
$(i,q)$ and $(i,p)$ for two different states $q,p \in Q$, we also require that $\nr((i,q)) = (j,q)$
and $\nr((i,p)) = (j,p)$, i.e., nodes with the same ID will keep this property. 
%
%Consequently, this affects the assignment of node-IDs to nodes that are generated during a
%derivation step. Let $A$ be a nonterminal of $G$ that has an internal node with ID $i$ in its
%right-hand side. When deriving such a nonterminal edge labeled $A$, we assign a new ID $j$ to the
%new node that is generated from the internal node with ID $i$. For a rule with internal nodes
%$(i,q)$ the newly generated nodes during the derivation still obtain the ID $j$, i.e., for every $q \in
%Q$ a node $(j,q)$ is generated.

We recall some necessary definitions. A \emph{nondeterministic finite automaton (NFA)} over an
alphabet $\Sigma$ is a tuple $\A = (Q, \delta, q_i, q_f)$, where $Q$ is a finite set of states, $\delta
\subseteq Q \times (\Sigma \cup \{\varepsilon\}) \times Q$ is the transition relation, $q_i \in Q$
is the initial state, and $q_f \in Q$ is the final state. A tuple $(q,w) \in Q \times \Sigma^*$ is
called a \emph{configuration} of $\A$. A configuration $(p,v)$ can follow a configuration $(q,w)$,
denoted $(q,w) \vdash (p,v)$, if
\begin{enumerate}
    \item $w = \sigma v$ for some $\sigma \in \Sigma$ and $(q,\sigma,p) \in \delta$, or
    \item $w = v$ and $(q,\varepsilon,p) \in \delta$.
\end{enumerate}
We denote the transitive closure of $\vdash$ by $\vdash^*$.
%We write $(q,w) \vdash^* (p,v)$ if there is a sequence of single steps that lead from $(q,w)$ to
%$(p,v)$.
The language decided by $\A$ is denoted as $L(\A) = \{w \mid (q_i, w) \vdash^* (q_f,
\varepsilon)\}$. 
%We define the size of an NFA as $|\A| = |Q| + |\delta|$.
Let $\A = (Q,\delta,q_i,
q_f)$ and $\A' = (Q', \delta', q_i', q_f')$ be two NFA's. By $\A \otimes \A' = (Q_\otimes,
\delta_\otimes)$ we denote the \emph{product construction} defined by $Q_\otimes = Q \times Q'$ and
$\delta_\otimes = \delta_\Sigma \cup \delta_\varepsilon$, where
\begin{equation*}
    \begin{array}{llll}
        \delta_\Sigma &= \{&((q,q'),\sigma,(p,p')) \mid& q,p \in Q, q',p' \in Q', \sigma \in \Sigma,\\
            &&&(q,\sigma,p) \in \delta,\text{ and }(q',\sigma,p') \in \delta'\},\text{ and} \\
        \delta_\varepsilon &= \{&((q,q'), \varepsilon, (p,p')) \mid& q,p \in Q, q', p' \in Q', \\
            &&&(q,\varepsilon,p) \in \delta \text{ and } (q', \varepsilon, p') \in \delta'\text{, or} \\
            &&&(q,\varepsilon,p) \in \delta \text{ and } q' = p'\text{, or} \\
            &&&q = p \text{ and } (q', \varepsilon, p') \in \delta'\}.\\
    \end{array}
\end{equation*}
Note that, while the product construction defines a system of state transitions, it does not set
initial and final states by default. However, the NFA $(Q_\otimes, \delta_\otimes, (q_i,q_i'),
(q_f,q_f'))$ does decide the language $L(\A) \cap L(\A')$. This property is used to decide whether
there exists a path satisfying a regular path query between two given nodes. Such automata without
initial/final states are also just called transition systems. The product construction still works
if either or both the automata are transition systems. A simple graph $g = (V, E, \att, \lbl, \ext)$
with edge labels from $\Sigma$ defines a transition system $\trans(g) = (Q_g, \delta_g)$ with $Q_g = V$
and $\delta_g = \{(v,\sigma,u) \mid \exists e \in E: \lbl(e) = \sigma$ and $\att(e) = vu\}$.

Consider the product construction of $\trans(g)$ and $\A$ for some simple graph $g$ and an NFA $\A$.
It represents every possible run of $\A$ on any path within $g$, only by setting an initial and
final state we choose specific start and end-nodes within $g$, and initial/final states of $\A$.
Thus, for any two nodes $u,v$ of $g$ we can check whether there is a path between them that $\A$
accepts with the same product construction. We next extend this notion to SL-HR grammars
representing simple graphs, by describing a construction which, given an SL-HR grammar $G$ and an
NFA $\A$, generates an SL-HR grammar representing the product construction of $\trans(\val(G))$ and
$\A$. Thus, we achieve a \emph{compressed representation} of all runs of $\A$ on paths of $\val(G)$.
\begin{lemma}
    Given an SL-HR grammar $G = (N,P,S)$ that generates a simple graph and an NFA $\A = (Q, \delta,
    q_i, q_f)$, we can compute in $O(|G||\A|)$ time an SL-HR grammar $G_{\A} = (N_\A, P_\A, S_\A)$,
    such that $\trans(\val(G_{\A})) = \trans(\val(G)) \otimes \A$.
    \label{lem:product_grammar}
\end{lemma}
\begin{proof}
    The construction is straightforward: for every rule $p=(A,g)$ of $G$, the product
    construction of $g$ and $\A$ is computed. Due to this, the rank increases: if $\rank(g) = k$,
    then $\rank(g \otimes \A) = k|Q|$. Accordingly, some care has to be taken with nonterminal
    edges, as these also increase in rank. They need to be attached to the same node-IDs as before,
    but once for every state in $Q$. To make sure that nodes attached to nonterminal edges and
    external nodes match up correctly, we enforce some (arbitrary) order on $Q$.

    Let $Q = \{q_1,\ldots, q_n\}$ and let $<_Q$ be a total order $q_1 <_Q q_2 <_Q \cdots <_Q q_n$ on
    $Q$. For a string $w = w_1\cdots w_m$ let $w \times Q$ be the string $(w_1, q_1)(w_2,
    q_1)\cdots(w_m, q_1)(w_1, q_2)\cdots(w_m, q_n)$. For every $(A, g) \in P$, $P_\A$ contains a
    rule $(A, g_\otimes)$ defined as follows. Let $V_\otimes = V \times Q$ and
    \begin{align*}
        E_\otimes =& \:\{e_{i,q,\sigma,p} \mid e_i \in E, (q,\sigma,p) \in \delta, \lbl(e_i) =
    \sigma\} \cup \\
        &\:\{e_{i,q,\varepsilon,p} \mid i \in V, (q,\varepsilon,p) \in \delta\} \cup \\
        &\:\{e_{i} \mid e_i \in E, \lbl(e_i) \in N\}
    \end{align*}
    such that 
    \begin{align*}
        \att_\otimes(e_{i,q,\sigma,p}) &= (u,q) \cdot (v,p) & (\text{with } \att(e_i) = uv) &&
        \lbl_\otimes(e_{i,q,\sigma,p}) &= \sigma \\
        \att_\otimes(e_{i,q,\varepsilon,p}) &= (i,q) \cdot (i,p) &&&
        \lbl_\otimes(e_{i,q,\varepsilon,p}) &= \varepsilon \\
        \att_\otimes(e_i) &= \att(e_i) \times Q &&&
        \lbl_\otimes(e_i) &= \lbl(e_i).
    \end{align*}
%    \begin{align*}
%        \att_\otimes(e_{i,q,\sigma,p}) &= (u,q) \cdot (v,p) & \text{with } \att(e_i) = uv \\
%        \att_\otimes(e_{i,q,\varepsilon,p}) &= (i,q) \cdot (i,p) \\
%        \att_\otimes(e_i) &= \att(e_i) \times Q,
%    \end{align*}
%    and 
%    \begin{align*}
%        \lbl_\otimes(e_{i,q,\sigma,p}) &= \sigma \\
%        \lbl_\otimes(e_{i,q,\varepsilon,p}) &= \varepsilon \\
%        \lbl_\otimes(e_i) &= \lbl(e_i).
%    \end{align*}
    For the external nodes we set $\ext_\otimes = \ext \times Q$. This defines the rules of $G_\A$.
    Clearly, we do not add any nonterminals, and thus have $N_\A = N$. Finally $S_\A$ is constructed
    from $S$ and $\A$ using the same construction as described above for the right-hand sides of the
    rules (i.e., $S_\A = S_\otimes$).
\end{proof}
It should be clear that $|G_\A| \in O(|G||\A|)$. It can now be used to decide regular path queries:
\begin{theorem}
    Given an SL-HR grammar $G$, where $\val(G)$ is a simple graph, an RPQ $\alpha$, and a pair of
    nodes $u,v$ from $\val(G)$, it can be decided in time $O(|\alpha||G|)$ whether $u$ and $v$ 
    satisfy $\alpha$.
    \label{thm:rpq_query}
\end{theorem}
\begin{proof}
    First compute an NFA $\A$ from $\alpha$ by using Thompson's well known
    construction~\cite{DBLP:journals/cacm/Thompson68}. Note that $|\A| \in O(|\alpha|)$. Using
    Lemma~\ref{lem:product_grammar} compute the grammar $G_{\A}$ generating the product construction of
    $G$ and $\A$. Now, using Theorem~\ref{theo:reachable}, decide reachability from $(u,q_i)$ to
    $(v,q_f)$ in $\val(G_\A)$. Note that $u,v$ are node-IDs in $\val(G)$, but a $G$-representation
    $e_1\cdots e_n i$ of $v$ can be converted into a $G_{\A}$-representation of $(v,q_i)$ by using
    $e_1\cdots e_n (i,q_i)$. If a path exists, we can conclude that there is a path from $u$ to
    $v$ satisfying the RPQ $\alpha$. Otherwise, it does not exist.
\end{proof}
As an example showing how powerful this construction is, consider the string $a^{40}$, which can be
represented by a grammar as $a^{32}a^8$ and the obvious SL-HR grammar $G$ representing the string-graph
encoding this string. The minimal NFA $\A$ deciding $L = \{a^x \mid x \bmod 5 = 0 \}$ has 5 states. The
product construction of $\sgraph(a^{40})$ and $\A$ would therefore have 205 states, whereas the grammar
$G_\A$ only has 95 states, and just replacing the nonterminals with their skeleta in the small start
graph (20 states), e.g., when testing for reachability between $(0,q_i)$ and $(40,q_f)$, immediately
shows that there is an accepting run from the first to the last node (i.e., nodes $0$ and $40$,
respectively) of the graph.

We can also use this ``product grammar'' to decide a more general problem. In a way, the product
construction encodes every pair of nodes that fulfills the query $\alpha$. Thus, the product grammar
$G_\A$ is a compressed representation of every such pair. We can therefore decide, whether there
exists a pair of nodes at all, that satisfies the query. Intuitively, this is done by computing,
bottom-up, for every rule, which external nodes can be reached from any initial state, and from
which external nodes a final state can be reached. If an external node appears in both these lists
at some point, we know that there is a path between some initial and final state. 
\begin{theorem}
    Given an SL-HR grammar $G$, where $\val(G)$ is a simple graph, and an RPQ $\alpha$, it can be
    decided in time $O(|\alpha||G|)$, if there exist of nodes $u,v$ in $\val(G)$ satisfying
    $\alpha$.
    \label{thm:rpq_exists}
\end{theorem}
\begin{proof}
    As before, construct $\A$ from $\alpha$, and $G_\A$ from $G$ and $\A$ according to
    Lemma~\ref{lem:product_grammar}. We now compute some helpful information in a bottom-up pass
    over the grammar, i.e., iterating over the rules in the reverse $\SLOrder$-order. For a rule
    $(A, g)$ we first compute whether any node $(y,q_f)$ can be reached from any node $(x,q_i)$ (where
    $(x,q_i)$ and $(y,q_f)$ are nodes in $g$) within $\val(g)$, using the methods from
    Theorem~\ref{theo:reachable}. If this is the case, a pair satisfying $\alpha$ exists. Otherwise,
    we compute two sets of nodes $P_A$ and $Q_A$. Intuitively, $P_A$ contains the external 
    nodes of $g$ that can be reached from some $(v,q_i)$ (an initial state) while $Q_A$ contains the
    external nodes of $g$ from which a node $(v,q_f)$ (a final state) is reachable.

    $P_A$ and $Q_A$ for a rule $(A,g)$ are computed in the following way: first we compute two sets
    of nodes, $X$ and $Y$. The set $X$ contains every node $x \in V_g$ with the following
    properties: 
    \begin{itemize}
        \item there exists a nonterminal edge $e \in E_g$ such that $x \in \att(e)$, and
        \item if $x$ is on position $i$ in $\att(e)$, then the $i$th external node of
            $\rhs(\lbl(e))$ is in $P_{\lbl(e)}$.
    \end{itemize}
    The set $Y$ is defined analogously, but using $Q_{\lbl(e)}$ instead of $P_{\lbl(e)}$.
    Now we replace every nonterminal edge of $g$ by their skeleton (see
    Definition~\ref{def:skeleton}) and let this graph be $g'$. To $g'$ we add a new node $s$ and %add
    an edge from $s$ to every node in $\{(x,q_i) \in V_g \mid x \in \N\} \cup X$. Every external node
    now reachable from $s$ belongs to the set $P_A$. This set can be computed by a single
    BFS-traversal starting in $s$. To compute $Q_A$ we add a new node $t$ to $g'$ that has an edge
    \emph{from} every node in $\{(x,q_f) \in V_g \mid x \in \N\} \cup Y$. Every external node that
    can reach $t$ is part of $Q_A$. This can be computed by complementing all the edges and again
    starting a BFS traversal from $t$.

    If, for any $A \in N$, $P_A \cap Q_A \neq \emptyset$, then a pair of nodes $u,v$ exists within
    $\val(A)$ such that they satisfy $\alpha$. Otherwise, we compute $X$ and $Y$ for the start graph
    $S_\A$, replace the nonterminals in $S_\A$ by their skeleta, and add both $s$ and $t$ to $S_\A$
    as above. If $t$ is now reachable from $s$, then there exist nodes $u,v$ within $\val(G)$
    satisfying $\alpha$. Otherwise no such nodes exist.
\end{proof}
Note that the constructions in this section all generalize to hypergraphs (instead of simple graphs)
in a straightforward way. There is some ambiguity on how to define transition systems using
hyperedges. We suggest to use a similar approach as for the definition of paths within hypergraphs
(cf. Section~\ref{sse:preliminaries}). A hyperedge in a transition system would thus be considered
as directed with one source node (the first one it is attached to) and multiple target nodes.
Semantically, a rank $k$ ($k > 2$) hyperedge labeled $\sigma$ and attached to nodes $v_1\cdots v_k$
would then be the same as $k-1$ simple edges all labeled $\sigma$, starting in $v_1$, and using
$v_2,\ldots, v_k$ as target nodes.

%% file: conclusion.tex
\section{Conclusions}

We present a generalization to (hyper)graphs 
of the RePair compression scheme as known for strings and trees.
Our generalization produces from a given graph a straight-line
hyperedge replacement grammar (SL-HR grammar).
We prove some theoretical results about SL-HR grammars, for instance,
if the given graph is a string or a tree, then an SL-HR grammar
cannot compress much better than an ordinary string or tree grammar;
thus, in these cases the use of graph grammars does not
offer stronger compression than the existing native grammars.
In terms of the RePair compression, 
we prove that (as for trees) the choice of the maximal
rank of a grammar can heavily influence the compression behavior. 

We then study an implementation of RePair. We observe that for graphs,
finding a digram with a maximal number of non-overlapping occurrences
is computationally hard; using state-of-the-art algorithms it requires
at least cubic time. We therefore introduce an approximation which
counts greedily and heavily depends on the order in which the graph is traversed.
We experiment with several node orders and find that an order that generalizes
the node degree order and is similar to the one used in the Weisfeiler-Lehman approximative
isomorphism test~\cite{WeisfeilerLehman68} gives the best compression results.
We compare our compressor to state-of-the-art graph compressors.
Over network graphs (which have no edge labels),
we do not obtain a conclusive answer: sometimes our
compressor gives strongest compression, sometimes the other compressors do. 
Over RDF graphs (which are edge-labeled) our compressor gives the best
results, sometimes factors of several magnitudes smaller than other
compressors. We also obtain the best results for ``version graphs'' which
are disjoint unions of versions of the same graph (which are very similar).  

We prove that, as in the case of strings and trees, there exist interesting
``speed-up'' algorithms. Such algorithms can run faster on the compressed
grammar than on the original, by a factor proportional to the compression
ration of the grammar. We show that reachability between two given nodes
offers such an algorithm, as well as evaluating regular path queries
over two nodes. The latter asks if there exists a path so that the string
of path labels matches a given regular expression. We show that even
without a candidate pair of nodes given, we can determine within similar
time bounds, whether or not there exists any pair of nodes in the original
graph with a path between them matching the regular expression.

For future work, there are several paths to follow. It would be interesting
to consider a RePair compression scheme for graphs that is based on
node replacement (NR) graph grammars. NR graph grammars can compress some
graphs, e.g. cliques, much stronger than HR grammars. On the other hand,
rules of NR grammars are more complex and expensive to store.
For our compressor, other node orders should be considered which can give
rise to stronger compression. Especially for version graphs we believe that
different orders could give better results; one possibility, for instance,
is to compute the edit distance between two versions of a graph, and to
use it to compute a node order that is beneficial for our RePair compressor.
In terms of speed-up algorithms there is much work to be done. Foremost,
we would like to implement our algorithms and show that they can be faster
than existing state-of-the art graph databases.
Fundamentally of utmost interest and importance is the study of the complexity
of the isomorphism problem for SL-HR grammars. Can it be decided in polynomial
time if two graphs represented by SL-HR grammars are isomorphic? 
For strings and trees similar results hold (see~\cite{Lohrey12_stringSLPsurvey} 
and~\cite{DBLP:journals/is/BusattoLM08}), but note that even extending
the latter result from (ordered ranked) trees to 
arbitrary trees is non-trivial~\cite{DBLP:conf/icalp/LohreyMP15} and that the complexity of deciding
the isomorphism of two explicitly given graphs is currently not known to be polynomial.

%% file: main.bbl
\begin{thebibliography}{10}

\bibitem{DBLP:journals/kais/Alvarez-GarciaB15}
S.~{\'{A}}lvarez{-}Garc{\'{\i}}a, N.~R. Brisaboa, J.~D. Fern{\'{a}}ndez, M.~A.
  Mart{\'{\i}}nez{-}Prieto, and G.~Navarro.
\newblock Compressed vertical partitioning for efficient {RDF} management.
\newblock {\em Knowl. Inf. Syst.}, 44(2):439--474, 2015.

\bibitem{Apostolico09_graphCompressionBFS}
A.~Apostolico and G.~Drovandi.
\newblock Graph {C}ompression by {BFS}.
\newblock {\em Algorithms}, 2(3):1031--1044, 2009.

\bibitem{Boldi09_permutingWebGraphs}
P.~Boldi, M.~Santini, and S.~Vigna.
\newblock Permuting web graphs.
\newblock In {\em Algorithms and Models for the Web-Graph}, pages 116--126.
  2009.

\bibitem{DBLP:conf/www/BoldiV04}
P.~Boldi and S.~Vigna.
\newblock The webgraph framework {I:} compression techniques.
\newblock In {\em {WWW}}, pages 595--602, 2004.

\bibitem{DBLP:journals/is/BrisaboaLN14}
N.~R. Brisaboa, S.~Ladra, and G.~Navarro.
\newblock Compact representation of web graphs with extended functionality.
\newblock {\em Inf. Syst.}, 39:152--174, 2014.

\bibitem{Buehrer08_scalableWebGraphCompression}
G.~Buehrer and K.~Chellapilla.
\newblock A {S}calable {P}attern {M}ining {A}pproach to {W}eb {G}raph
  {C}ompression with {C}ommunities.
\newblock In {\em WSDM}, pages 95--106, 2008.

\bibitem{DBLP:journals/is/BusattoLM08}
G.~Busatto, M.~Lohrey, and S.~Maneth.
\newblock Efficient memory representation of {XML} document trees.
\newblock {\em Inf. Syst.}, 33(4-5):456--474, 2008.

\bibitem{DBLP:journals/combinatorica/CaiFI92}
J.~Cai, M.~F{\"{u}}rer, and N.~Immerman.
\newblock An optimal lower bound on the number of variables for graph
  identifications.
\newblock {\em Combinatorica}, 12(4):389--410, 1992.

\bibitem{Charikar05_smallest}
M.~Charikar, E.~Lehman, D.~Liu, R.~Panigrahy, M.~Prabhakaran, A.~Sahai, and
  A.~Shelat.
\newblock The smallest grammar problem.
\newblock {\em IEEE Transactions on Information Theory}, 51(7):2554--2576,
  2005.

\bibitem{Claude10_compactWebGraphRep}
F.~Claude and G.~Navarro.
\newblock Fast and {C}ompact {W}eb {G}raph {R}epresentations.
\newblock {\em {TWEB}}, 4(4):16:1--16:31, 2010.

\bibitem{DBLP:journals/tcs/CourcelleM93}
B.~Courcelle and M.~Mosbah.
\newblock Monadic {S}econd-{O}rder {E}valuations on {T}ree-{D}ecomposable
  {G}raphs.
\newblock {\em Theor. Comput. Sci.}, 109(1{\&}2):49--82, 1993.

\bibitem{DBLP:conf/gg/DrewesKH97}
F.~Drewes, H.{-}J. Kreowski, and A.~Habel.
\newblock Hyperedge {R}eplacement {G}raph {G}rammars.
\newblock In {\em Handbook of Graph Grammars and Computing by Graph
  Transformations, Volume 1: Foundations}, pages 95--162, 1997.

\bibitem{blossom}
J.~Edmonds.
\newblock Paths, trees, and flowers.
\newblock {\em Canad. J. Math.}, 17:449--467, 1965.

\bibitem{DBLP:journals/tit/Elias75}
P.~Elias.
\newblock Universal codeword sets and representations of the integers.
\newblock {\em {IEEE} Transactions on Information Theory}, 21(2):194--203,
  1975.

\bibitem{DBLP:journals/acta/EngelfrietH92}
J.~Engelfriet and L.~Heyker.
\newblock {C}ontext-{F}ree {H}ypergraph {G}rammars have the {S}ame
  {T}erm-{G}enerating {P}ower as {A}ttribute {G}rammars.
\newblock {\em Acta Inf.}, 29(2):161--210, 1992.

\bibitem{DBLP:conf/gg/EngelfrietR97}
J.~Engelfriet and G.~Rozenberg.
\newblock Node {R}eplacement {G}raph {G}rammars.
\newblock In {\em Handbook of Graph Grammars and Computing by Graph
  Transformations, Volume 1: Foundations}, pages 1--94, 1997.

\bibitem{Engelfriet:1997:CGG:267871.267874}
Joost Engelfriet.
\newblock Context-free graph grammars.
\newblock In Grzegorz Rozenberg and Arto Salomaa, editors, {\em Handbook of
  Formal Languages, Vol. 3}, pages 125--213. 1997.

\bibitem{DBLP:conf/icdt/Fan12}
W.~Fan.
\newblock Graph pattern matching revised for social network analysis.
\newblock In {\em {ICDT}}, pages 8--21, 2012.

\bibitem{DBLP:conf/www/FernandezGM10}
J.~D. Fern{\'{a}}ndez, C.~Gutierrez, and M.~A. Mart{\'{\i}}nez{-}Prieto.
\newblock {RDF} compression: basic approaches.
\newblock In {\em {WWW}}, pages 1091--1092, 2010.

\bibitem{DBLP:conf/semweb/FernandezMG10}
J.~D. Fern{\'{a}}ndez, M.~A. Mart{\'{\i}}nez{-}Prieto, and C.~Gutierrez.
\newblock Compact representation of large {RDF} data sets for publishing and
  exchange.
\newblock In {\em {ISWC}}, pages 193--208, 2010.

\bibitem{DBLP:journals/iandc/GalperinW83}
H.~Galperin and A.~Wigderson.
\newblock Succinct representations of graphs.
\newblock {\em Information and Control}, 56(3):183--198, 1983.

\bibitem{DBLP:conf/dcc/GasieniecKPS05}
L.~Gasieniec, R.~M. Kolpakov, I.~Potapov, and P.~Sant.
\newblock Real-time traversal in grammar-based compressed files.
\newblock In {\em {DCC}}, page 458, 2005.

\bibitem{DBLP:conf/vldb/GoldmanW97}
R.~Goldman and J.~Widom.
\newblock Data{G}uides: {E}nabling {Q}uery {F}ormulation and {O}ptimization in
  {S}emistructured {D}atabases.
\newblock In {\em VLDB}, pages 436--445, 1997.

\bibitem{DBLP:journals/dam/GrabowskiB14}
S.~Grabowski and W.~Bieniecki.
\newblock Tight and simple web graph compression for forward and reverse
  neighbor queries.
\newblock {\em Discrete Applied Mathematics}, 163:298--306, 2014.

\bibitem{DBLP:books/sp/Habel92}
A.~Habel.
\newblock {\em Hyperedge {R}eplacement: {G}rammars and {L}anguages}.
\newblock Lecture Notes in Computer Science. 1992.

\bibitem{DBLP:journals/tcs/HabelKL93}
A.~Habel, H.~Kreowski, and C.~Lautemann.
\newblock A {C}omparison of {C}ompatible, {F}inite, and {I}nductive {G}raph
  {P}roperties.
\newblock {\em Theor. Comput. Sci.}, 110(1):145--168, 1993.

\bibitem{DBLP:journals/kais/HernandezN14}
C.~Hern{\'{a}}ndez and G.~Navarro.
\newblock Compressed representations for web and social graphs.
\newblock {\em Knowl. Inf. Syst.}, 40(2):279--313, 2014.

\bibitem{Larsson00_rePair}
N.~J. Larsson and A.~Moffat.
\newblock Off-line dictionary-based compression.
\newblock {\em Proceedings of the IEEE}, pages 1722--1732, 2000.

\bibitem{DBLP:journals/jcss/LengauerW92}
T.~Lengauer and K.~W. Wagner.
\newblock The correlation between the complexities of the nonhierarchical and
  hierarchical versions of graph problems.
\newblock {\em JCSS}, 44(1):63--93, 1992.

\bibitem{DBLP:journals/siamcomp/LengauerW88}
T.~Lengauer and E.~Wanke.
\newblock Efficient solution of connectivity problems on hierarchically defined
  graphs.
\newblock {\em {SIAM} J. Comput.}, 17(6):1063--1080, 1988.

\bibitem{DBLP:conf/sigmod/LiefkeS00}
H.~Liefke and D.~Suciu.
\newblock {XMILL:} an efficient compressor for {XML} data.
\newblock In {\em SIGMOD}, pages 153--164, 2000.

\bibitem{Lohrey12_stringSLPsurvey}
M.~Lohrey.
\newblock Algorithmics on {SLP}-compressed strings: {A} survey.
\newblock {\em Groups Complexity Cryptology}, 4(2):241--299, 2012.

\bibitem{DBLP:journals/jcss/Lohrey12}
M.~Lohrey.
\newblock Model-checking hierarchical structures.
\newblock {\em JCSS}, 78(2):461--490, 2012.

\bibitem{DBLP:conf/dlt/Lohrey15}
M.~Lohrey.
\newblock Grammar-{B}ased {T}ree {C}ompression.
\newblock In {\em {DLT}}, pages 46--57, 2015.

\bibitem{Lohrey06_SLPTreeComplexity}
M.~Lohrey and S.~Maneth.
\newblock The complexity of tree automata and {XPath} on grammar-compressed
  trees.
\newblock {\em Theor. Comp. Sci.}, 363(2):196--210, 2006.

\bibitem{Lohrey13_treeRePair}
M.~Lohrey, S.~Maneth, and R.~Mennicke.
\newblock {XML} tree structure compression using {R}e{P}air.
\newblock {\em Information Systems}, 38(8):1150--1167, 2013.

\bibitem{DBLP:journals/jcss/LohreyMS12}
M.~Lohrey, S.~Maneth, and M.~Schmidt{-}Schau{\ss}.
\newblock Parameter reduction and automata evaluation for grammar-compressed
  trees.
\newblock {\em J. Comput. Syst. Sci.}, 78(5):1651--1669, 2012.

\bibitem{DBLP:conf/icalp/LohreyMP15}
Markus Lohrey, Sebastian Maneth, and Fabian Peternek.
\newblock Compressed tree canonization.
\newblock In {\em {ICALP}}, pages 337--349, 2015.

\bibitem{lmr16}
Markus Lohrey, Sebastian Maneth, and Carl~Philip Reh.
\newblock Traversing grammar-compressed trees with constant delay.
\newblock In {\em DCC}, 2016.

\bibitem{DBLP:journals/corr/abs-1012-5696}
S.~Maneth and T.~Sebastian.
\newblock Fast and tiny structural self-indexes for {XML}.
\newblock {\em CoRR}, abs/1012.5696, 2010.

\bibitem{DBLP:conf/icde/ManethP16}
Sebastian Maneth and Fabian Peternek.
\newblock Compressing graphs by grammars.
\newblock In {\em {ICDE}}, pages 109--120, 2016.

\bibitem{DBLP:journals/njc/MaratheHR94}
M.~V. Marathe, H.~B. {Hunt III}, and S.~S. Ravi.
\newblock The {C}omplexity of {A}pproximation {PSPACE}-{C}omplete {P}roblems
  for {H}ierarchical {S}pecifications.
\newblock {\em Nord. J. Comput.}, 1(3):275--316, 1994.

\bibitem{DBLP:journals/siamcomp/MaratheHSR98}
M.~V. Marathe, H.~B. {Hunt III}, R.~E. Stearns, and V.~Radhakrishnan.
\newblock Approximation {A}lgorithms for pspace-{H}ard {H}ierarchically and
  {P}eriodically {S}pecified {P}roblems.
\newblock {\em {SIAM} J. Comput.}, 27(5):1237--1261, 1998.

\bibitem{DBLP:journals/tcs/MaratheRHR97}
M.~V. Marathe, V.~Radhakrishnan, H.~B. {Hunt III}, and S.~S. Ravi.
\newblock Hierarchically {S}pecified {U}nit {D}isk {G}raphs.
\newblock {\em Theor. Comput. Sci.}, 174(1-2):23--65, 1997.

\bibitem{DBLP:conf/sac/Martinez-PrietoFC12}
M.~A. Mart{\'{\i}}nez{-}Prieto, J.~D. Fern{\'{a}}ndez, and R.~C{\'{a}}novas.
\newblock Compression of {RDF} dictionaries.
\newblock In {\em {SAC}}, pages 340--347, 2012.

\bibitem{DBLP:journals/jair/Nevill-ManningW97}
C.~G. Nevill{-}Manning and I.~H. Witten.
\newblock Identifying {H}ierarchical {S}tructure in {S}equences: {A}
  {L}inear-{T}ime {A}lgorithm.
\newblock {\em JAIR}, 7:67--82, 1997.

\bibitem{DBLP:journals/iandc/PapadimitriouY86}
C.~H. Papadimitriou and M.~Yannakakis.
\newblock A {N}ote on {S}uccinct {R}epresentations of {G}raphs.
\newblock {\em Information and Control}, 71(3):181--185, 1986.

\bibitem{Peshkin07_graphSequitour}
L.~Peshkin.
\newblock Structure induction by lossless graph compression.
\newblock In {\em {DCC}}, pages 53--62, 2007.

\bibitem{DBLP:conf/slate/SwachaG15}
J.~Swacha and S.~Grabowski.
\newblock {OFR:} {A}n {E}fficient {R}epresentation of {RDF} {D}atasets.
\newblock In {\em {SLATE}}, pages 224--235, 2015.

\bibitem{DBLP:journals/siamcomp/Tarjan72}
R.~E. Tarjan.
\newblock {D}epth-{F}irst {S}earch and {L}inear {G}raph {A}lgorithms.
\newblock {\em {SIAM} J. Comput.}, 1(2):146--160, 1972.

\bibitem{DBLP:journals/cacm/Thompson68}
K.~Thompson.
\newblock Regular {E}xpression {S}earch {A}lgorithm.
\newblock {\em Commun. {ACM}}, 11(6):419--422, 1968.

\bibitem{DBLP:journals/concurrency/UrbaniMDSB13}
J.~Urbani, J.~Maassen, N.~Drost, F.~J. Seinstra, and H.~E. Bal.
\newblock Scalable {RDF} data compression with {M}ap{R}educe.
\newblock {\em Concurrency and Computation: Practice and Experience},
  25(1):24--39, 2013.

\bibitem{WeisfeilerLehman68}
B.~Weisfeiler and A.~A. Lehman.
\newblock A reduction of a graph to a canonical form and an algebra arising
  during this reduction.
\newblock {\em Nauchno-Technicheskaya Informatsia}, Seriya 2(9):12--16, 1968.
\newblock (in Russian).

\bibitem{Wood12_querySurvey}
P.~T. Wood.
\newblock {Q}uery {L}anguages for {G}raph {D}atabases.
\newblock {\em SIGMOD Record}, pages 50--60, 2012.

\bibitem{DBLP:conf/synasc/Simecek09}
I.~Šimeček.
\newblock Sparse {M}atrix {C}omputations {U}sing the {Q}uadtree {S}torage
  {F}ormat.
\newblock In {\em SYNASC}, pages 168--173, 2009.

\end{thebibliography}
